\documentclass[12pt]{article}

\usepackage{amsmath}
\usepackage{amssymb}
\usepackage{mathtools}
\usepackage{ifthen}
\usepackage{amsthm}
\usepackage{graphicx}
\usepackage{epstopdf}
\usepackage{float}
\usepackage{placeins}

\newtheorem{lemma}{Lemma}
\newtheorem{theorem}{Theorem}

\newtheorem{proposition}{Proposition}

\newcommand{\ds}{\displaystyle}
\newcommand{\J}{\ensuremath{\mathcal{J}}}

\newcommand\T{\ensuremath{\mathcal{T}}}
\newcommand{\Z}{\ensuremath{\mathbb Z}}
\newcommand\D{\ensuremath{\mathcal{D}}}
\newcommand\F{\ensuremath{\mathcal{F}}}

\newcommand\R{\ensuremath{\mathbb{R}}}

\newcommand{\muinf}{\ensuremath{\mu_{\inf}}}

\newcommand{\by}{\boldsymbol{y}}
\newcommand{\bp}{\boldsymbol{p}}
\newcommand{\bpi}{\boldsymbol{\pi}}
\newcommand{\bn}{\boldsymbol{n}}
\newcommand{\be}{\boldsymbol{e}}
\newcommand{\bh}{\boldsymbol{h}}

\newcommand{\bxi}{\boldsymbol{\xi}}
\newcommand{\bz}{\boldsymbol{z}}
\newcommand{\bX}{\boldsymbol{X}}
\newcommand{\bmu}{\boldsymbol{\mu}}

\newcommand{\bphi}{\boldsymbol{\phi}}

\newcommand{\myprime}{\,\prime}
\newcommand{\qit}{\text{qit}}
\newcommand{\Id}{\operatorname{Id}}

\newcommand{\norm}[1]{\left\lVert #1 \right\rVert}
\newcommand{\BS}{\ensuremath{\ell_1(\Z)}}

\newcommand{\bzero}{\boldsymbol{0}}
\newcommand{\ts}{\textstyle}
\newcommand{\bx}{\boldsymbol{x}}
\newcommand{\mw}{m_{\text{w}}}
\newcommand{\vw}{m_{2,\text{w}}}
\newcommand{\sigmasquaredw}{\sigma^2_{\text{w}}}

\title{Mutation--selection balance on an infinite trait space: 
confinement, drift and equilibrium}

\author{Phil Pollett}

\begin{document}
\maketitle

\begin{center}
\begin{minipage}{14cm}
{\bf Abstract} \ \ 
\setlength{\parindent}{1.5em}
We study a trait-structured population model incorporating mutation,
selection and density-dependent regulation on a countably infinite trait
space. The underlying stochastic process is a continuous-time
Markov chain in which individuals reproduce at a
trait-independent rate, offspring traits are determined by a mutation
kernel on $\Z$, mortality depends on trait, and births are
progressively suppressed as the population approaches a fixed population
ceiling. Using results for density-dependent Markov population processes
with countably many types, we derive a deterministic approximation in
the form of an infinite system of nonlinear differential equations.

We establish existence and positive invariance of solutions, and
investigate the equilibrium structure of the deterministic system. A
fundamental distinction emerges between bounded and confining mortality
profiles. When mortality remains bounded, mutation may continually
transport mass through the trait space and a stationary trait distribution
need not exist. In contrast, when mortality increases without bound as the
absolute value of the trait index becomes large, the operator
$L=D^{-1}P$, where $P$ and $D$ govern mutation and mortality, is
compact. By combining compactness with Kre\u{\i}n-Rutman theory for
compact positive operators on Banach lattices, we show that $L$ has an
algebraically simple principal eigenvalue with a strictly positive
eigenvector, and we derive a threshold condition for the existence of a
non-zero equilibrium. In this regime the equilibrium is unique, and its
trait distribution is determined by the principal eigenvector of $L$.
Numerical experiments support the theoretical results and illustrate the
contrasting behaviours associated with bounded and confining mortality
profiles.
\end{minipage}
\end{center}



\section{Introduction}
\label{sec:intro}

Mutation generates heritable variation, while natural selection filters
that variation through differential survival and reproduction.
Understanding the balance between these forces is a central problem in
population genetics, evolutionary biology, and mathematical ecology
\cite{Burger2003,CrowKimura1970,Ewens2004,TravisEtAl2023}. Classical
mutation--selection theory seeks to determine how mutation and selection
shape the long-term distribution of traits within a population, and
whether stable equilibria emerge. The subject has
its origins in the work of Fisher, Haldane and Wright; modern
mathematical accounts include Crow and Kimura~\cite{CrowKimura1970},
Ewens~\cite{Ewens2004}, and B\"urger~\cite{Burger2003}.

Most mutation--selection models are formulated on finite trait spaces,
where matrix analysis, especially the Perron--Frobenius theorem,
provides a natural framework for studying equilibrium trait
distributions and identifying persistence thresholds. On an infinite
trait space, however, mutation may continually transport mass towards
increasingly distant traits, so that a stationary trait distribution
might not exist. Whether selection can counteract this dispersal
depends not only on its overall strength, but also on its behaviour in
the tails of the trait space.

We investigate this question in a population model combining mutation,
selection and demographic regulation. Individuals possess traits
in~$\Z$, although the operator-theoretic arguments in particular regimes
extend to more general countable trait spaces. Individuals reproduce at
a trait-independent rate~$\lambda$, offspring traits are determined by a
mutation kernel $(p_i,\,i\in\Z)$, and selection acts through
trait-dependent per-capita death rates $(\mu_i,\, i\in \Z)$. Population
growth is regulated by there being a global population ceiling~$N$, with
births progressively suppressed as the total population approaches~$N$.
The resulting stochastic process is a continuous-time Markov chain with
an absorbing state corresponding to population extinction.

The present work also connects with a broader literature on the
interaction between ecological and evolutionary processes. While
classical mutation--selection theory has focused primarily on the
maintenance of variation, mutation load, and equilibrium trait
frequencies \cite{Burger2003,CrowKimura1970,Ewens2004}, increasing
attention has been devoted to systems in which evolutionary and
demographic processes operate on comparable time scales. Examples
include density-dependent selection
\cite{Christiansen,TravisEtAl2023}, stochastic population-genetic models
with variable population size~\cite{Parsons2010}, adaptive-dynamics
approaches linking individual-level stochastic models to macroscopic
evolutionary behaviour~\cite{ChampagnatFerriereMeleard}, and studies of
density-dependent evolution in fluctuating
environments~\cite{Lande2009}. Related questions arise in
structured-population theory, where demographic, physiological, spatial,
or phenotypic heterogeneity plays a central role in population
dynamics~\cite{Diekmann2021,Webb1985}.

Within the context of this literature, the present model places
particular emphasis on the combination of two features. First, mutation,
selection and density regulation are represented within a single
stochastic population model, and are retained explicitly in an
approximating deterministic model. Second, the trait space is countably
infinite, allowing mutation to generate continual movement through the trait
space. This separates two questions that are closely connected in
finite-dimensional models: whether the population persists
demographically, and whether its trait distribution remains localized.

From a stochastic-process perspective, the model belongs to the class
of density-dependent population processes studied by
Kurtz~\cite{Kur70,Kur71}. Those seminal papers established laws of large
numbers under which suitably rescaled Markov population processes
converge, over finite time intervals, to solutions of deterministic
differential equations. Because the trait space in our model is
countably infinite, the original framework is not directly applicable.
We instead appeal to the results of Barbour and
Luczak~\cite{BL12b}, who developed an approximation theory for
population processes with countably many types and obtained explicit
error bounds.

This leads to a principal theme of the paper: demographic regulation
alone need not confine the population in the trait space. When mortality
is bounded, mutation may continually disperse mass towards distant
traits, preventing the formation of a stationary trait distribution even
when the aggregate population persists. In contrast, under confining
mortality, extreme traits become increasingly disadvantageous. This
distinction between unconfined and confined evolutionary dynamics is
expressed mathematically through the positive operator $L=D^{-1}P$,
where $P$ governs mutation and $D=\operatorname{diag}(\mu_i)$ governs
mortality. In the confining case $L$ is compact. If, additionally,
mortality is bounded away from $0$ and the mutation kernel is
irreducible, a version of the Kre\u{\i}n--Rutman theorem yields an
algebraically simple principal eigenvalue, and a strictly positive
eigenvector that is unique up to scalar multiples. We show that the
spectral radius $r(L)$ determines a threshold analogous to the basic
reproduction number of epidemic theory: a positive equilibrium exists if
and only if $\lambda r(L)>1$. Whenever this condition holds, the
equilibrium is unique and its trait distribution is determined by the
principal eigenvector of~$L$. Thus, persistence of a localized trait
distribution depends not merely on population regulation, but on the
ability of selection to confine mutation-driven dispersal.

The paper is organised as follows. Section~\ref{sec:model} introduces
the stochastic model, derives its deterministic approximation, and
establishes existence, uniqueness and positive invariance of its
solutions. Sections~\ref{sec:noselection} and~\ref{sec:nomutation}
analyse two special cases that isolate the effects of mutation and
selection. Section~\ref{sec:generalcase} develops the spectral theory of
the operator~$L$ and establishes the equilibrium threshold result.
Numerical examples illustrate directional, stabilizing and disruptive
selection. Section~\ref{sec:discussion} discusses the biological
implications of the results and several directions for future work.
The longer proofs are collected in the Appendix, in the order in which 
the corresponding results appear in the text.

\section{The model}
\label{sec:model}

Each individual in the population has one of a set $\Z$ of traits. Let
$n_i(t)$ be the number of individuals with trait $i\in \Z$ at time
$t\geq 0$. Individuals with
trait~$k$ give birth at per-capita rate~$\lambda$, the offspring having
trait $k+l$ (genetic distance~$l$) with probability $p_l$, $l\in \Z$
($\sum_l p_l=1$, $p_0\leq 1$), and die at per-capita rate $\mu_k$. We
model selection only through the death rates. For example, $\mu_i$ might
be an increasing function of $|i|$, so that traits near $0$ are more
advantageous for survival. We call $\bp=(p_i,\, i\in \Z)$ 
the {\em mutation kernel\/}, and $\bmu=(\mu_i,\, i\in \Z)$ the
{\em mortality profile\/}. 
Suppose also that the total number $\sum_{i\in \Z} n_i$ of
individuals is bounded above by a fixed population ceiling~$N$. Births
are progressively suppressed as the population approaches this ceiling,
vanishing when the ceiling is reached. Mutation continually replenishes
traits that would otherwise be lost through selection.

Let $(\bn(t),\, t\geq 0)$, with $\bn=(n_i,\,
i\in \Z)$, be a Markov chain in continuous time taking values in
the countable state space
$S=\{\bn\in \Z_+^{\Z}: \sum_i n_i\leq N\}$.
Write $|\bn|:=\sum_{j\in\Z}n_j$ for the total population size.
Then, the non-zero transition rates are 
\begin{equation}
q(\bn,\bn+\be_i)
=
\lambda\left(1-\frac{|\bn|}{N}\right)
\sum_{k\in\Z}n_kp_{i-k},
\qquad
q(\bn,\bn-\be_i)=\mu_i n_i,
\qquad i\in\Z,
\label{TheModel}
\end{equation}
where $\be_i$ is the unit vector with $1$ in its $i$-th position.
For each $\bn\in S$, only finitely many coordinates $n_i$ are non-zero.
Consequently, for every $\bn\in S$, both
$$
\sum_{i\in\Z}q(\bn,\bn+\be_i)
=
\lambda\left(1-\frac{|\bn|}{N}\right)|\bn|
\quad\text{and}\quad
\sum_{i\in\Z}q(\bn,\bn-\be_i)
=
\sum_{i\in\Z}\mu_i n_i
$$
are finite, even when the mortality profile is unbounded. Thus, the
displayed transition rates define a stable $Q$-matrix on~$S$. However,
when the mortality profile is unbounded, the total transition rate need
not be uniformly bounded over~$S$, so stability alone does not establish
non-explosivity of the associated minimal process. To verify
non-explosivity, observe that the aggregate birth rate out of state
$\bn$ is bounded above by $\lambda N$. Consequently, the number $B(t)$
of births by any finite time $t$ is stochastically dominated by a
Poisson random variable with mean $\lambda Nt$, and hence is almost
surely finite. If $D(t)$ denotes the number of deaths by time $t$, then
$|\bn(t)|=|\bn(0)|+B(t)-D(t)\geq 0$, and therefore
$D(t)\leq|\bn(0)|+B(t)$. It follows that the total number of transitions
by time $t$ satisfies $B(t)+D(t)\leq|\bn(0)|+2B(t)<\infty$, almost
surely. Thus, the minimal process is non-explosive.

Observe that $\bzero$ is an absorbing state, corresponding to population
extinction. We assume that $\bzero$ is accessible from every state in
$S\setminus\{\bzero\}$, and that the latter states form an irreducible
communicating class. In particular, we assume that $\mu_i>0$ for every
$i\in\Z$, and that the mutation kernel $\bp$ is irreducible (in the
sense that the random walk which jumps according to~$\bp$ is irreducible
on~$\Z$). Although there are at most $N$ individuals, they may have
arbitrarily large trait indices, and thus extinction might not occur.

For comparison, if there were only a single trait $0$, and hence no
mutation or trait variation, the population size $n(t)$ would be a
birth--death process taking values in $\{0,1,\ldots,N\}$ with birth
rates $\lambda n(1-n/N)$ and death rates $\mu n$, where $\mu=\mu_0$.
This is the continuous-time SIS epidemic model (Weiss and
Dishon~\cite{WD71}). It provides a useful comparison: although
extinction occurs with probability~$1$, when $\lambda>\mu$ and $N$ is
large, the system may remain in a quasi-equilibrium (endemic) state over
a long time period. This raises the question of whether similar
behaviour arises in the present model, and what determines the resulting
trait composition. One means of addressing these questions is via a
deterministic approximation that is expected to be faithful in the limit
as~$N$ becomes large. An approximating model is easily identified
because the transition rates~(\ref{TheModel}) are {\em density
dependent\/} (refer to Kurtz~\cite{Kur70}), suggesting the following
system of differential equations:
\begin{equation}
\dot x_i = F_i(\bx):=
\lambda\left(1-\sum_jx_j\right)\sum_kx_kp_{i-k} -\mu_i x_i, 
\qquad i\in\Z,
\label{PKP1} 
\end{equation}
with $x_i(t)$ to be interpreted as the (large-$N$) density of
individuals with trait~$i$ at time~$t$ measured {\em relative to the
population ceiling\/} (in (\ref{PKP1}), and henceforth, sums shall be
over~$\Z$).
The natural state space is
$E=
\left\{\bx\in\BS: x_i\geq0\text{ for all }i,\ \sum_i
x_i\leq1\right\}$.
However, when mortality is unbounded, the right-hand side of
\eqref{PKP1} need not belong to $\BS$ for every $\bx\in E$, so we seek
solutions to~\eqref{PKP1} in the mild sense; see
Theorem~\ref{Thm:mildsolution} below.

In particular, we approximate $\bX(t):=\bn(t)/N$ by $\bx(t)$
when $N$ is large. That the trait space is countably infinite prevents a
direct application of the original results of Kurtz~\cite{Kur70,Kur71}.
Instead, we appeal to the results of Barbour and Luczak~\cite{BL12b},
who developed an approximation theory for population processes with
countably many types and obtained explicit error bounds.

The behaviour of~\eqref{PKP1} turns out to depend fundamentally on
whether mortality remains bounded or is confining. We call the mortality
profile \emph{confining} if $\mu_i\to\infty$ as $|i|\to\infty$. When
$\sup_i\mu_i<\infty$, the selective pressure remains uniformly bounded
across the trait space, whereas under confining mortality individuals
with extreme traits experience increasingly strong mortality. 
These two regimes are not exhaustive, since an unbounded mortality profile 
need not be confining, but they capture the principal contrast studied here.
As we shall see, the two regimes lead to markedly different long-term
behaviour. In particular, confining mortality provides a mechanism for
the existence of non-trivial deterministic equilibria. 

However, when the death rates are unbounded, $D\bx$ need not belong to
$\BS$ for every $\bx\in E$, so the vector field need not define an
$\BS)$-valued map on all of $E$. Classical Banach-space ODE theory (see
for example~\cite{Fat99}) does not apply directly. We therefore first
establish the existence of mild solutions to (\ref{PKP1}), which will
provide the basis for the deterministic approximation.

\begin{theorem}\label{Thm:mildsolution}
For every $\bx_0\in\BS$, there exists
$t_{\max}\in(0,\infty]$ such that the initial value problem~\eqref{PKP1},
with $\bx(0)=\bx_0$, has a unique maximal mild solution $\bx\in
C([0,t_{\max});\BS)$. 
It is given componentwise by 
\begin{equation}
x_i(t)=e^{-\mu_i t}x_i(0)+\int_0^t
e^{-\mu_i(t-s)}\,\lambda\bigl(1-m(\bx(s))\bigr)\, (P\bx(s))_i\,ds,
\qquad i\in \Z,
\label{mildsolution}
\end{equation}
where $m(\bx)=\sum_i x_i$ and $(P\bx)_i = \sum_{k} x_k p_{i-k}$.
\end{theorem}

We are only interested in biologically meaningful deterministic
trajectories: those in $E$. Our next theorem shows that mild solutions
that start in $E$ remain in $E$.

\begin{theorem}
\label{staysinE}
Let $\bx\in C([0,t_{\max});\BS)$ be the maximal mild solution
to~\eqref{PKP1}, with initial value~$\bx_0$, specified in
Theorem~\ref{Thm:mildsolution}. If $\bx_0\in E$, then $t_{\max}=\infty$
and $\bx(t)\in E$ for all $t\geq0$.
\end{theorem}

\medskip
\noindent
We next present a quantitative stochastic approximation result based on
Theorem~4.7 of Barbour and Luczak~\cite{BL12b}. Their general theorem
requires moment bounds that control the movement of population mass
through the type space, as well as the growth of the transition rates
associated with distant types. We therefore impose, for this result
only, polynomial growth of the mortality profile, a sufficiently high
moment condition on the mutation kernel, and corresponding moment
conditions on the initial states. These additional assumptions are
satisfied by all the mutation kernels and mortality profiles used in the
numerical experiments reported below. They are not standing assumptions
for other results of the paper.

\begin{theorem}
\label{BL}
Suppose that there exist an integer $q\geq1$ and a constant
$C_{\!\bmu}<\infty$ such that
\[
 \mu_i\leq C_{\!\bmu}(1+|i|)^q,\qquad i\in\mathbb Z,
\]
and suppose that the mutation kernel satisfies the sufficient
moment condition
\[
 \sum_{\ell\in\mathbb Z}(1+|\ell|)^{11q}p_\ell<\infty.
\]
Let $(\bx(t),t\geq0)$ be the mild solution to~\eqref{PKP1}, with
initial value $\bx_0\in E$, specified in
Theorem~\ref{Thm:mildsolution}, and let
$\bX(t):=N^{-1}\bn(t)$. Suppose also that
\[
 \sum_i(1+|i|)^{11q}x_i(0)<\infty
\]
and that, for some constant $C_*<\infty$ independent of $N$,
\[
 \sum_i(1+|i|)^{11q}X_i(0)\leq C_*.
\]
Then, for each $t<\infty$ and each $K_1<\infty$, there exist constants
$K_2$ and $K_3$, depending on $t$ and $K_1$ but not on $N$, such that, 
for all sufficiently large $N$, if
\[
 \norm{\bX(0)-\bx_0}_1
 \leq
 K_1\sqrt{\frac{\log N}{N}},
\]
then
\[
 \Pr\left(
 \sup_{0\leq s\leq t}
 \norm{\bX(s)-\bx(s)}_1
 >
 K_2\sqrt{\frac{\log N}{N}}
 \right)
 \leq
 K_3\frac{\log N}{N}.
\]
\end{theorem}

\medskip
\noindent
\emph{Remark}.
The curious exponent $11q$ is a sufficient, rather than an intrinsic
or optimal, moment order. It arises from the particular choice of
weights used to verify the hypotheses of Barbour and
Luczak~\cite{BL12b}; the details are given in the Appendix. A
model-specific martingale argument might permit a smaller exponent.

\medskip
\noindent
Thus, subject to the additional assumptions in Theorem~\ref{BL}, the
probability that $\bX$ deviates from $\bx$ by more than order
$\sqrt{\log N/N}$ in the $\ell_1$ norm over a fixed finite time interval
is of order at most $\log N/N$. The restriction to finite time intervals
is important. Theorem~\ref{BL} does not describe the eventual behaviour
of the stochastic process, which may be affected by extinction and, in
non-confining regimes, by the migration of population mass through the
trait space.

The moment and growth conditions in Theorem~\ref{BL} arise from our
direct application of the general theorem of Barbour and
Luczak~\cite{BL12b} and are not claimed to be sharp for the present
model. Since the aggregate birth rate is uniformly bounded and death
transitions are dissipative, some of these conditions might be weakened
by estimating the relevant martingales directly. We do not pursue that
refinement here.

\medskip
A formal argument based on summing $\dot x_i = F_i(\bx)$ over $i\in \Z$
(refer to \eqref{PKP1}) shows that aggregate density
$m(t):=\sum_i x_i(t)$ satisfies
\begin{equation}
\dot{m}=\lambda m(1-m) - \sum_i \mu_i x_i.
\label{dmdt}
\end{equation}
We establish this rigorously for the mild solution under the
condition that mortality is bounded.

\begin{theorem}
\label{thm:dmdt}
Let $(\bx(t),t\geq 0)$ be the mild solution to \eqref{PKP1},
with initial value $\bx_0\in E$, specified in
Theorem~\ref{Thm:mildsolution}, and let
$m(t)=\sum_i x_i(t)$. If $\sup_i\mu_i<\infty$, then
$(m(t),t\geq 0)$ satisfies~\eqref{dmdt}
with initial value $m(0)=m_0:=\sum_i x_i(0)$.
\end{theorem}

\noindent
\emph{Remarks}.
(1) The ODE (\ref{dmdt}) seems very natural; for example, if
$\mu_i\equiv\mu$ then (\ref{dmdt}) reduces to the Verhulst
model~\cite{Ver1838}, $\dot{m}=\lambda m(\rho-m)$, where
$\rho=1-\mu/\lambda$ (the {\em carrying capacity\/}). It is known to be
a faithful approximation to SIS model, referred to earlier, when~$N$ is
large (Theorem~1 of \cite{Pol24}).
(2) Direct differentiation of the componentwise mild-solution formula
\eqref{mildsolution} shows that, for each fixed $i$, $x_i$ is
continuously differentiable and satisfies $\dot x_i=F_i(\bx)$, even when
mortality is unbounded. This does not imply that $\bx$ is differentiable
as a $\BS$-valued map, because $D\bx$, and hence $F(\bx)$, need not
belong to $\BS$. 
When mortality is bounded, $D$ is a bounded operator on
$\ell_1(\mathbb Z)$, and the mild solution is therefore classical; see
the proof of Theorem~\ref{thm:dmdt} in the Appendix.

\medskip
\noindent
One would expect that, whenever $\sum_j n_j(t)>0$, the \emph{proportion}
$n_i(t)/\sum_j n_j(t)$ of individuals having trait $i$ would be
approximated by the corresponding deterministic trait proportion
$\pi_i(t):=x_i(t)/\sum_j x_j(t)$ when $N$ is large.
Theorem~\ref{thm:pidot_continuum} below identifies a differential
equation governing $\bpi(t)=(\pi_i(t),\, i\in\Z)$ under the condition
that mortality is bounded. We first need to establish that deterministic
trait proportions are well defined.

By Theorem~\ref{thm:dmdt}, $m$ is continuously differentiable and 
satisfies~\eqref{dmdt}. Since $0\leq m(t)\leq 1$ and $\sum_i\mu_i x_i(t)\leq
Mm(t)$, it follows that $\dot m(t)\geq -Mm(t)$. Therefore,
$$
\frac{d}{dt}\bigl(e^{Mt}m(t)\bigr)
= e^{Mt}\bigl(\dot m(t) +M m(t) \bigr) \geq 0,
$$
and hence $e^{Mt}m(t)\geq m_0$. So, if $m_0>0$, then
$m(t)\geq m_0e^{-Mt}>0$ for all $t\geq 0$.
In particular, the ratio $\pi_i(t)={x_i(t)}/{m(t)}$ is well defined for all
$t\geq 0$.

\begin{theorem}
\label{thm:pidot_continuum}
In the setting of Theorem~\ref{thm:dmdt}, suppose additionally that
$\bx_0\neq\bzero$, and define $\bpi(t)=(\pi_i(t),i\in\mathbb Z)$ by
$\pi_i(t):=x_i(t)/m(t)$. Then, for all $i\in\mathbb Z$,
\begin{equation}
\dot\pi_i=
\lambda (1-m)\left( \sum_{k} \pi_k p_{i-k} - \pi_i \right) 
+ \pi_i (\bar\mu - \mu_i),
\label{pidot}
\end{equation}
where $\bar\mu(t)= \sum_i \mu_i \pi_i(t)$.
\end{theorem}

\begin{proof}
Since $m(t)>0$, we may differentiate $\pi_i=x_i/m$ to obtain
$\dot\pi_i=\dot x_i/m-\pi_i\dot m/m$. From \eqref{PKP1} and
\eqref{dmdt}, $\dot x_i/m=\lambda(1-m)\sum_k\pi_kp_{i-k}-\mu_i\pi_i$ and
$\dot m/m=\lambda(1-m)-\bar\mu$. Substitution and rearrangement give
\eqref{pidot}.
\end{proof}

\noindent
\emph{Remark}.
It is useful to interpret $\pi_i(t)$ as the probability that an
individual sampled from the population at time~$t$ has trait~$i$. Thus,
$\bar\mu(t)$ is the mean per-capita death rate at time~$t$. The first
term in~\eqref{pidot} represents mutation, or dispersal through the trait
space, while the second represents selection relative to the mean
mortality. Equation~\eqref{pidot} is similar in structure to the
finite-dimensional \emph{replicator-mutator equation} described
in~\cite{PN2002}, but here mutation is modulated by the demographic
factor $\lambda(1-m)$.

\medskip
\noindent
{\bf Numerical experiments}.
The numerical calculations were performed on finite-dimensional
truncations of the model. For fixed $I\geq 1$ and $K\leq I$, traits were
restricted to $\T_I=\{-I,\ldots,I\}$, while mutation displacements were
restricted to $\D_K=\{-K,\ldots,K\}$. Throughout, we used \mbox{$K=20$}.
The trait-space truncation parameter $I$ was chosen so that truncation
effects were negligible over the time interval of interest. The roles of
the two truncation parameters are different. The parameter $K$
determines the range of mutation jumps retained in the approximation,
whereas $I$ determines the size of the trait space. For a given pair
$(I,K)$, the traits $\{-I+K,\ldots,I-K\}$ form a ``safe region'', in the
sense that individuals with traits in this set experience the same
mutation mechanism as in the infinite model. By contrast, traits within
distance $K$ of the boundaries $\pm I$ are affected by the truncation
because some mutation jumps would otherwise leave $\T_I$. 
For each experiment, $I$ was chosen sufficiently large that, over the
time interval displayed, the occupied trait range in the stochastic
simulation remained within the safe region. No truncation artefact was
apparent in the corresponding deterministic trajectories. In
directional-selection examples, where the trait distribution continues
to move through the trait space, a fixed truncation can only be expected to
provide an accurate approximation over a finite time interval.

We took $\bp$ to be the discrete Laplace distribution~\cite{IK06} with
parameter $r\in(0,1)$, namely
\begin{equation}
p_i=\frac{1-r}{1+r}\,r^{|i|},
\qquad i\in\Z,
\label{laplace}
\end{equation}
and approximated it by the renormalised 
truncated kernel $\bp_K=(p_i,\,i\in\D_K)$,
given by
\begin{equation}
p_i=
\frac{1-r}{1+r-2r^{K+1}}\,r^{|i|},
\qquad i\in\D_K.
\label{laplacet}
\end{equation}
This kernel places greatest weight on trait~0, with probabilities
decreasing geometrically as~$|i|$ increases. Consequently, most
mutations produce small phenotypic changes, while larger mutational
steps occur only rarely. For the value $r=0.4$ used in the simulations,
the truncation at $K=20$ omits less than $10^{-8}$ of the total mutation
probability mass, so the error introduced by truncating the mutation
kernel is negligible on the scale of the numerical results reported
here.
We considered two mortality profiles corresponding to
``directional selection'' with $c=\sup_i \mu_i<\infty$. The first was
$$
  \mu_i=\mu_\infty + (c-\mu_\infty)\left( 1/2 - \tan^{-1}(i/e)/\pi\right),
$$
where $\mu_\infty$ is the limiting per-capita mortality rate, and $e>0$
controls the slope near $i=0$, so that $\mu_i$ decreases from $c$ to
$\mu_\infty$, and larger values of $e$ elongate the mortality profile.
Thus, higher trait values confer a survival advantage, with improvements
becoming less significant at high trait values. Examples might include
running speed, toxin resistance, thermal tolerance, visual acuity,
metabolic efficiency, and camouflage effectiveness. We call this the
``arctan'' profile. The other profile we considered was
$$
\mu_i=\mu_\infty + \frac{c-\mu_\infty}{1+\exp(i/e)},
$$
where the parameters have similar interpretations. We call this the
``logistic'' profile. Observe that the former has 
$\mu_i-\mu_\infty\sim (c-\mu_\infty) e/(\pi i)$, 
while the latter has $\mu_i-\mu_\infty\sim
(c-\mu_\infty)e^{-i/e}$, representing a sharper decrease. Together,
these choices capture a realistic mutation-selection balance in which
variation is continuously generated by mutation, shaped by selection,
and regulated by density-dependent competition.

All experiments were conducted using MATLAB\footnote{MATLAB Version:
26.1.0.3276743 (R2026a) Update 3, Natick, Massachusetts: The MathWorks
Inc.; 2026.}. 
Stochastic trajectories were generated using the standard
holding-time/jump-chain construction for continuous-time Markov chains.
The deterministic system was integrated using {\tt ode45}, with relative 
and absolute tolerances both set to $10^{-4}$.
In our first experiment we set $N=1000$ and $\lambda=0.06$. The mutation
kernel was the truncated discrete Laplace distribution~\eqref{laplacet}
with $K=20$ and $r=0.4$, and the mortality profile was logistic with 
$\mu_\infty=0.1$,
$c=1$, and $e=10$. The initial numbers were $n_i(0)=10$ for
$i=-40,\ldots,40$, and $0$ otherwise, so that $\sum_{i\in\T_I}
n_i(0)=810$. The results are shown in Figure~\ref{fig:FiguresA}.
Plot~(a) displays the expected trait value $\sum_i i n_i(t)/\sum_j
n_j(t)$, together with the minimum and
maximum traits present in the population. The expected trait increases
steadily, reflecting selection in favour of larger trait values, while
the minimum trait also drifts upward. By time $t_{\rm fix}=50.2678$ only
trait~32 remained, although this final lineage disappeared shortly
afterwards at time $t_{\rm ext}=54.9199$. Plot~(b) shows the aggregate
population density relative to $N$, together with its deterministic
approximation. The simulated trajectory closely follows the
deterministic trajectory, consistently with Theorem~\ref{BL}.
Nevertheless, the population displays evanescent behaviour: all traits
initially represented in the population disappear within a relatively
short period, and the deterministic densities $x_i(t)$ also tend to $0$
as the total density declines. Plot~(c) compares the simulated expected
trait value with its deterministic approximation over the initial phase
of the evolution, where the agreement is particularly close.

\begin{figure}[htbp]
\centering
\includegraphics[height=0.25\textheight,width=0.7\textwidth]{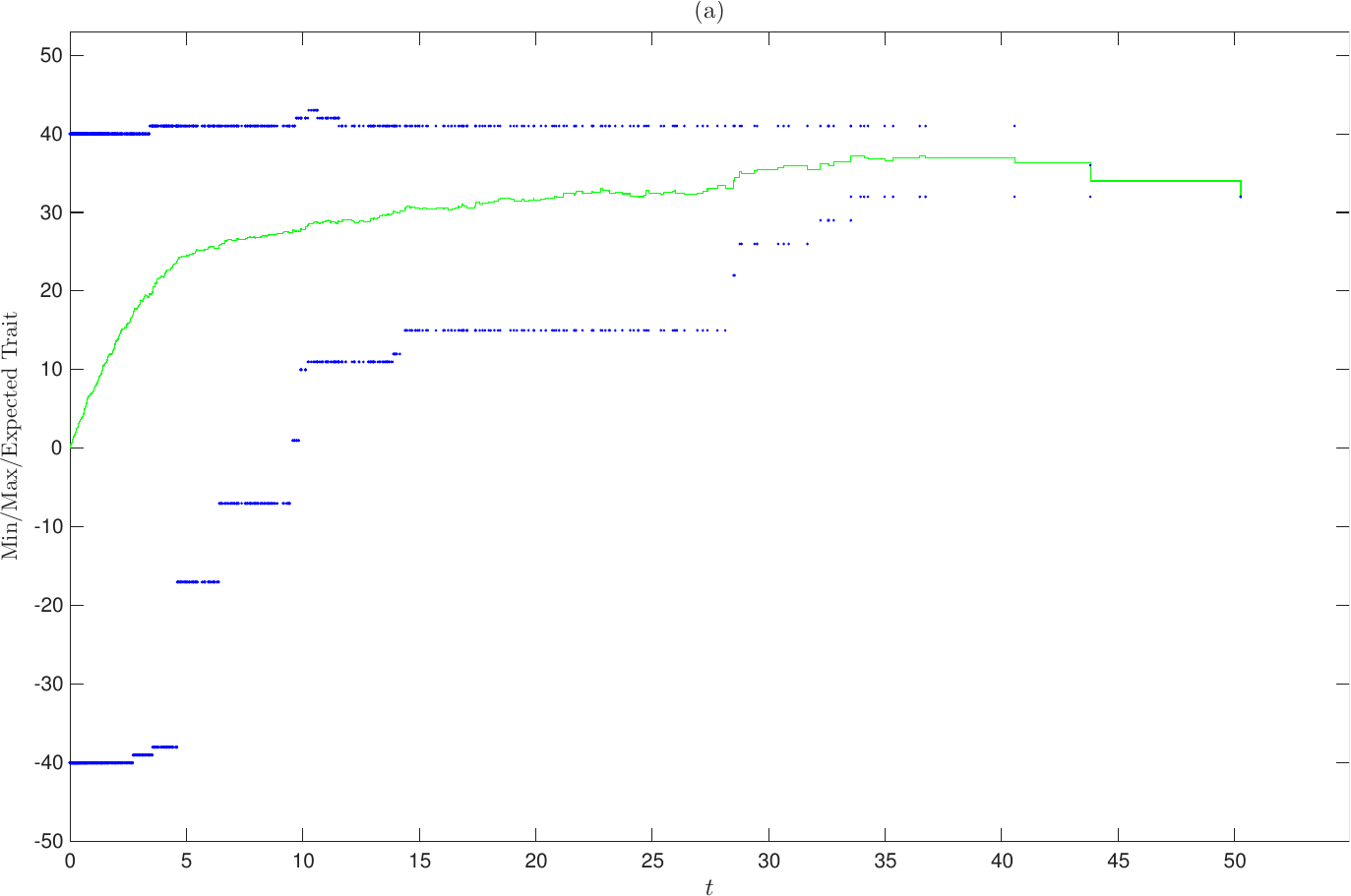}
\includegraphics[height=0.25\textheight,width=0.7\textwidth]{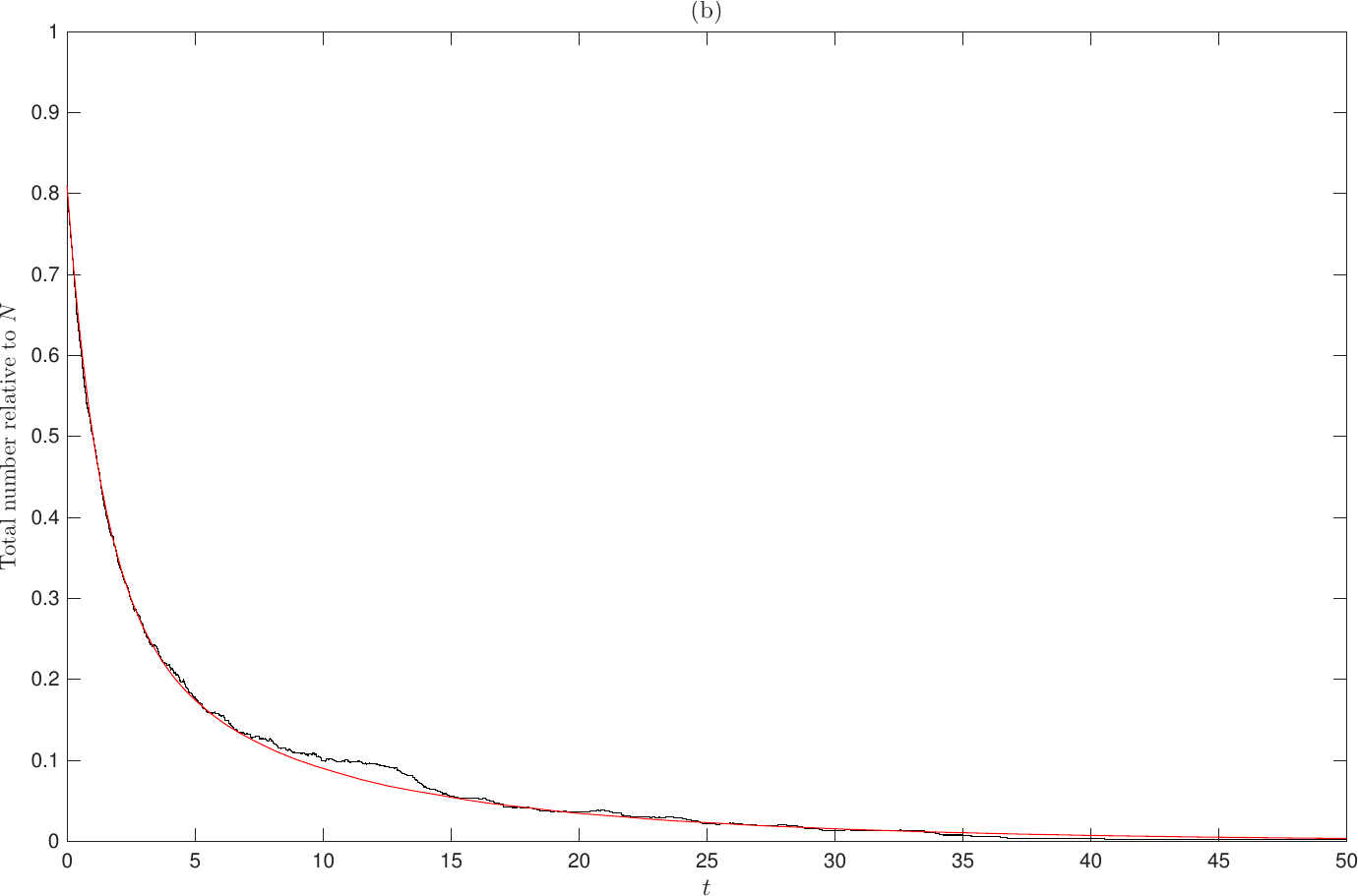}
\includegraphics[height=0.25\textheight,width=0.7\textwidth]{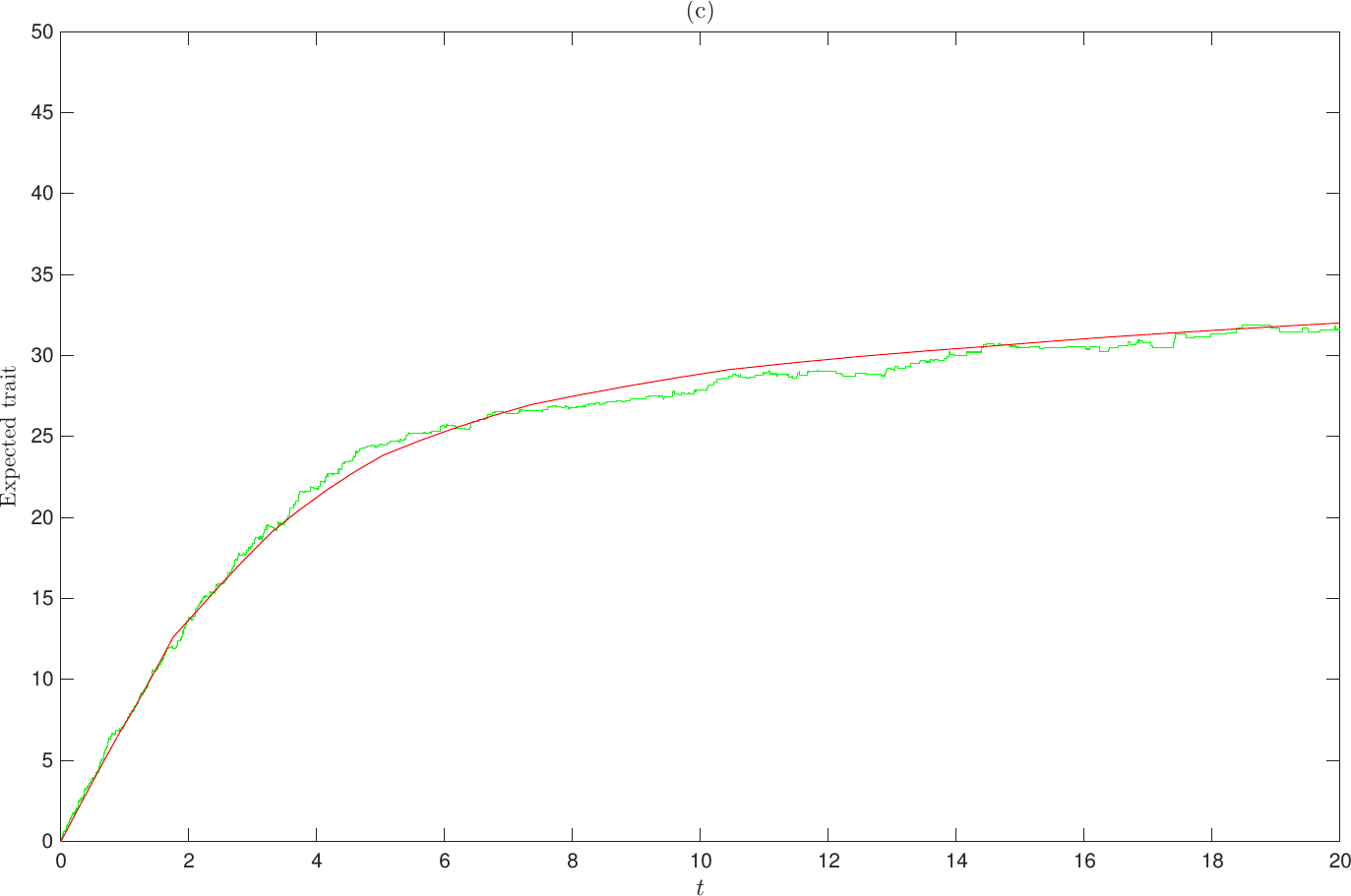}
\caption{Simulation and deterministic approximation for the logistic mortality
profile: $I=200$, $N=1000$, $\lambda=0.06$, $r=0.4$, $\mu_\infty=0.1$,
$c=1$, and $e=10$. (a)~Expected trait value (green) together with the
minimum and maximum traits present (blue). Selection drives the
population towards larger trait values.
By time $t_{\rm fix}=50.2678$, potentially temporary fixation of trait~32 
had occurred, 
although this final lineage disappeared shortly afterwards at time 
$t_{\rm ext}=54.9199$. (b)~Aggregate population density relative to
$N$ (black) and its deterministic approximation (red). The stochastic
trajectory closely tracks the deterministic approximation. (c)~Expected
trait value (green) and its deterministic approximation (red) over the
initial phase of the evolution. The figure illustrates evanescent
behaviour: all trait densities eventually decline to $0$.}
\label{fig:FiguresA}
\end{figure}

\begin{figure}[htbp]
\centering
\includegraphics[height=0.27\textheight,width=0.85\textwidth]{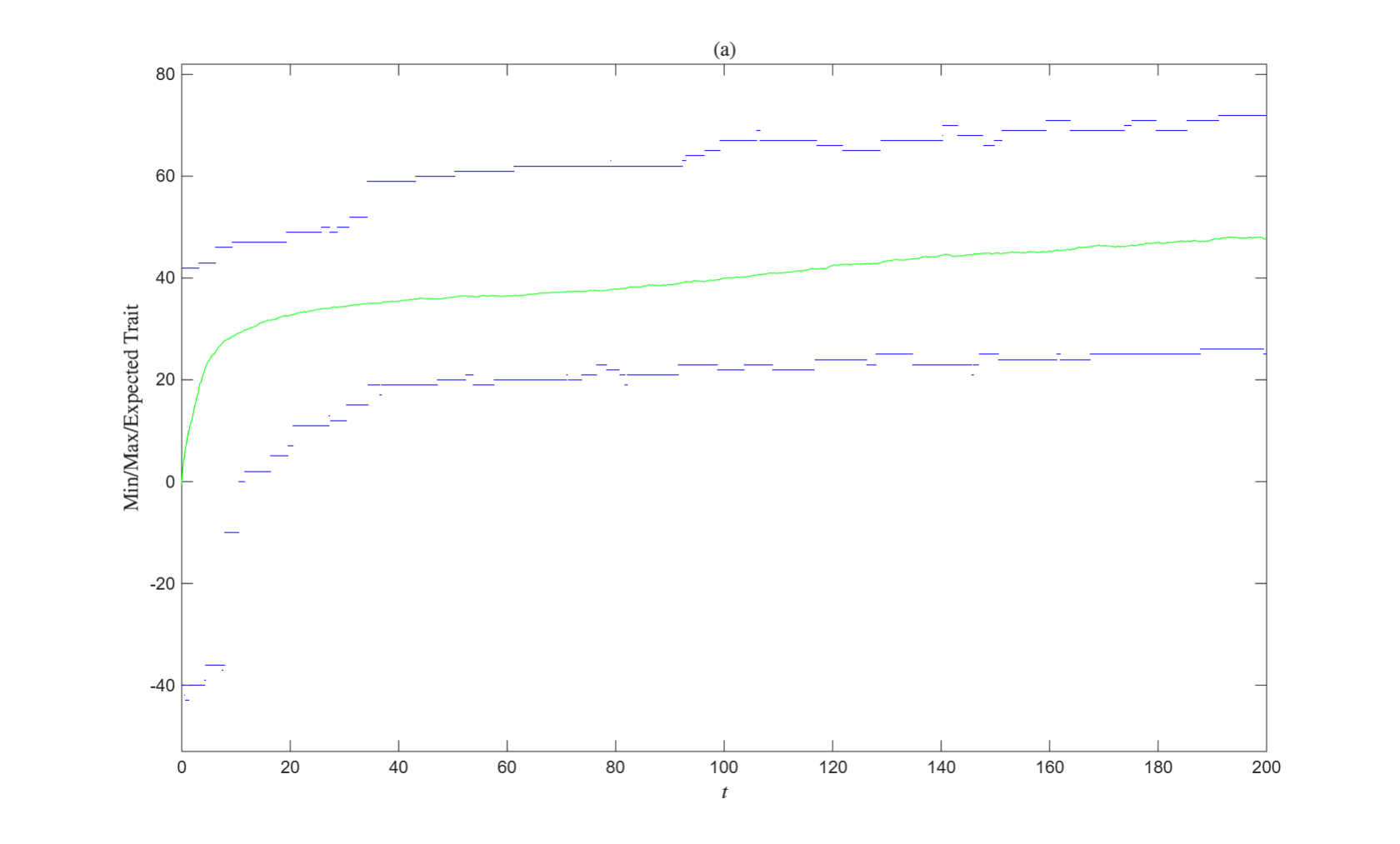}
\includegraphics[height=0.27\textheight,width=0.85\textwidth]{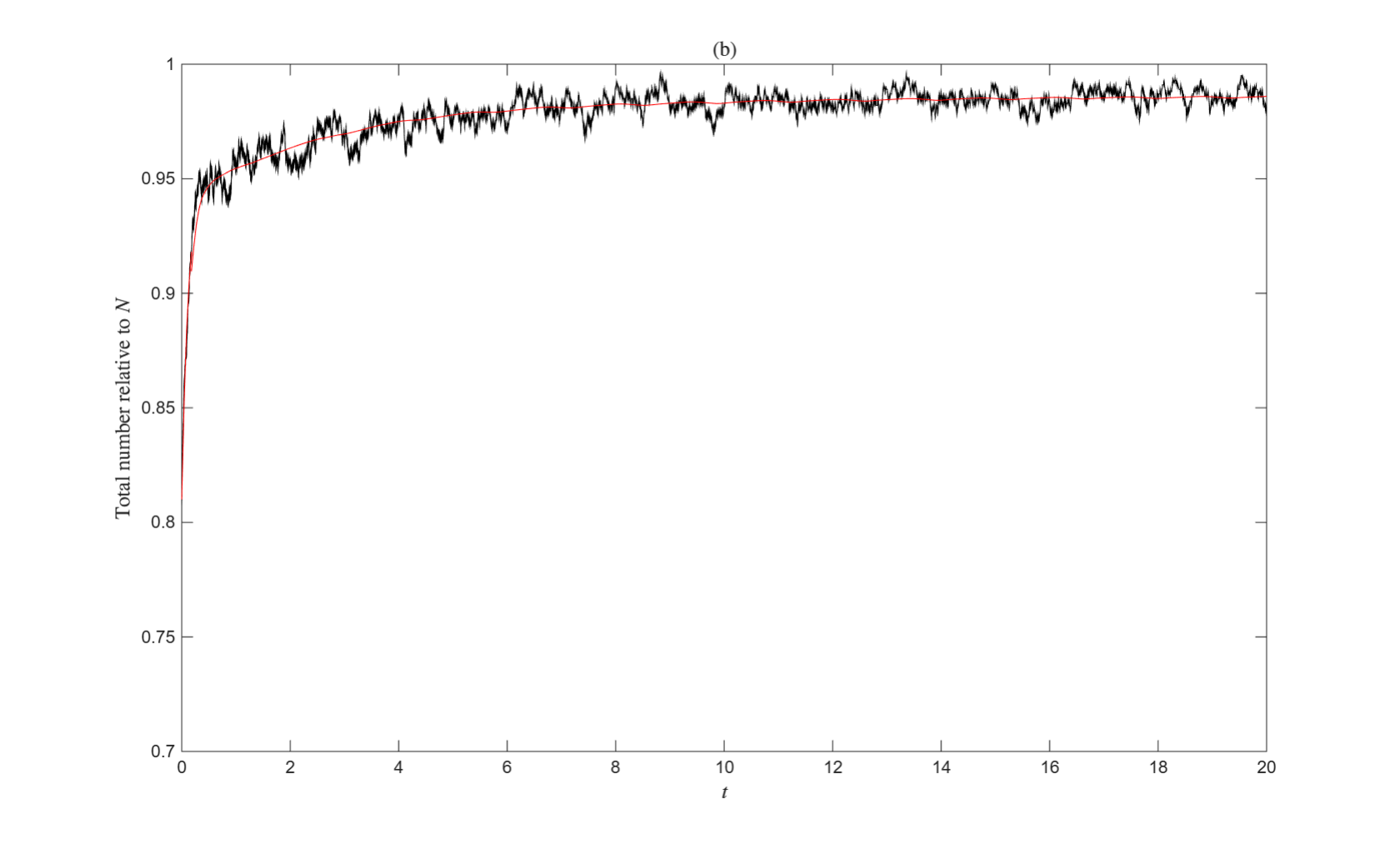}
\includegraphics[height=0.27\textheight,width=0.85\textwidth]{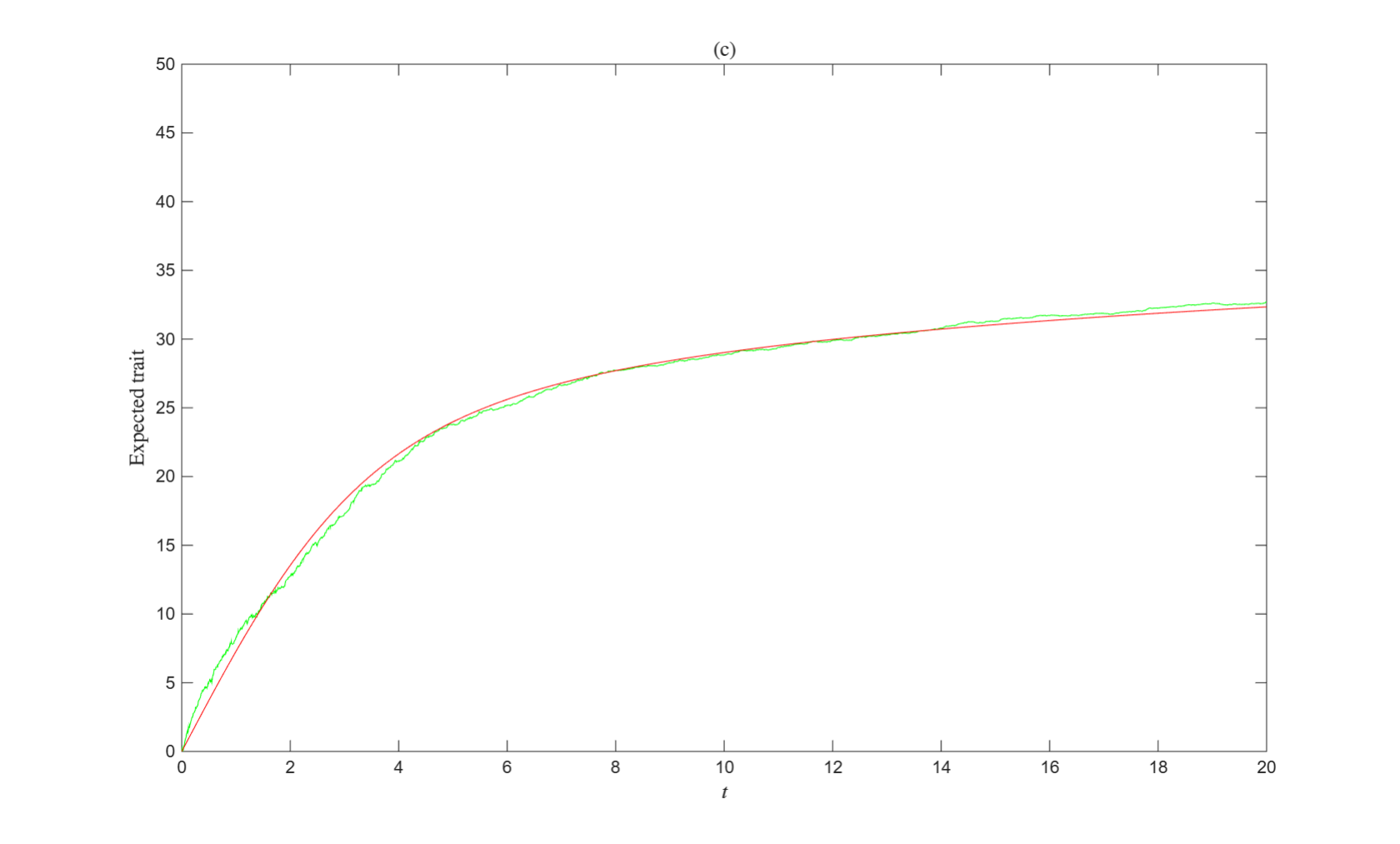}
\caption{Simulation and deterministic approximation for the logistic mortality
profile: $I=2000$, $N=1000$, $\lambda=10$, $r=0.4$, $\mu_\infty=0.1$, $c=1$,
and $e=10$. (a)~Expected trait value (green) together with the minimum
and maximum traits present (blue). (b)~Aggregate population density
relative to $N$ (black) and its deterministic approximation (red). 
The aggregate density rapidly approaches a quasi-equilibrium level and 
fluctuates about its deterministic counterpart.
(c)~Expected trait value (green) and its
deterministic approximation (red) over the initial phase of the
evolution. Although the population persists at positive density, the
mean trait continues to increase, suggesting ongoing evolutionary drift
through the trait space.}
\label{fig:FiguresB}
\end{figure}

\begin{figure}[htbp]
\centering
\includegraphics[height=0.27\textheight,width=0.85\textwidth]{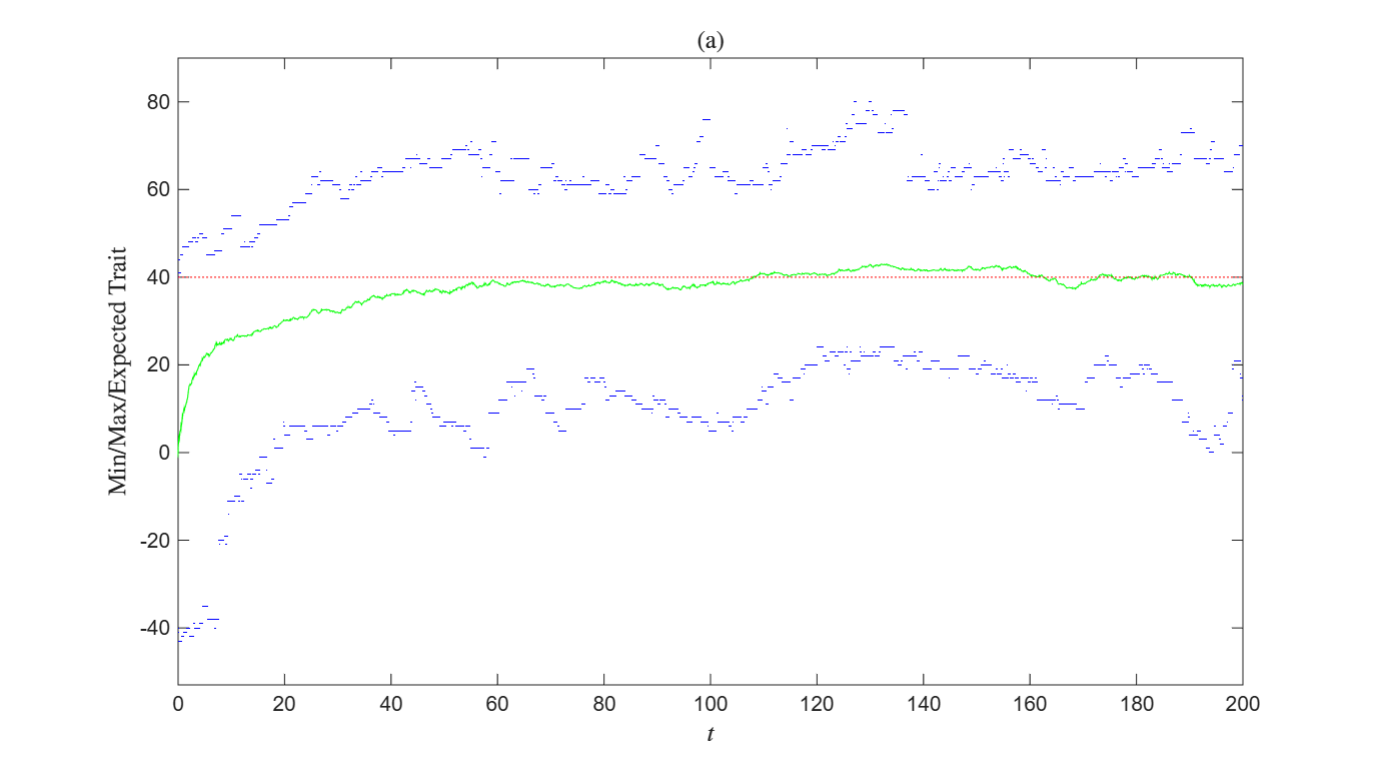}
\includegraphics[height=0.27\textheight,width=0.85\textwidth]{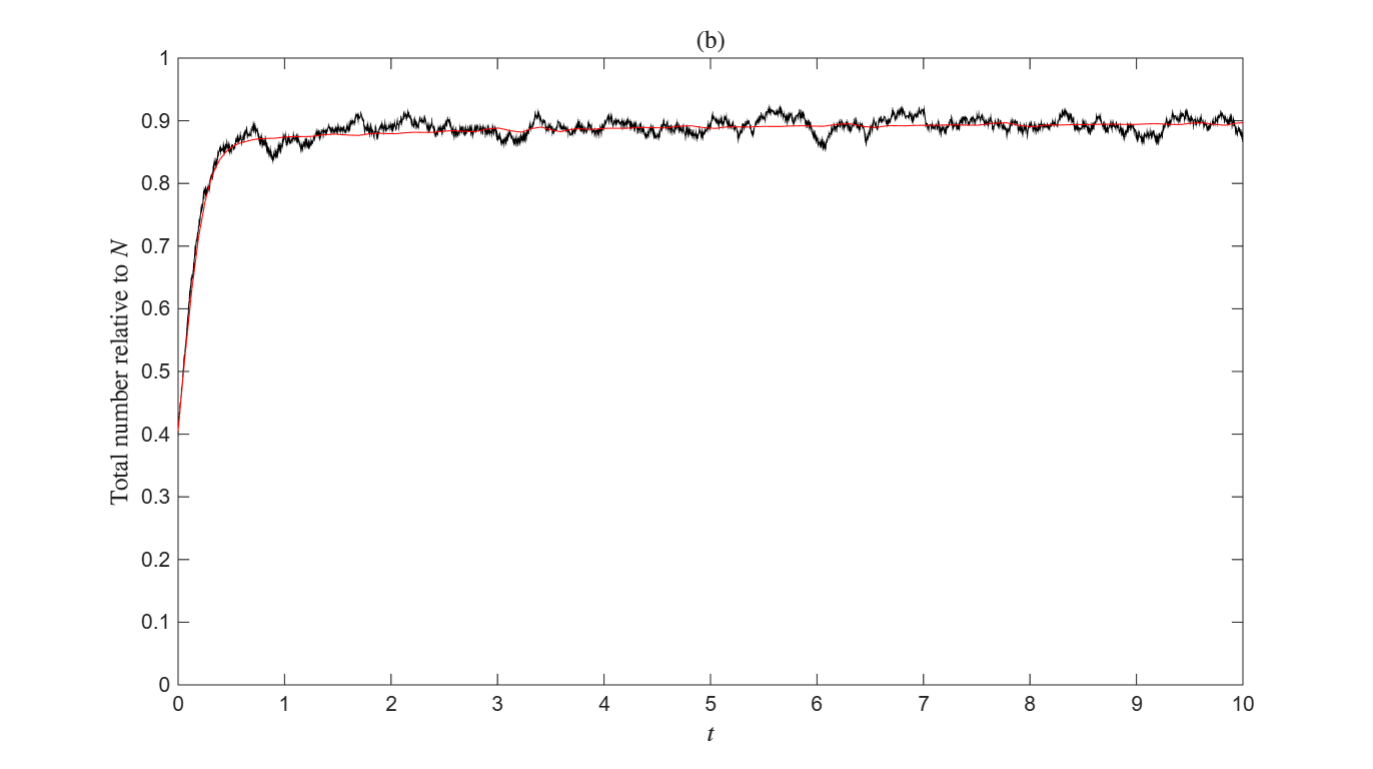}
\includegraphics[height=0.27\textheight,width=0.85\textwidth]{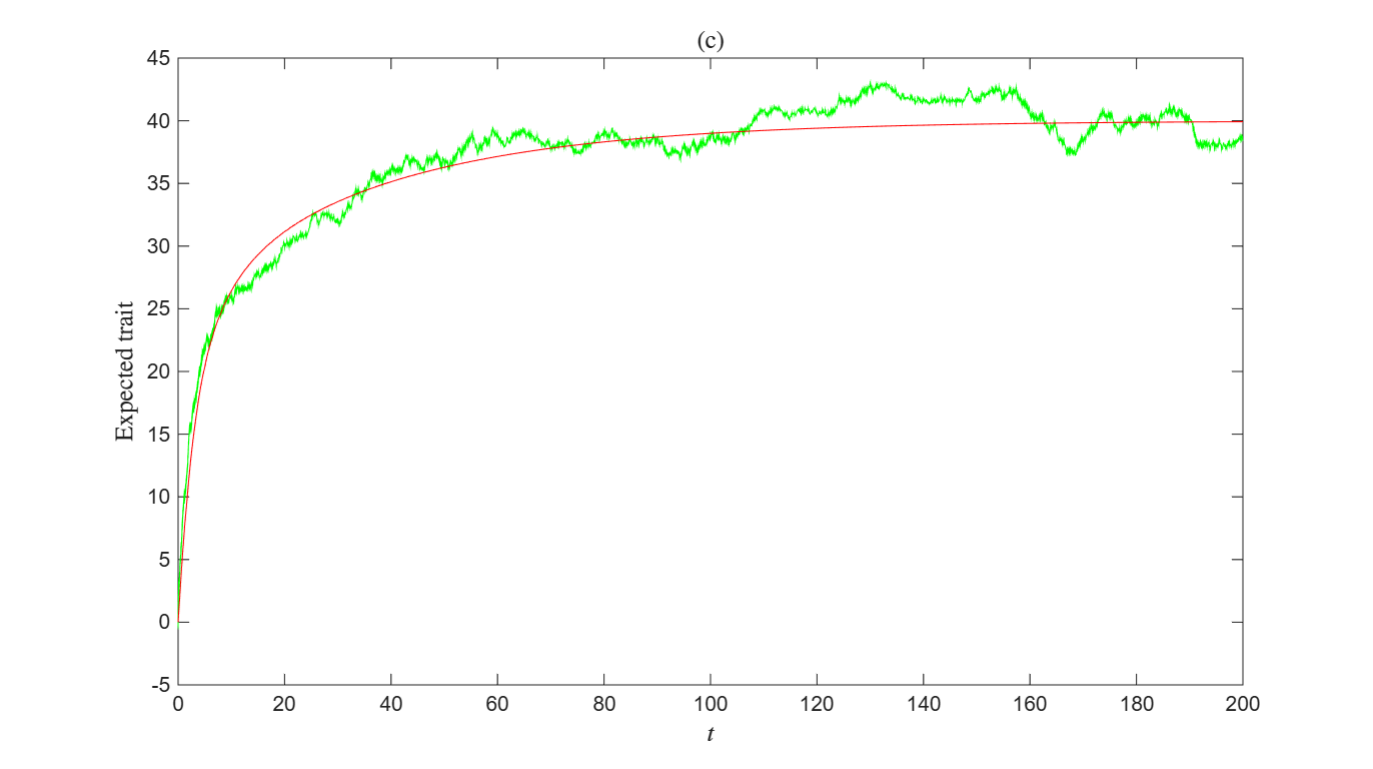}
\caption{Simulation and deterministic approximation for 
stabilizing selection: $I=200$, $N=1000$, $\lambda=10$, $r=0.4$, 
$\mu_{\text{min}}=1$, $f=10$. 
(a)~Expected trait value (green) together with the minimum
and maximum traits present (blue). 
The red dotted line indicates the optimal trait $i^*=40$.
(b)~Aggregate population density
relative to $N$ (black) and its deterministic approximation (red). The
density rapidly approaches a quasi equilibrium fluctuating about the
deterministic equilibrium. (c)~Expected trait value (green) and its
deterministic approximation (red).
The population remains localized near the optimal trait,
consistent with the existence of a confining mutation--selection
equilibrium.}
\label{fig:FiguresC}
\end{figure}

Our next experiment differed only in that the birth rate was increased
to $\lambda=10$. The results are shown in Figure~\ref{fig:FiguresB}. The
behaviour is markedly different. The aggregate density now rapidly
approaches a persistent quasi equilibrium, fluctuating about its
deterministic counterpart. Once again, the simulated trajectory closely
tracks the deterministic trajectory. At the same time, however, the
mean trait value continues to increase. This raises a natural question:
does the rate of increase eventually slow sufficiently for the trait
distribution to enter a quasi-stationary regime? 
The numerical evidence suggests otherwise. Although the logistic
mortality profile is not constant, the special case $\mu_i\equiv\mu$,
analysed in the next section, exhibits a closely related mechanism: the
population may persist at positive aggregate density while
$\pi_i(t)\to 0$ for every fixed $i\in\Z$.
These experiments motivate the special-case analyses of
Sections~\ref{sec:noselection} and~\ref{sec:nomutation}, and the general
analysis of Section~\ref{sec:generalcase}.

We next considered two confining mortality profiles.
The first corresponds to stabilizing selection:
$$
\mu_i=\mu_{\min}+a(i-i^*)^2,
$$
where the minimum mortality $\mu_{\min}$ is attained uniquely at the
optimal trait $i^*$. For a particular value of $I$, we set
$a=(f-1)\mu_{\text{min}}/\max((I-i^*)^2,(-I-i^*)^2)$, where $f>1$,
so that $\max_{i\in \T_I} \mu_i = f\mu_{\text{min}}$. When experimenting with
varying $I$ (see Proposition~\ref{truncation} below) $a$ had a fixed
value. The second corresponds to disruptive selection:
$$
\mu_i=\mu_{\min}+a(i^2-(i^*)^2)^2.
$$
For the disruptive-selection profile, we chose
$$
a=
\frac{(f-1)\mu_{\min}}
{\max_{i\in\T_I}(i^2-(i^*)^2)^2},
$$
so that $\max_{i\in\T_I}\mu_i=f\mu_{\min}$.
In this case the minimum mortality $\mu_{\min}$ is attained at the two
equally favourable traits $\pm i^*$. 

For the first of these profiles we set $N=1000$, $\lambda=10$,
$i^*=40$, and $f=10$. (The disruptive-selection example is deferred
until Section~\ref{sec:generalcase}.) The results are shown in
Figure~\ref{fig:FiguresC}.
Plot~(a) displays the expected trait value together with the minimum and
maximum traits present in the population. Selection initially drives the
population towards larger trait values, but, in contrast to
Figure~\ref{fig:FiguresB}, the increase slows as the population
approaches the optimal trait $i^*=40$ (indicated). 
The occupied trait range likewise
becomes concentrated around this optimum.
This behaviour contrasts sharply with that observed in
Figure~\ref{fig:FiguresB}, where trait values continue to drift through
the trait space.
Plot~(b) shows the aggregate population density relative to $N$,
together with its deterministic approximation. 
Consistent with Theorem~\ref{BL}, the simulated trajectory closely
tracks the deterministic solution over the displayed time interval. Both
trajectories settle near a positive equilibrium.
Finally, Plot~(c) compares the simulated expected trait value with its
deterministic approximation. The agreement is again seen to be good.
Taken together, the three plots suggest convergence towards the
equilibrium predicted by the deterministic model, with the population
remaining localized near the optimal trait.
Not shown are simulations with the much smaller value of $\lambda=0.06$. 
The observed evanescent behaviour was similar to that shown in
Figure~\ref{fig:FiguresA}.

\section{$\mu_i\equiv\mu$ (no selection)}
\label{sec:noselection}

The special case $\mu_i\equiv\mu$ is particularly revealing, since the
effects of mutation and population regulation can be studied in the
absence of selection. The case $\bx_0=\bzero$ is trivial, so throughout
this section we suppose that $\bx_0\neq\bzero$, and write
$m_0=m(0)=\sum_i x_i(0)>0$.
We see immediately from~(\ref{dmdt}) that $m$
follows the familiar Verhulst model~\cite{Ver1838}:
$$
\dot m = \lambda m(1-m)-\mu m
= \lambda m \left(\rho -m \right),
$$
where $\rho=1-\mu/\lambda$.
It has the explicit solution
$$
m(t)=
\begin{cases}
\begin{aligned}
&{\ds \frac{\rho m_0}{m_0 + (\rho-m_0)e^{-\lambda\rho t}},}
&\text{if } \lambda\neq \mu,\\[1ex]
&{\ds \frac{m_0}{1+m_0 \mu t},}
&\text{if } \lambda=\mu,
\end{aligned}
\end{cases}
$$
where $m_0=\sum_i x_i(0)$.
It has the stable equilibrium $m^*=\rho$ when $\lambda>\mu$. Otherwise,
$0$ is the stable equilibrium. 
So~(\ref{PKP1}) becomes
\begin{equation}
\dot x_i
= 
\lambda (1-m(t))\sum_{k} x_k p_{i-k} -\mu x_i,
\label{PKP1b}
\end{equation}
with $m(t)$ available explicitly.
The trivial case where there is no mutation ($p_0=1$) has
$$
\dot x_i
= x_i \left(\lambda (1-m)-\mu \right)
= \lambda x_i \left(\rho-m\right),
$$
that is, $\dot x_i/x_i = \dot m/m$, from which it follows that 
$x_i(t)=x_i(0)m(t)/m(0)$.
So, if $\lambda> \mu$, $x_i(t)\to x_i(0)\rho/m_0$,
while if $\lambda\leq \mu$, $x_i(t)\to 0$ for all $i$.

Now assume that $p_0<1$. It is clear that if $\lambda\leq\mu$ then
$x_i(t)\to 0$ for all $i$, because $m(t)\to 0$. 
We now show that $x_i(t)\to 0$ for all $i$ even when $\lambda>\mu$. 
Under this condition, any non-zero equilibrium point $\bx^*$
of~\eqref{PKP1b} satisfies $\sum_{k} x_k^* p_{i-k}=x_i^*$, because
$m^*=\rho=1-\mu/\lambda$ implies that $\lambda(1-m^*)=\mu$. However,
this would require the equilibrium trait proportions $\bpi^*$ to satisfy
$P\bpi^*=\bpi^*$.
Unlike the finite-dimensional
replicator-mutator equation, the mutation operator~$P$ admits no
invariant probability distribution on $\Z$. 
Since $p_0<1$, the associated random walk is non-trivial; by the
standing irreducibility assumption, it is an irreducible random walk 
on~$\Z$, and hence is either transient or null recurrent.
Consequently there is no probability distribution~$\bpi$
satisfying $P\bpi=\bpi$, and hence (\ref{PKP1b}) possesses no
non-trivial equilibrium. To determine the long-term behaviour of
solutions, we employ a time-change argument, which at the same time
yields an explicit expression for the trait proportions $\bpi(t)$.

From~(\ref{pidot}) we see that $\bpi(t)$ satisfies $\dot\bpi = \lambda
(1-m)(P - \Id)\bpi$. So, changes in trait composition occur on a time
scale determined by the factor $\lambda(1-m(t))$, a quantity that can be
interpreted as the effective per-capita birth rate (which decreases as
the population approaches its ceiling). Since mutation occurs only at
birth, the natural clock governing changes in trait composition is
cumulative reproductive opportunity rather than chronological time.
Accordingly, we introduce the new time variable
$$
\tau(t)=\int_0^t \lambda(1-m(s))\,ds,
$$
which can be interpreted as an ``evolutionary clock'':
equal increments of $\tau$ correspond to equal opportunities for
mutation and hence exploration of the trait space.
Since $d\tau/dt=\lambda(1-m(t))$ and $m(t)<1$ for $t>0$, the map $t\mapsto\tau(t)$ is strictly increasing. With a slight abuse of notation, writing
$\bpi(\tau)$ for the reparametrized vector $\bpi(t(\tau))$, 
\begin{equation}
\frac{d\bpi}{d\tau} = (P - \Id)\bpi.
\label{dpibydtau}
\end{equation}
In $\tau$-time every individual experiences mutations at unit rate, 
and the trait distribution evolves according to~(\ref{dpibydtau}).
We can show that, in $\BS$, the operator $P-\Id$ is bounded.
$$
\|P\bx\|_1 = \sum_i |(P\bx)_i|
= \sum_i \left|\sum_{k} x_k p_{i-k}\right|
\leq  \sum_i \sum_{k} |x_k| p_{i-k}
=  \sum_{k} |x_k| \sum_i p_{i-k}
=  \sum_{k} |x_k| = \|\bx\|_1,
$$
and clearly $\|\Id \bx\|_1= \|\bx\|_1$. So,
$$
\|(P-\Id)\bx\|_1 = \|P\bx - x\|_1 
\leq \|P\bx\|_1 + \|\bx\|_1 \leq  2\|\bx\|_1,
$$
and therefore $\|P-\Id\|\leq 2$. It follows that
\begin{equation}
\bpi(\tau )
= \exp(\tau(P-\Id))\bpi(0)
=\sum_{n=0}^\infty e^{-\tau} \frac{\tau^n}{n!} P^n \bpi(0),
\label{mexp}
\end{equation}
where here $\exp(\cdot)$ is the operator exponential on $\BS$.
So, $\bpi(\tau)$ is precisely the law at time $\tau$ of a random walk
that makes transitions according to $\bp$ at points of a unit-rate
Poisson process.
Also, we have an explicit expression for~$\tau(t)$. 
Since $\dot\tau=\lambda(1-m)$ and $\dot m/m = \lambda(1-m)-\mu$, we
have that $\dot\tau = \mu+\dot m/m$. Integrating this gives
$$
\tau(t) =\mu t + \log\left( \frac{m(t)}{m_0} \right),
$$
which can be substituted into (\ref{mexp}) to get an explicit
expression for $\bpi(t)$. Notice that since $m(t)\to\rho>0$ when
$\lambda>\mu$, it follows that $\tau(t)=\mu t+O(1)$ as $t\to\infty$, and
hence $\tau(t)\to\infty$. We can now deduce from~(\ref{mexp}) that
$\pi_i(t)\to 0$, as follows.
Since our random walk is irreducible and not positive recurrent, 
we have, for each fixed $i$, that $p^{\,n*}_i \to 0$ as $n\to\infty$,
where $\bp^{\,n*}$ denotes the $n$-fold convolution of $\bp$.
Hence, by dominated convergence,
$$
(P^n\bpi(0))_i = \sum_k \pi_k(0)\,p^{\,n*}_{i-k} \to 0.
$$
Since \eqref{mexp} is a Poisson mixture of the vectors $P^n\bpi(0)$ and
$\tau(t)\to\infty$, it follows that $\pi_i(t)\to 0$ for every fixed $i$
(given $\varepsilon>0$, 
if we choose $n_\varepsilon$ 
such that such that
$(P^n\bpi(0))_i<\varepsilon$ for all $n\geq n_\varepsilon$, 
then $\pi_i(t)\leq
\Pr\{\operatorname{Poisson}(\tau(t))<n_\varepsilon\}+\varepsilon$, and the
probability on the right tends to $0$ as $t\to\infty$.) Note that
$\pi_i(t)\to 0$ even though $\|\bpi(t)\|_1=1$. Finally, since
$x_i(t)=m(t)\pi_i(t)$, and $m(t)\to \rho$, we deduce that $x_i(t)\to 0$
for every $i$. Thus, when $\mu_i\equiv\mu$ and $p_0<1$, we observe
coordinatewise disappearance: $x_i(t)\to 0$ for every fixed $i\in\Z$,
even though the aggregate population density converges to the positive
limit $\rho$ when $\lambda>\mu$. The trait distribution diffuses
indefinitely through $\Z$, with mass escaping to increasingly distant
trait values. This is in sharp contrast to cases considered later, where
$\mu_i\to\infty$; there the random walk ceases to wander freely because
mortality confines it.

The conclusion that $\pi_i(t)\to 0$ for each fixed trait $i$ does not
require any moment assumptions. To obtain quantitative asymptotics for
the rate of convergence, and for the spread of the trait distribution
through $\Z$, we now assume that $\bp$ has mean $\mw:=\sum_{k\in\Z}kp_k$
and finite second moment $\vw:=\sum_{k\in\Z}k^2p_k$. If
$\sigmasquaredw:=\sum_{k\in\Z}(k-\mw)^2p_k$ denotes the variance of a
single mutation displacement, then $\vw=\sigmasquaredw+\mw^2$.
Equation~\eqref{mexp} identifies the proportion vector with the law at
time $\tau$ of the continuous-time random walk on $\Z$ with generator
$P-\Id$ described above. More explicitly, let
$S_\tau=\sum_{j=1}^{N_\tau}Y_j$, where $N_\tau$ is Poisson with mean
$\tau$ and the independent displacements $Y_1,Y_2,\ldots$ have
distribution $\bp$. 
Since the mutation walk is irreducible on $\Z$, and ``Poissonization''
removes any periodicity of the corresponding discrete-time walk, a local
central limit theorem for the continuous-time lattice random walk with
finite second moment (compare Theorems~2.1.3 and~2.3.9
of~\cite{LawlerLimic2010}) gives
$$
\Pr\{S_\tau=i\}
=
\frac{1}{\sqrt{2\pi\vw\tau}}
\exp\!\left\{
-\frac{(i-\mw\tau)^2}{2\vw\tau}
\right\}
+
o\!\left(\tau^{-1/2}\right),
$$
uniformly in $i\in\Z$. Since the random walk has initial distribution
$\bpi(0)$, we have $\pi_i(\tau) =
\sum_{k\in\Z}\pi_k(0)\Pr\{S_\tau=i-k\}$. Convolution with the fixed
initial distribution $\bpi(0)$ does not alter the leading term. Indeed,
this follows by first restricting the convolution to a finite set of
initial states, for which fixed translations do not affect the leading
Gaussian approximation, and then using the arbitrarily small total mass
of the remaining tail, together with the uniform bound
$\sup_{i\in\Z}\Pr\{S_\tau=i\}=O(\tau^{-1/2})$. Returning to
chronological time, we consequently obtain
$$
\pi_i(t)
=
\frac{1}{\sqrt{2\pi\vw\tau(t)}}
\exp\!\left\{
-\frac{(i-\mw\tau(t))^2}{2\vw\tau(t)}
\right\}
+
o\!\left(\tau(t)^{-1/2}\right),
$$
uniformly in $i\in\Z$.

Here the variance parameter is $\vw$, rather than $\sigmasquaredw$,
because the number of mutation steps occurring by evolutionary time
$\tau$ is Poisson with mean $\tau$. Indeed, conditioning on $N_\tau$
gives
$$
\operatorname{Var}(S_\tau)
= \mathbb{E}(N_\tau)\operatorname{Var}(Y_1)
+ \operatorname{Var}(N_\tau)\bigl(\mathbb{E}Y_1\bigr)^2 
= \tau\sigmasquaredw+\tau\mw^2 = \vw\tau.
$$
Furthermore, recall that $\tau(t)\sim\mu t$. If, additionally, $\bp$ is
symmetric, then $\mw=0$ and $\vw=\sigmasquaredw$, implying that, for
every fixed trait value $i$,
$$
\pi_i(t)
\sim
\frac{1}{\sqrt{2\pi\sigmasquaredw\mu t}},
\qquad t\to\infty.
$$
Here we have used $\exp\!\left(-{i^2}/(2\sigmasquaredw\tau(t))\right)\to
1$. Provided that the initial trait distribution has finite variance,
independence of the initial trait and the subsequent mutation
displacements gives $\operatorname{Var}_{\pi(t)}(X) =
\operatorname{Var}_{\pi(0)}(X)+\vw\tau(t)$. It therefore grows linearly
in $t$, and the standard deviation is of order $\sqrt{t}$. For the
discrete Laplace mutation kernel~\eqref{laplace}, used in our numerical
experiments, we have $\mw=0$ and $\sigmasquaredw=\vw=2r/(1-r)^2$. Thus,
both the asymptotic decay of the proportion associated with each fixed
trait and the diffusive spread of the trait distribution are determined
explicitly by the mutation parameter $r$.

\begin{figure}[htbp]
\centering
\includegraphics[height=6cm,width=0.45\textwidth]{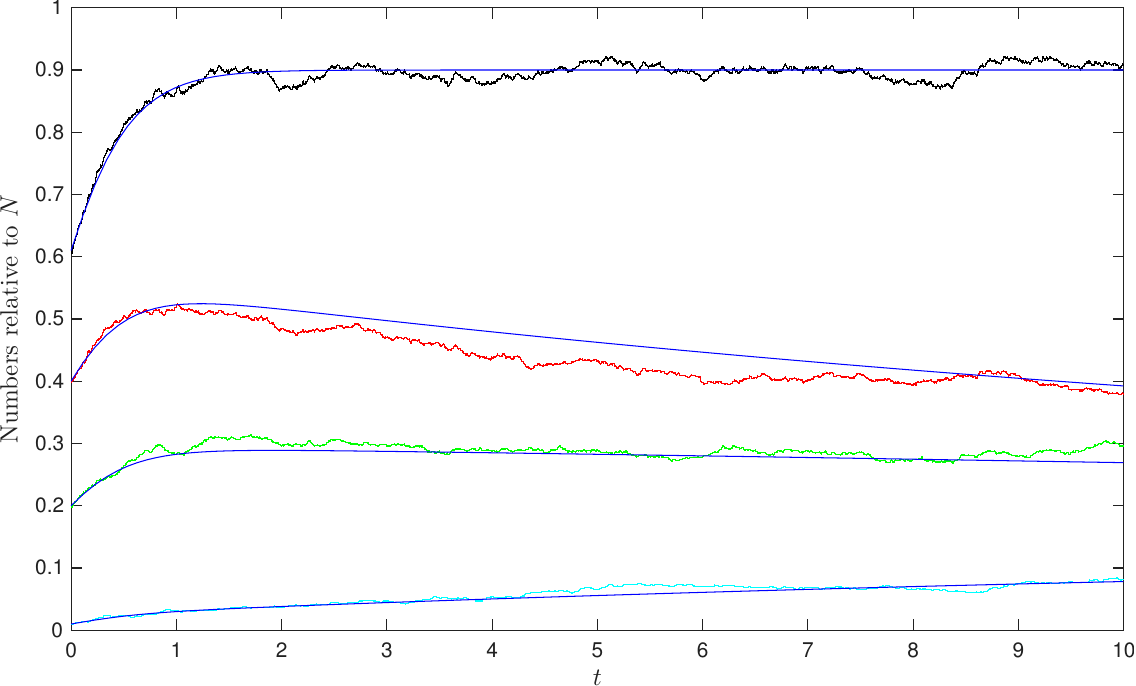}
\includegraphics[
  height=6.3cm,width=0.54\textwidth,
  trim=0cm 3mm 0cm 0cm,
  clip
]{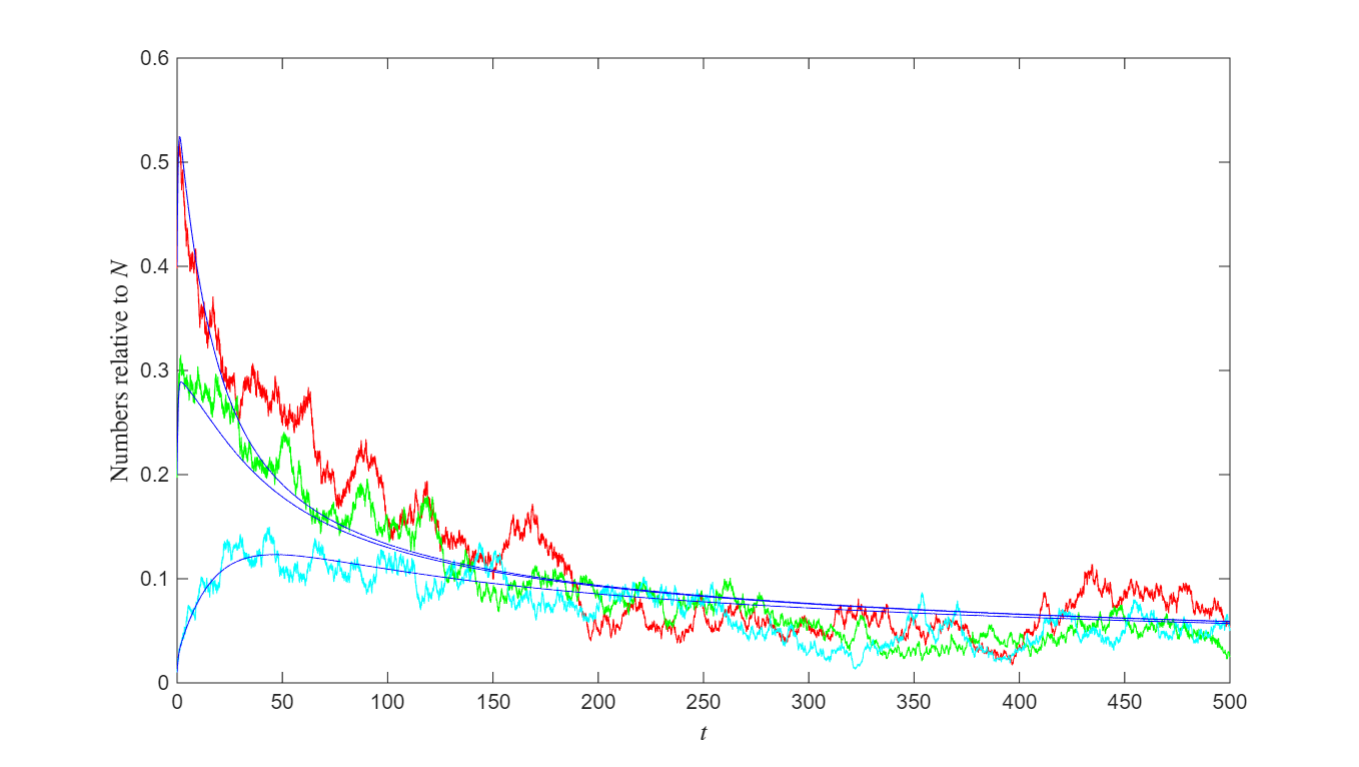}
\caption{
Simulated sample paths of individual trait densities, for traits $-1$
(red), $0$ (green), and~$1$ (cyan), as well as the aggregate density
(black) in the first plot. The corresponding deterministic trajectories
are in blue. The parameter values are $I=200$,
$K=20$, $N=1000$, $\lambda=3$, $\mu=0.3$, and $r=0.4$, with initial
numbers $n_{-1}(0)=400$, $n_{0}(0)=200$, $n_{1}(0)=10$.}
\label{fig:muconstant}
\end{figure}

\medskip
\noindent
{\bf Numerical experiments}.
We simulated the model with $\bp$ following the truncated discrete
Laplace distribution~(\ref{laplacet}) with $K=20$ and $r=0.4$, $N=1000$, 
and with parameters $\lambda=3$, $\mu=3/10$, so
that $\rho=9/10$. The initial numbers for traits~$-1$, $0$, and~$1$ were
$n_{-1}(0)=400$, $n_{0}(0)=200$, $n_{1}(0)=10$, with $0$ for the
remaining traits. The results are shown in
Figure~\ref{fig:muconstant}. The sample paths of trait densities for
traits~$-1$, $0$, and $1$ are shown, together with the corresponding
deterministic trajectories. The first plot includes the aggregate
density (across all traits); notice that even over this relatively short
time scale it settles down to a quasi equilibrium.
The second plot has the same settings, but covers a longer time
interval. The densities associated with the three displayed traits
decline towards $0$. Thus, the numerical results illustrate the
coexistence of demographic persistence, with $m(t)\to\rho=0.9$ in the
deterministic model, and coordinatewise disappearance, with $x_i(t)\to 0$
for every fixed trait~$i$.

\section{$p_0=1$ (no mutation)}
\label{sec:nomutation}

In this section we set aside the standing irreducibility assumption on
the mutation kernel and consider the degenerate case $p_0=1$.
Traits are inherited without change and
$(P\bx)_i=x_i$. The dynamics therefore decouple across traits except
through the total population density. System~\eqref{PKP1} reduces to
\begin{equation}
\dot x_i=\lambda x_i(\rho_i-m),\qquad i\in\Z,
\label{eq:nomutation}
\end{equation} 
where $\rho_i=1-\mu_i/\lambda$. In particular, each coordinate
hyperplane is invariant: if $x_i(0)=0$, then $x_i(t)=0$ for all
$t\geq0$. If $\lambda\leq\mu_i$, then the large-$N$ density of trait $i$
is non-increasing, and is strictly decreasing whenever it is positive.
If $\lambda>\mu_i$, it increases when the \emph{aggregate} density
$m(t)$ is below $\rho_i$ and decreases when it is above $\rho_i$.
Although each coordinate has logistic form, its growth or decline is
determined by $m(t)$, rather than by the density $x_i(t)$ of the trait
itself. Consequently, the traits interact only through competition, 
and the long-term behaviour is governed by differences in mortality.

Let $\bx^*$ be a non-zero equilibrium and put $m^*=\sum_i x_i^*$.
From~\eqref{eq:nomutation}, if $x_i^*>0$, then $m^*=\rho_i$. Consequently,
all traits represented at equilibrium must have the same mortality. For
each value $\mu$ attained by the mortality profile, let 
$J_\mu:=\{i\in\Z:\mu_i=\mu\}$ be the corresponding ``mortality class'', and
set $\rho_\mu:=1-\mu/\lambda$. A non-zero equilibrium supported in
$J_\mu$ can exist only when $\mu<\lambda$, in which case its total mass
must equal $\rho_\mu$. Conversely, every member of 
$$
E_\mu := \left\{ \bx\in E: {\ts\sum_{i\in J_\mu}x_i}=\rho_\mu,
\quad x_i=0\text{ for }i\notin J_\mu \right\}
$$
is an equilibrium. Thus, for each nonempty mortality class $J_\mu$ with
$\mu<\lambda$, $E_\mu$ is precisely the equilibrium set supported in
that class. 
The set also contains equilibria supported on proper subsets of $J_\mu$,
obtained by allowing one or more coordinates within the class to vanish.
Together with the extinction equilibrium $\bzero$,
these sets comprise all equilibria of~\eqref{eq:nomutation}.

There may be many such equilibrium sets, corresponding to the
different mortality classes with mortality below $\lambda$. Their
relevance to a particular trajectory depends strongly on its initial
support. Since traits are inherited without mutation, a trait absent
initially remains absent for all time, and an equilibrium involving that
trait cannot be approached from the given initial state.
Proposition~\ref{minmortality} below identifies the limiting equilibrium
when the minimum mortality among the traits initially present is
attained.

The competitive structure of the system can be seen by comparing the
densities of two traits. Let $i,j\in\Z$ be distinct and suppose that
$x_i(0)>0$ and $x_j(0)>0$. Equation~\eqref{eq:nomutation} implies that
both coordinates remain positive for all $t\geq0$. Hence
$$
\frac{\dot x_i}{x_i}=\lambda(\rho_i-m)
\quad\text{and}\quad
\frac{\dot x_j}{x_j}=\lambda(\rho_j-m).
$$
Subtracting gives
$$
\frac{d}{dt}\log\frac{x_i}{x_j}
=\lambda(\rho_i-\rho_j)
=-(\mu_i-\mu_j),
$$
and therefore
\begin{equation}
\frac{x_i(t)}{x_j(t)}
=
\frac{x_i(0)}{x_j(0)}
e^{\lambda(\rho_i-\rho_j)t}
=
\frac{x_i(0)}{x_j(0)}
e^{-(\mu_i-\mu_j)t}.
\label{ratio}
\end{equation}
Thus, if $\rho_i>\rho_j$, equivalently $\mu_i<\mu_j$, then
$x_i(t)/x_j(t)\to\infty$ as $t\to\infty$, whereas if
$\rho_i<\rho_j$, then $x_i(t)/x_j(t)\to 0$. If
$\rho_i=\rho_j$, the ratio $x_i(t)/x_j(t)$ remains constant.
Consequently, among traits represented initially, those with lower
mortality increase exponentially relative to those with higher
mortality, while the relative abundances of traits belonging to the
same mortality class are preserved.

In the case where $\mu_i$ achieves a minimum 
among traits represented in the initial support, we have the following
result. It says that (i) the total mass concentrates on the mortality class with
minimal mortality among traits initially present, (ii) within that class
the initial proportions are preserved, and (iii) the limit point
is determined by the initial trait density.

\begin{proposition}
\label{minmortality}
Suppose that there is a minimum mortality 
among traits represented in the initial support:
$\mu^*=\min\{\mu_i: x_i(0)>0\}$ and suppose that $\lambda>\mu^*$.
Let $J^*=\{i: x_i(0)>0,\, \mu_i=\mu^*\}$ be the corresponding mortality
class. Then, for every $i\in \Z$, $x_i(t) \to x_i^*$, where
$$
x_i^*= 
\begin{cases}
\begin{aligned}
&{\ds \rho^* \frac{x_i(0)}{\sum_{j\in J^*} x_j(0)}}
           &\text{if } i\in J^*\\[1ex]
&{\ds 0}   &\text{if } i\notin J^*,
\end{aligned}
\end{cases}
$$
and $\rho^*=1-\mu^*/\lambda$.
\end{proposition}

Proposition~\ref{minmortality} shows that, in the absence of mutation,
selection acts as a ``winner-takes-all'' mechanism across mortality
classes: only the class with the lowest mortality (with members present
initially) survives, while the relative proportions within that class
remain unchanged.

\medskip
\noindent
\emph{Remark}.
The ratio formula shows that, among traits represented initially,
those with larger mortality are eliminated exponentially relative to
those with smaller mortality. This is consistent with the linearization
at an equilibrium $\bx^*\in E_\mu$. Since $m(\bx^*)=\rho_\mu$, the
Fr\'echet derivative is given by
$$
(DF(\bx^*)\bh)_i
=
\lambda(\rho_i-\rho_\mu)h_i
-\lambda x_i^*\sum_k h_k.
$$
Perturbations supported in $J_\mu$ that preserve total mass are
annihilated by the linearization and therefore have zero growth rate.
These perturbations correspond to redistributions of mass within the
equilibrium set $E_\mu$. The direction $\bx^*$, which changes the total
mass while preserving the equilibrium trait proportions, has eigenvalue
$-\lambda\rho_\mu$.
For a trait $i\notin J_\mu$, the corresponding coordinate of the
linearized equation satisfies
$\dot h_i=\lambda(\rho_i-\rho_\mu)h_i$. Thus, a newly introduced trait
has positive invasion growth rate precisely when its mortality is lower
than that of the resident mortality class. Consequently, $E_\mu$ can
be locally attracting only if $\rho_i<\rho_\mu$ for every
$i\notin J_\mu$, equivalently if every trait outside $J_\mu$ has
strictly larger mortality.

\medskip
In the case where $\mu_i$ does not achieve a minimum, and
$\mu^*=\inf\{\mu_i: x_i(0)>0\}$, there is no corresponding mortality
class $J^*$. Nevertheless, if $\mu_i>\mu^*$ for some trait $i$, then, by
the definition of infimum, there exists a trait $j$ with
$\mu^*<\mu_j<\mu_i$. The ratio formula~(\ref{ratio}) therefore implies that
$x_i(t)/x_j(t)\to 0$ as $t\to\infty$. Thus, every trait whose mortality
exceeds $\mu^*$ becomes negligible relative to traits with mortality
closer to $\mu^*$. This suggests a qualitatively different asymptotic
regime from that identified in Proposition~\ref{minmortality}. Since the
mortality infimum is not attained, there is no equilibrium set
associated with $\mu^*$. Instead, one expects the population mass to
become concentrated on traits whose mortality approaches $\mu^*$,
without any single trait or mortality class capturing a positive
limiting fraction of the population. We might therefore expect that,
whenever $\lambda>\mu^*$, $m(t)\to \rho^*:=1-{\mu^*}/{\lambda}$, while
$x_i(t)\to 0$ for every fixed trait $i$. In such a scenario the
aggregate population density converges, but the
trait densities do not converge to non-trivial equilibria.

\begin{proposition}
\label{notattained}
Suppose that $p_0=1$, $\bx(0)\neq\bzero$, and
$\mu^*:=\inf\{\mu_i:x_i(0)>0\}$
is finite but is not attained on the initial support. If
$\lambda>\mu^*$, then $m(t)\to\rho^*:=1-\mu^*/\lambda$.
Furthermore, for every $i\in\Z$,
$\pi_i(t)=x_i(t)/m(t)\to0$, and hence $x_i(t)\to0$.
\end{proposition}

\medskip
\noindent
{\bf Numerical experiments}.
We simulated the model with $p_0=1$. For our first group of experiments,
we set $I=50$, $N=9999$, $\lambda=3$,
$\mu_{-1}=\mu_0=\mu_1=\mu_A:=1$, and $\mu_i=\mu_B:=2$ otherwise.
Thus, traits $-1$, $0$, and $1$ form the mortality class corresponding
to the lower mortality $\mu_A$. In an obvious notation, let
$\rho_A:=1-\mu_A/\lambda=2/3$ and
$\rho_B:=1-\mu_B/\lambda=1/3$.

\begin{enumerate}
\item[(i)]
We set $n_i(0)=111$ for $i\in\{-30,\ldots,30\}$, and $n_i(0)=0$
otherwise, so that $\sum_i n_i(0)=111\times61=6771<N$.
The results are shown in Figure~\ref{fig:FiguresE}. The minimum
mortality among the traits represented initially is $\mu_A$, and it is
attained precisely by traits $-1$, $0$, and $1$.
Proposition~\ref{minmortality} therefore predicts that the deterministic
densities of these three traits each converge to
$\rho_A/3=2/9$, while all other trait densities converge to $0$.

In the stochastic simulation, the traits outside this
minimum-mortality class disappeared quickly, at around $t=12$. This is
not shown directly, but is reflected in the recovery of the aggregate
density after its initial decline. The remaining three traits then
entered a long-lived quasi-stationary regime. Trait~$1$ disappeared at
around $t=1465$, followed by trait~$-1$ at around $t=10674$, after which
trait~$0$ was the only remaining trait.

\item[(ii)]
Next, we set
$n_{-1}(0)=n_0(0)=n_1(0)=2N/9=2222$, and $n_i(0)=0$ otherwise.
Thus, $x_i(0)=2/9$ for $i\in\{-1,0,1\}$ and
$m(0)=2/3=\rho_A$, so the deterministic initial state belongs to the
equilibrium set $E_{\mu_A}$. The results are shown in
Figure~\ref{fig:FiguresF}(a).

As expected, the deterministic densities remain constant:
$x_i(t)=2/9$ for $i\in\{-1,0,1\}$ and $m(t)=2/3$ for all $t\geq0$.
The stochastic densities fluctuate around these equilibrium values over
the interval displayed. Since traits are inherited without mutation,
all traits absent initially remain absent.

\item[(iii)]
Finally, we set
$n_{-1}(0)=n_0(0)=n_1(0)=N/9=1111$, and $n_i(0)=0$ otherwise.
Thus, $x_i(0)=1/9$ for $i\in\{-1,0,1\}$ and
$m(0)=1/3=\rho_B<\rho_A$. The initially represented traits all have
mortality $\mu_A$, so the deterministic densities increase from $1/9$
towards $2/9$, and the trajectory approaches an equilibrium in
$E_{\mu_A}$. The results are shown in Figure~\ref{fig:FiguresF}(b).
The stochastic densities likewise move towards a quasi-stationary regime
centred near the corresponding equilibrium values.
\end{enumerate}

\begin{figure}[htbp]
\centering
\includegraphics[width=0.8\textwidth]{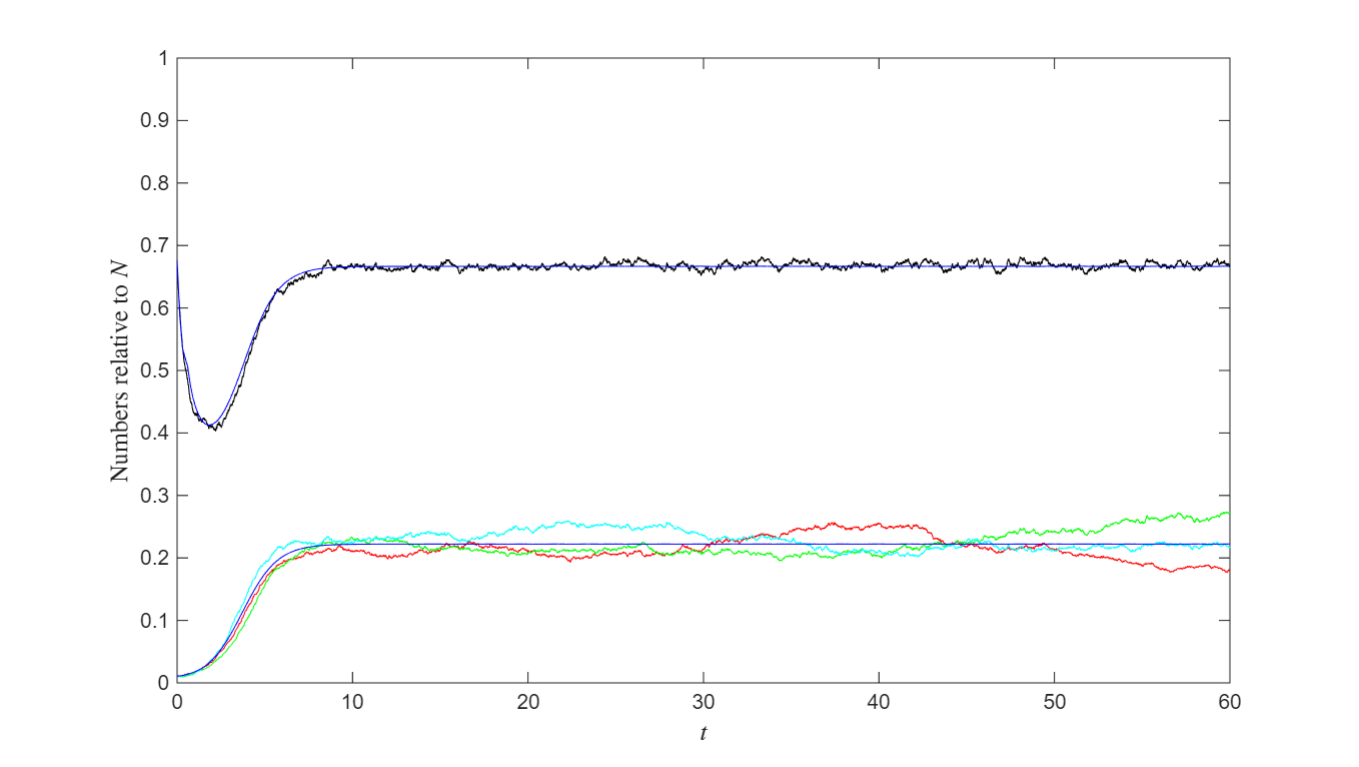}
\caption{Simulated sample paths of the densities of traits $-1$ (red),
$0$ (green), and $1$ (cyan), together with the aggregate density
(black). The corresponding deterministic trajectories are shown in
blue. The parameter values are given in the text. Traits outside the
minimum-mortality class disappear quickly, after which a long-lived
quasi-stationary regime is observed for traits $-1$, $0$, and $1$.
Proposition~\ref{minmortality} predicts that their deterministic
densities each converge to $2/9$.}
\label{fig:FiguresE}
\end{figure}

\begin{figure}[htbp]
\centering
\includegraphics[width=0.49\textwidth]{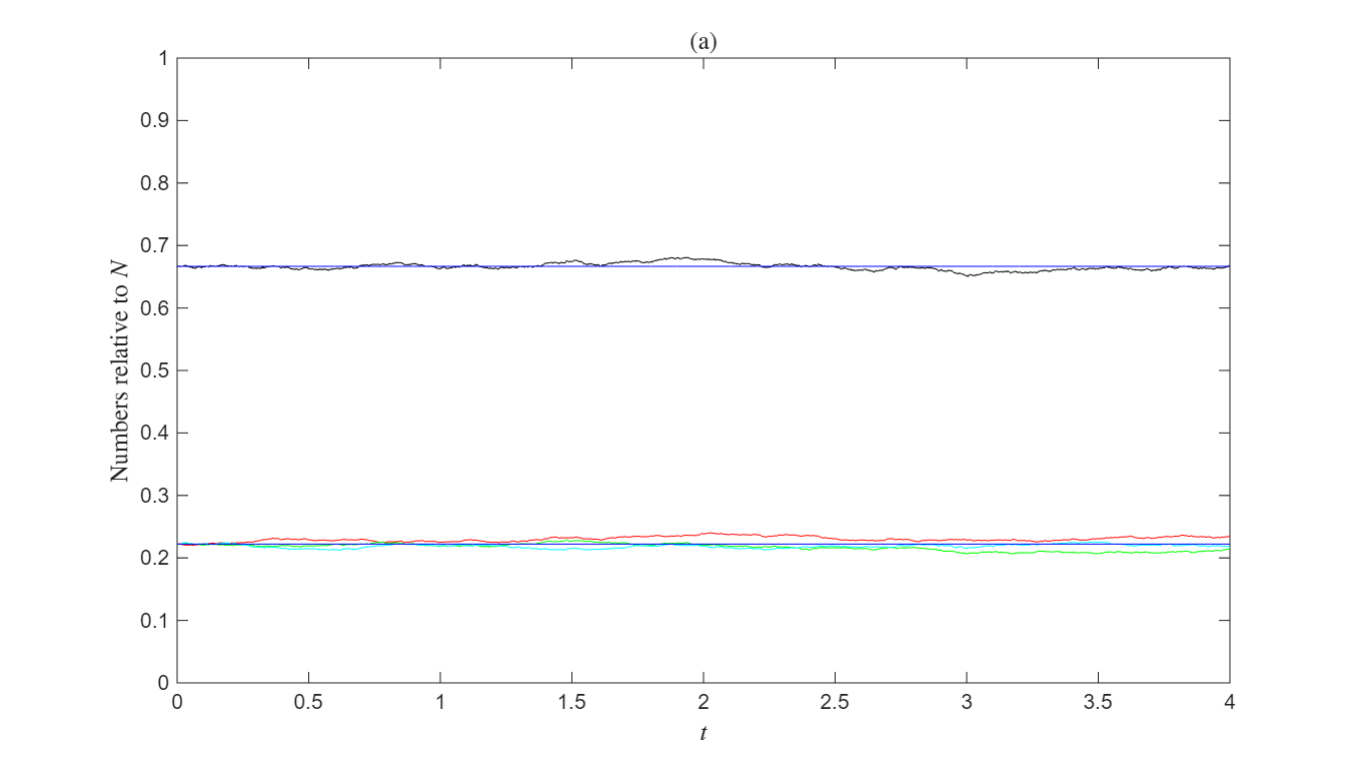}
\includegraphics[width=0.49\textwidth]{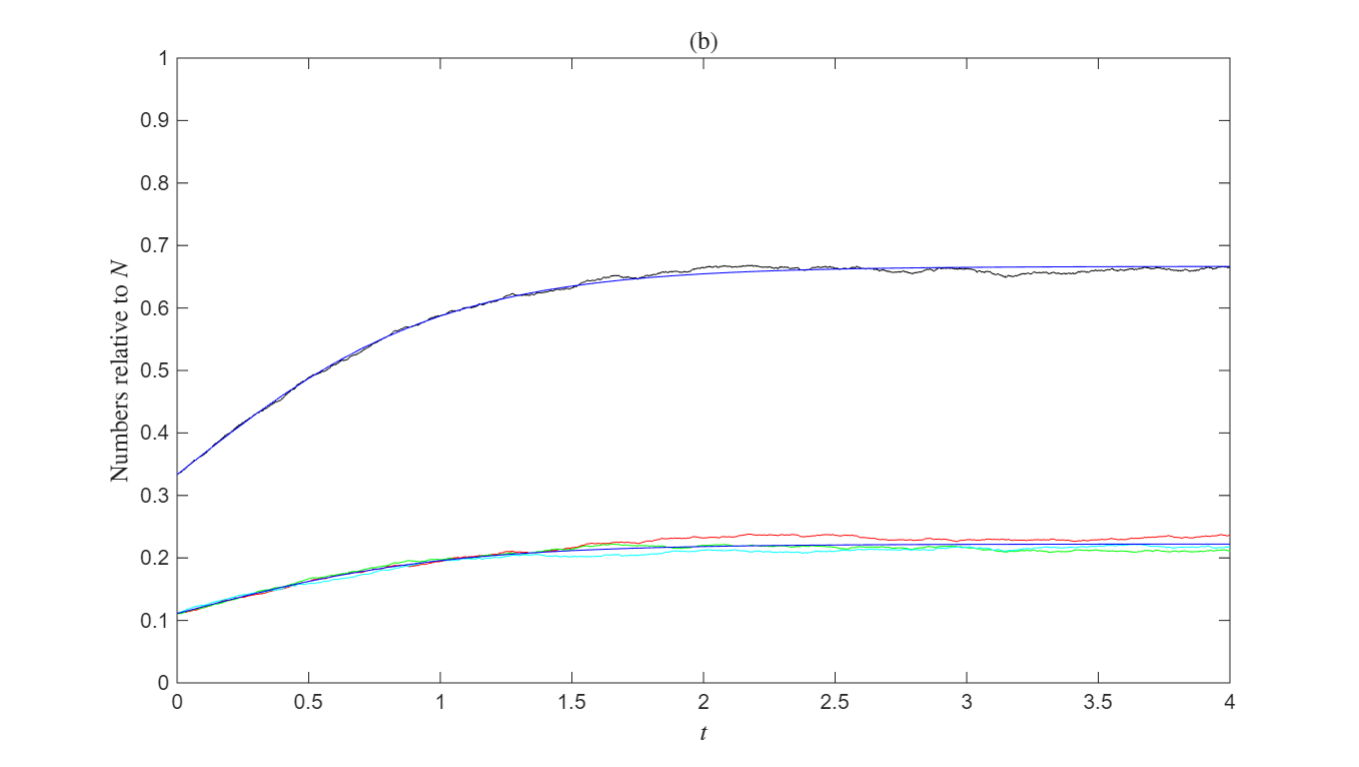}
\caption{Simulated sample paths of the densities of traits $-1$ (red),
$0$ (green), and $1$ (cyan), together with the aggregate density
(black). The corresponding deterministic trajectories are shown in
blue. The parameter values are given in the text.
In panel~(a), the deterministic trajectory remains at an equilibrium
in $E_{\mu_A}$, while the stochastic densities fluctuate around the
corresponding equilibrium values. In panel~(b), the deterministic and
stochastic densities move from lower initial values towards those
associated with $E_{\mu_A}$.}
\label{fig:FiguresF}
\end{figure}

For our second group of experiments, we set $I=50$, $N=2000$,
$\lambda=4$, $\mu_0=\mu_A:=1$,
$\mu_{-1}=\mu_1=\mu_B:=2$, and $\mu_i=\mu_C:=3$ otherwise.
There are therefore three mortality classes, with corresponding
positive equilibrium masses
$\rho_A:=1-\mu_A/\lambda=3/4$,
$\rho_B:=1-\mu_B/\lambda=1/2$, and
$\rho_C:=1-\mu_C/\lambda=1/4$.

\begin{enumerate}
\item[(iv)]
We set $n_{-2}(0)=780$, $n_{-1}(0)=600$, $n_1(0)=360$, and
$n_2(0)=240$, with $n_i(0)=0$ otherwise. In particular,
$n_0(0)=0$. Thus,
$x_{-2}(0)=0.39$, $x_{-1}(0)=0.30$,
$x_1(0)=0.18$, and $x_2(0)=0.12$.
The results are shown in Figure~\ref{fig:FiguresG}(a).

Since traits are inherited without mutation, trait~$0$, which is absent
initially, remains absent despite having the lowest mortality overall.
Among the traits represented initially, traits $-1$ and $1$ have the
minimum mortality $\mu_B$. Proposition~\ref{minmortality} therefore
predicts that traits $-2$ and $2$ disappear, while the deterministic
trajectory approaches an equilibrium in $E_{\mu_B}$. More precisely,
the limiting densities are
$$
x_{-1}^*
=
\rho_B\frac{x_{-1}(0)}{x_{-1}(0)+x_1(0)}
=
\frac{5}{16}
=
0.3125,
\qquad
x_1^*
=
\rho_B\frac{x_1(0)}{x_{-1}(0)+x_1(0)}
=
\frac{3}{16}
=
0.1875.
$$
In the stochastic simulation, traits $-2$ and $2$ disappeared quickly,
after which traits $-1$ and $1$ entered a quasi-stationary regime near
these deterministic equilibrium values.

\item[(v)]
Finally, we altered only the initial density of trait~$0$, setting
$x_0(0)=1/200$, corresponding to $n_0(0)=10$, while retaining the
other initial densities from experiment~(iv). The results are shown in
Figure~\ref{fig:FiguresG}(b).

The outcome is entirely different. Trait~$0$, which has the smallest
mortality, is now represented initially. Proposition~\ref{minmortality}
therefore predicts convergence of the deterministic trajectory to the
single-trait equilibrium in $E_{\mu_A}$, with
$x_0^*=\rho_A=3/4$ and $x_i^*=0$ for $i\neq0$.
In the stochastic realization, trait~$1$ disappeared at
$t=11.2559$, followed by trait~$-1$ at $t=11.8544$, leaving trait~$0$
as the only remaining trait. Its density subsequently entered a
quasi-stationary regime near the deterministic equilibrium value
$\rho_A=3/4$.
\end{enumerate}

\begin{figure}[htbp]
\centering
\includegraphics[width=0.49\textwidth]{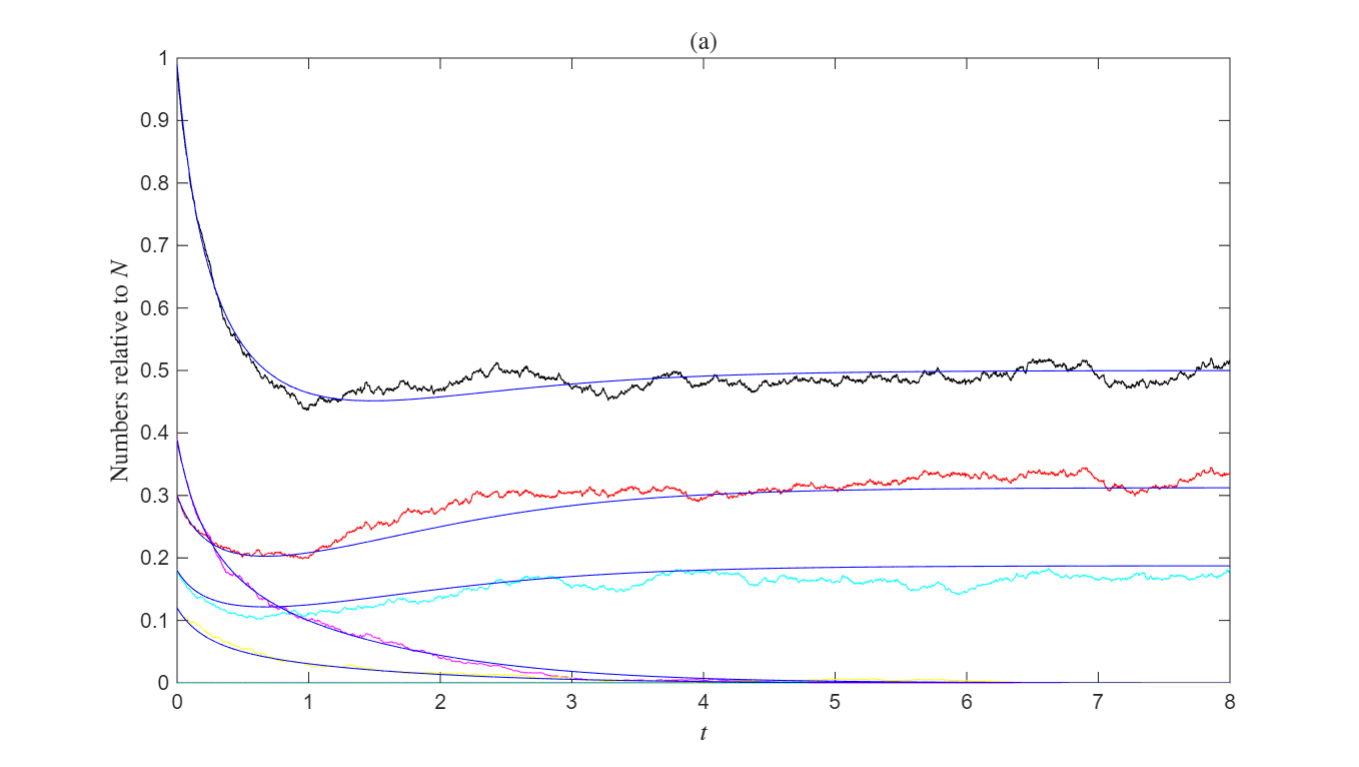}
\includegraphics[width=0.49\textwidth]{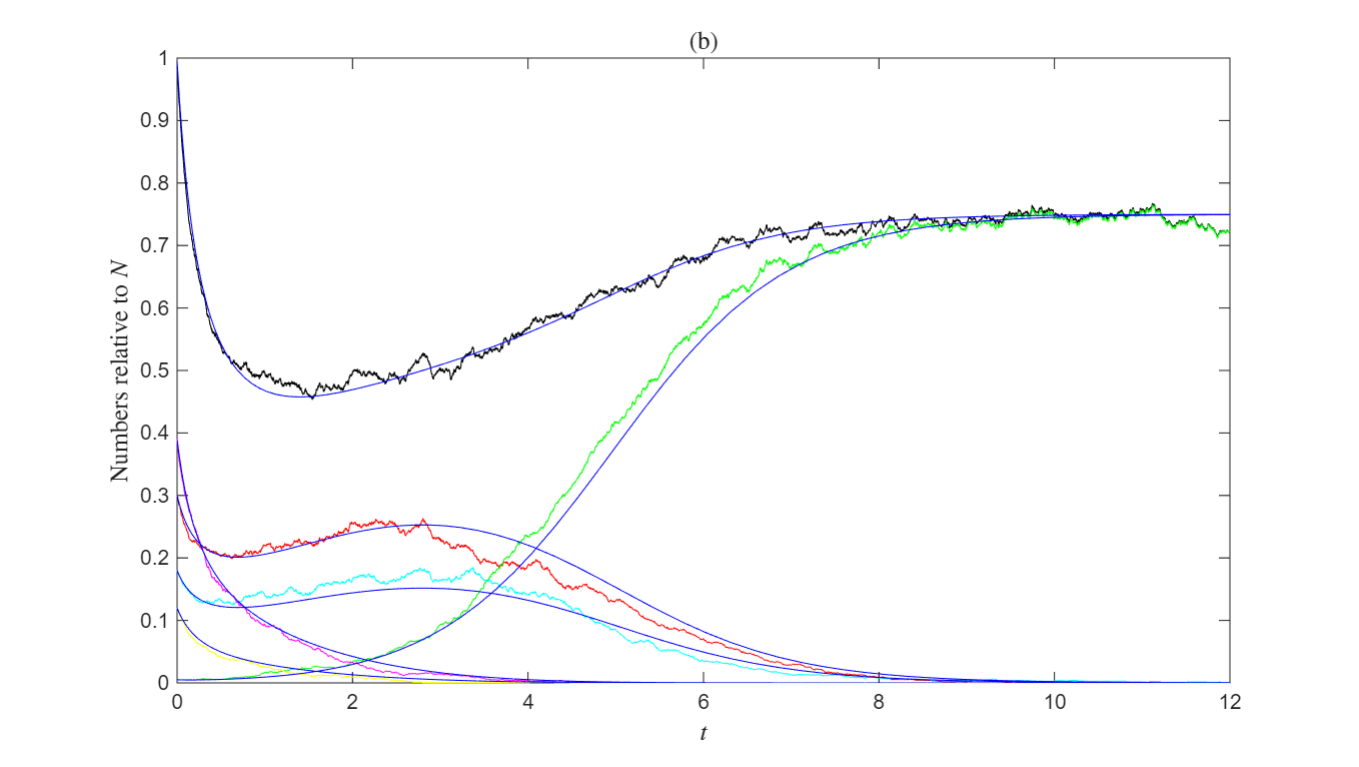}
\caption{Simulated sample paths of the densities of traits $-2$
(magenta), $-1$ (red), $0$ (green), $1$ (cyan), and $2$ (yellow),
together with the aggregate density (black). The corresponding
deterministic trajectories are shown in blue. The parameter values are
given in the text.
In panel~(a), trait~$0$ is absent initially and therefore remains
absent. The deterministic trajectory approaches an equilibrium in
$E_{\mu_B}$ supported on traits $-1$ and $1$, while the stochastic
densities enter a quasi-stationary regime near the corresponding
equilibrium values.
In panel~(b), a small initial population of trait~$0$ is present.
Because trait~$0$ has the lowest mortality, the deterministic trajectory
approaches the single-trait equilibrium in $E_{\mu_A}$, while in the
stochastic realization all other traits disappear.}
\label{fig:FiguresG}
\end{figure}

\FloatBarrier

\section{The general case}
\label{sec:generalcase}

Guided by what we have learned already about the long-term behaviour 
of our model in some special cases, let us consider the
general case where $(\mu_i)$ are not all the same and $p_0<1$. 
Our system of ODEs~(\ref{PKP1}) can be written
$$
\dot \bx(t) = \lambda (1-m(t)) P\bx(t) - D\bx(t),
$$
where $(D\bx)_i=\mu_i x_i$, 
$(P\bx)_i=\sum_k x_k p_{i-k}$, and $m(t)=\sum_i x_i(t)$.

Consider first the equilibrium $\bx=\bzero$. Its formal linearization
is governed by the operator $\lambda P-D$, with domain $\D(D)$.
Although the stability of the extinction equilibrium is naturally
related to the spectral properties of this operator, we shall not
analyse it directly. Instead, we introduce the bounded positive
operator $L=D^{-1}P$, whose spectral properties determine the existence
and form of non-zero equilibria.

A non-zero equilibrium $\bx^*$ satisfies $\lambda(1-m^*)P\bx^*=D\bx^*$,
with $m^*:=m(\bx^*)=\sum_i x_i^*$. Since $P\bx^*\in\BS$, the equilibrium
equation shows that $D\bx^*\in\BS$, and hence $\bx^*\in\D(D)$.
Notice also that $m^*<1$. Indeed, if $m^*=1$, then the equilibrium
equation gives $D\bx^*=0$, and hence $\bx^*=0$, since $\mu_i>0$ for
every $i\in\Z$.
Suppose now that $\muinf:=\inf_i\mu_i>0$, and define
$L:=D^{-1}P$. Applying $D^{-1}$ to the equilibrium equation gives
\begin{equation}
L\bx^*
=
\frac{1}{\lambda(1-m^*)}\bx^*,
\label{eq:equilibriumeigenvalue}
\end{equation}
where
\begin{equation}
(L\bx)_i
=
\frac{1}{\mu_i}\sum_kx_kp_{i-k},
\qquad i\in\Z.
\label{L}
\end{equation}

Recalling that $K$ is the positive cone in $\BS$, every non-zero
equilibrium is an eigenvector of $L$ belonging to
$K\setminus\{\bzero\}$, with positive eigenvalue $1/\{\lambda(1-m^*)\}$.
Conversely, suppose that $\bphi\in K\setminus\{\bzero\}$ is an
eigenvector of $L$ with eigenvalue $r>0$. If $\lambda r>1$, define
$$
\bx^*
=
\left(1-\frac{1}{\lambda r}\right)
\frac{\bphi}{\sum_i\phi_i}.
$$
Then $m^*=1-1/(\lambda r)$, and hence
$\lambda(1-m^*)=1/r$. Since $L\bphi=r\bphi$, we have
$P\bphi=rD\bphi$, and therefore
$\lambda(1-m^*)P\bx^*=D\bx^*$. Thus, $\bx^*$ is a non-zero
equilibrium. Its trait proportions are given by
$\pi_i^*=\phi_i/\sum_j\phi_j$.

It is worth contrasting $L$ with the mutation operator $P$ appearing in
Section~\ref{sec:noselection}. There, the underlying irreducible random
walk on $\Z$ was not positive recurrent, and so $P$ admitted no
invariant probability distribution. Consequently, mutation alone could
not sustain a stationary trait distribution. The operator $L=D^{-1}P$,
however, incorporates the effects of selection through mortality factors
$1/\mu_i$. As we shall see, when $\mu_i\to\infty$ these factors suppress 
the tails of
the trait space strongly enough for $L$ to be a compact operator. Combined
with irreducibility of the mutation kernel, this will lead to a
distinguished positive eigenvector.

Our strategy is to apply the {\em Kre\u{\i}n--Rutman theorem\/} to $L$.
It plays a role analogous to that of the Perron--Frobenius theorem for
matrices. Roughly speaking, it asserts that a compact positive operator
possesses a positive principal eigenvalue equal to its spectral radius
$r(L)$, together with a strictly positive eigenvector. We shall see that
this eigenpair determines the equilibrium trait distribution, and yields
a natural threshold condition for the existence of a non-zero
equilibrium, which under the assumptions below will be strictly
positive.

The operator $L$ is clearly linear and positive. It is also bounded.
Indeed, for $\bx\in\BS$, the triangle inequality and the positivity of
the mutation kernel give
$$
\|L\bx\|_1 \leq \frac{1}{\muinf}\sum_i\sum_k|x_k|p_{i-k}
= \frac{1}{\muinf}\|\bx\|_1.
$$
Thus $\|L\|\leq1/\muinf$.
To apply the Krein--Rutman theorem, we also require $L$ to be compact.
Under confining mortality, the factors $1/\mu_i$ suppress the tails of
$P\bx$ sufficiently strongly to give this property.

\begin{lemma}
\label{lemma:compact}
Suppose that $\muinf>0$ and that
$\mu_i\to\infty$ as $|i|\to\infty$. Then the operator $L$ given
by~\eqref{L} is compact on $\BS$.
\end{lemma}

The classical Krein--Rutman theorem (Theorem~19.2 of \cite{Dei85})
yields a positive eigenvector when $r(L)>0$, but does not by itself
provide all the conclusions needed below. Moreover, the usual strongly
positive version is unavailable because the positive cone of $\BS$ has
empty interior. We therefore use the ideal-irreducible version of Chang,
Wang and Wu~\cite{CWW20}.
We first observe that the positive cone $K$ is {\em total\/},
meaning that $\overline{K-K}=\BS$. In fact, $K-K=\BS$, since every
$\bx\in\BS$ admits the decomposition $\bx=\bx^+-\bx^-$, where
$x_i^+=\max\{x_i,0\}$ and $x_i^-=\max\{-x_i,0\}$. Both $\bx^+$ and
$\bx^-$ belong to $K$.

The classical theorem guarantees existence of a positive eigenvector
once positivity of $r(L)$ is known. The strong version of the
Kre\u{\i}n--Rutman theorem, Theorem~1.2 of Chang, Wang and
Wu~\cite{CWW20}, applies to compact strongly positive operators when the
positive cone has non-empty interior. Here, however,
$\operatorname{int}(K)=\varnothing$, so this version cannot be applied.
Indeed, if $\bx\in K$ and $\varepsilon>0$, then $x_i\to 0$ as
$|i|\to\infty$, so there exists $j\in\Z$ such that $x_j<\varepsilon/2$.
The vector $\by:=\bx-(x_j+\varepsilon/2)\be_j$ does not belong to $K$,
while $\|\by-\bx\|_1<\varepsilon$. Hence no element of~$K$ is an
interior point. We therefore appeal instead to the more
delicate form of the theorem,
Theorem~1.5 of~\cite{CWW20}. For this we first need to observe that
$\BS$ is a Banach {\em lattice\/} (it is equipped with the
ordering $\bx\leq \by$ if and only if $x_i\leq y_i$ for all $i\in \Z$),
and we need to introduce the notions of {\em quasi-interior\/} and {\em
ideal-irreducibility\/}, both of which have natural biological
interpretations in the present context.

Whilst the positive cone $K$ has empty interior, its
{\em quasi-interior\/} $\qit(K)$ is non-empty. Recall that
$\bx\in K$ is a quasi-interior point if the principal ideal generated
by $\bx$ is dense in $\BS$. In the present setting,
$$
\bx\in\qit(K)
\quad\Longleftrightarrow\quad
x_i>0\text{ for every }i\in\Z.
$$
Indeed, if $x_j=0$ for some $j$, then every vector in the principal
ideal generated by $\bx$ has $j$-th coordinate equal to~$0$, so that ideal
cannot be dense in $\BS$. Conversely, if $x_i>0$ for every~$i$, then
the principal ideal generated by $\bx$ contains every finitely
supported sequence, and is therefore dense in $\BS$.
Thus, in the present setting, the conclusion
$\bx\in\qit(K)$ of Theorem~1.5 of~\cite{CWW20} means that every trait
is represented with strictly positive density. Accordingly, we shall
call an equilibrium $\bx^*$ strictly positive if $x_i^*>0$ for every
$i\in\Z$, equivalently if $\bx^*\in\qit(K)$.

A positive operator on a Banach lattice is called {\em
ideal-irreducible\/} if it leaves no non-trivial closed ideal invariant.
The closed ideals in $\BS$ are precisely the coordinate ideals
$$
 I_A = \{\bx\in\BS:x_i=0\text{ whenever }i\notin A\},
\qquad A\subseteq\Z.
$$
Indeed, if $I$ is a closed ideal and $A:=\{i\in\Z:x_i\neq0\text{ for
some }\bx\in I\}$, then $\be_i\in I$ for every $i\in A$. Thus $I$
contains every finitely supported sequence with support in $A$, and
hence, by closedness, $I_A\subseteq I$. The reverse inclusion follows
from the definition of $A$, so $I=I_A$. Thus, ideal-irreducibility means
that there is no non-empty proper set of traits whose associated
coordinate ideal is invariant under the action of the operator.

\begin{lemma}
\label{lemma:ideal}
If the mutation kernel $\bp$ is irreducible, then $L$ is
ideal-irreducible.
\end{lemma}

Ideal irreducibility means that no non-trivial subpopulation can evolve
independently of the remainder of the trait space. Under repeated
application of $L$, every nonempty collection of traits is eventually
coupled through mutation to the rest of the trait space, so no proper
closed ideal is invariant.

The operator $L$ is positive, and Lemmas~\ref{lemma:compact} and
\ref{lemma:ideal} verify its compactness and
ideal-irreducibility. Theorem~1.5 of Chang, Wang and
Wu~\cite{CWW20} therefore gives the following result.

\begin{theorem}
\label{CWW}
Suppose that $\muinf>0$, that $\mu_i\to\infty$ as
$|i|\to\infty$, and that~$\bp$ is irreducible. Then:
\begin{enumerate}
\item[(i)] $r(L)>0$;
\item[(ii)] $r(L)$ is an algebraically simple eigenvalue of $L$ with
an eigenvector $\bphi\in\qit(K)$;
\item[(iii)] $L$ has no real eigenvalues other than $r(L)$ associated
with eigenvectors in $K\setminus\{\bzero\}$.
\end{enumerate}
\end{theorem}

Thus, the principal eigenspace of $L$ is one-dimensional and is
generated by an eigenvector with strictly positive entries.
Now for the main result of this section.

\begin{theorem}
\label{critical}
Suppose that $\muinf>0$, that $\mu_i\to\infty$ as
$|i|\to\infty$, and that $\bp$ is irreducible. Then:
\begin{enumerate}
\item[(i)] if $\lambda r(L)\leq1$, no non-zero equilibrium exists;
\item[(ii)] if $\lambda r(L)>1$, there exists a unique non-zero
equilibrium, and it is strictly positive.
\end{enumerate}
\end{theorem}

Thus $\lambda r(L)=1$ is the equilibrium threshold: a non-zero
equilibrium exists precisely when $\lambda r(L)>1$, and in that case it
is unique and strictly positive.

Although Theorem~\ref{critical} establishes existence and uniqueness of
a strictly positive equilibrium, it does not address its stability.
Numerical calculations suggest that, whenever $\lambda r(L)>1$, this
equilibrium is locally exponentially stable. In particular, for the
finite-dimensional truncations discussed in
Section~\ref{sec:discussion}, the eigenvalue of the linearization with
largest real part is negative and appears to remain bounded away from
the imaginary axis in the left half-plane as the truncation size
increases.

\medskip
\noindent
{\bf Numerical experiments}.
As with earlier experiments, calculations are necessarily performed on
finite-dimensional truncations of $L$. Since our primary objective here
is to estimate the threshold parameter $r(L)$, we do not pursue a
detailed analysis of the associated eigenvectors. Instead, we shall
content ourselves with the following result, which establishes
convergence of the truncated operators in operator norm and convergence
of their spectral radii to $r(L)$. This provides the principal
justification for the numerical investigations which follow.

\begin{proposition}
\label{truncation}
Suppose that $\muinf>0$, that $\mu_i\to\infty$ as
$|i|\to\infty$, and that the mutation kernel~$\bp$ is irreducible.
Let $Q_I$ be the coordinate projection
$$
(Q_I\bx)_i=
\begin{cases}
x_i, & |i|\leq I,\\
0, & |i|>I,
\end{cases}
$$
and let $L_I:=Q_ILQ_I$. Then $\|L_I-L\|\to 0$ as $I\to\infty$.
Furthermore, $r(L_I)\to r(L)$.
\end{proposition}

\noindent
\emph{Remark}.
It is also natural to expect convergence of suitably
normalized Perron--Frobenius eigenvectors of $L_I$ to the principal
eigenvector of $L$. Such a result would further justify the use of
finite-dimensional truncations when approximating equilibrium trait
distributions. Questions of this type belong to the theory of compact
and collectively-compact operator approximation (see, for example,
Theorem~4.21 of Anselone~\cite{Anselone71}), but we will not pursue them here.

\medskip
The numerical experiments illustrate two related contrasts.
Figures~\ref{fig:FiguresA} and~\ref{fig:FiguresB} show how
demographic behaviour under bounded mortality depends on the birth rate,
while Figures~\ref{fig:FiguresB} and~\ref{fig:FiguresC} contrast
continuing evolutionary drift under bounded mortality with localization
under confining mortality. In the latter case, Theorem~\ref{critical}
gives a unique non-zero equilibrium whenever $\lambda r(L)>1$.

For simplicity, we use the notation $\lambda_{c,I}$ for the threshold
computed from the numerical matrix, which uses both trait-space
truncation at $I$ and mutation-displacement truncation at $K=20$. For
the finite-dimensional matrices used in our computations, these
thresholds were approximately $\lambda_{c,I}=0.10002$ for the
experiments illustrated in Figures~\ref{fig:FiguresA}
and~\ref{fig:FiguresB}, and $\lambda_{c,I}=1.033$ for the experiment
illustrated in Figure~\ref{fig:FiguresC}.
In the confining case of Figure~\ref{fig:FiguresC},
Proposition~\ref{truncation} shows that the spectral radii of the exact
finite sections $Q_I LQ_I$ converge to $r(L)$. Since $r(L)>0$, it
follows that
$$
\lambda_{c,I}=\frac{1}{r(L_I)} \to \frac{1}{r(L)}.
$$
The numerical matrices additionally use the truncated mutation
kernel~\eqref{laplacet}; as noted earlier, for $K=20$ and $r=0.4$,
the omitted mutation probability is less than $10^{-8}$, so this
additional approximation is expected to be negligible on the scale
reported here.
Thus the computed value $\lambda_{c,I}=1.033$ provides
a numerical approximation to the equilibrium threshold $1/r(L)$ of the
infinite-dimensional model. The parameter value $\lambda=10$ lies well
above this threshold, consistently with the persistent localized
behaviour observed in the simulations and with the existence of a
strictly positive equilibrium trait distribution.

The interpretation is different for the bounded directional profiles in
Figures~\ref{fig:FiguresA} and~\ref{fig:FiguresB}.
Proposition~\ref{truncation} does not apply here, and the
finite-dimensional value $\lambda_{c,I}$ should not be interpreted as
an approximation to a threshold for the existence of a localized
equilibrium in the infinite-dimensional model. Here
$\mu_i\to\mu_\infty=0.1$ as $i\to+\infty$, and the computed value
$\lambda_{c,I}=0.10002$ is numerically close to the natural candidate
demographic threshold $\mu_\infty$. Indeed, since the mortality
profile is bounded, Theorem~\ref{thm:dmdt} applies. Moreover,
$\mu_i\geq\mu_\infty$ for every $i$, and hence
$$
\dot m(t) = \lambda m(t)(1-m(t))-\sum_i\mu_i x_i(t) 
  \leq \lambda m(t)(1-m(t))-\mu_\infty m(t).
$$

Comparison with the corresponding scalar equation shows that
$m(t)\to 0$ whenever $\lambda\leq\mu_\infty$. More precisely, it gives
an exponential upper bound when $\lambda<\mu_\infty$ and, when
$\lambda=\mu_\infty$, the algebraic upper bound
$$
m(t)\leq \frac{m(0)}{1+\lambda m(0)t}.
$$
When $\lambda>\mu_\infty$, the simulations
suggest that demographic persistence may instead be achieved through
continual movement towards traits whose mortality approaches
$\mu_\infty$, without convergence to a stationary trait distribution.
Thus, in the bounded case, $\lambda_{c,I}$ is best regarded as a
finite-truncation diagnostic consistent with the limiting mortality
level, rather than as a rigorously justified infinite-dimensional
equilibrium threshold.

Returning to the confining regime, we tested the prediction that the
equilibrium trait distribution is determined by the principal
eigenvector of $L$. We computed the principal eigenvector of $L_I$ and
normalized it to have coordinate sum $1$, thereby obtaining the
predicted equilibrium trait proportions. Although
Proposition~\ref{truncation} concerns only the convergence of spectral
radii, the numerical results reported below are consistent with the
expectation that suitably normalized principal eigenvectors of the
truncated operators approximate the principal eigenvector of $L$.

We then compared this deterministic prediction with the average trait
distribution obtained from $100$ stochastic simulations. None of the
simulated populations was extinct at time $t=200$. For each realization,
the trait proportions were calculated as $n_i(t)/\sum_j n_j(t)$, and
these proportions were then averaged across the $100$ realizations. The
results are shown in Figure~\ref{fig:bargraph}.

\begin{figure}[htbp]
\centering
\includegraphics[width=0.9\textwidth]{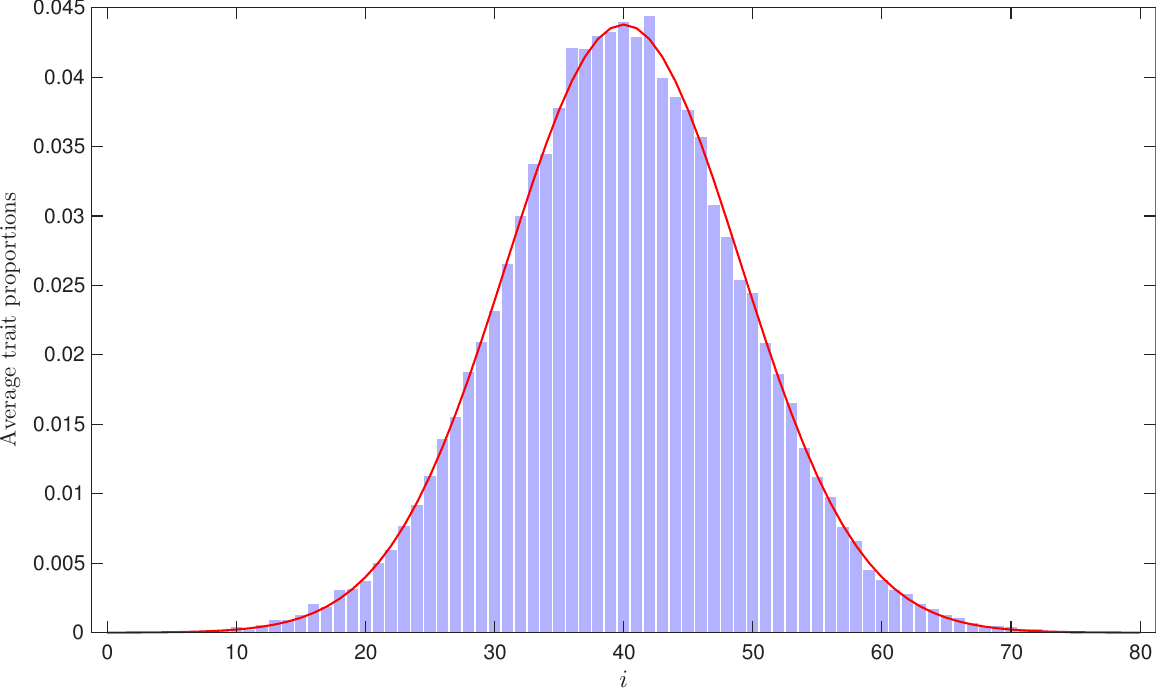}
\caption{
Average trait proportions from $100$ stochastic simulations at
time $t=200$ (bars), compared with the equilibrium profile predicted by
the finite-dimensional approximation of the principal eigenvector of $L$
(red curve).
The parameter values are $I=200$, $K=20$, $N=1000$, $\lambda=10$, $r=0.4$,
$\mu_{\min}=1$, $i^*=40$, and $f=10$. The close agreement between the
two profiles provides numerical support for the theory developed in
Section~\ref{sec:generalcase}.
}
\label{fig:bargraph}
\end{figure}

The close agreement between the stochastic and deterministic profiles
provides strong numerical evidence that the principal eigenvector of $L$
determines the equilibrium trait distribution arising from the combined
effects of mutation, selection and population regulation.

\begin{figure}[H]
\centering
\includegraphics[width=0.49\textwidth]{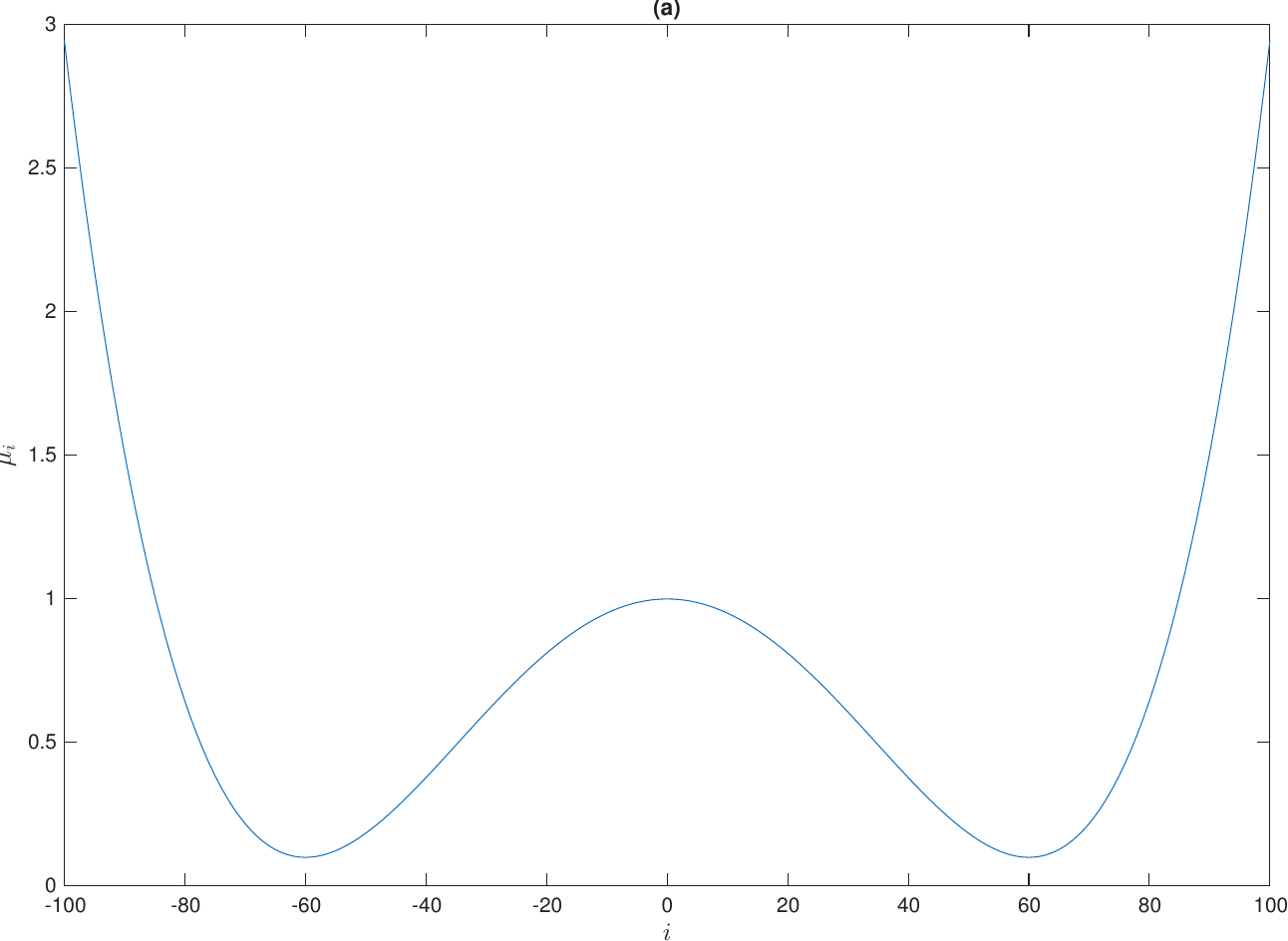}
\includegraphics[width=0.49\textwidth]{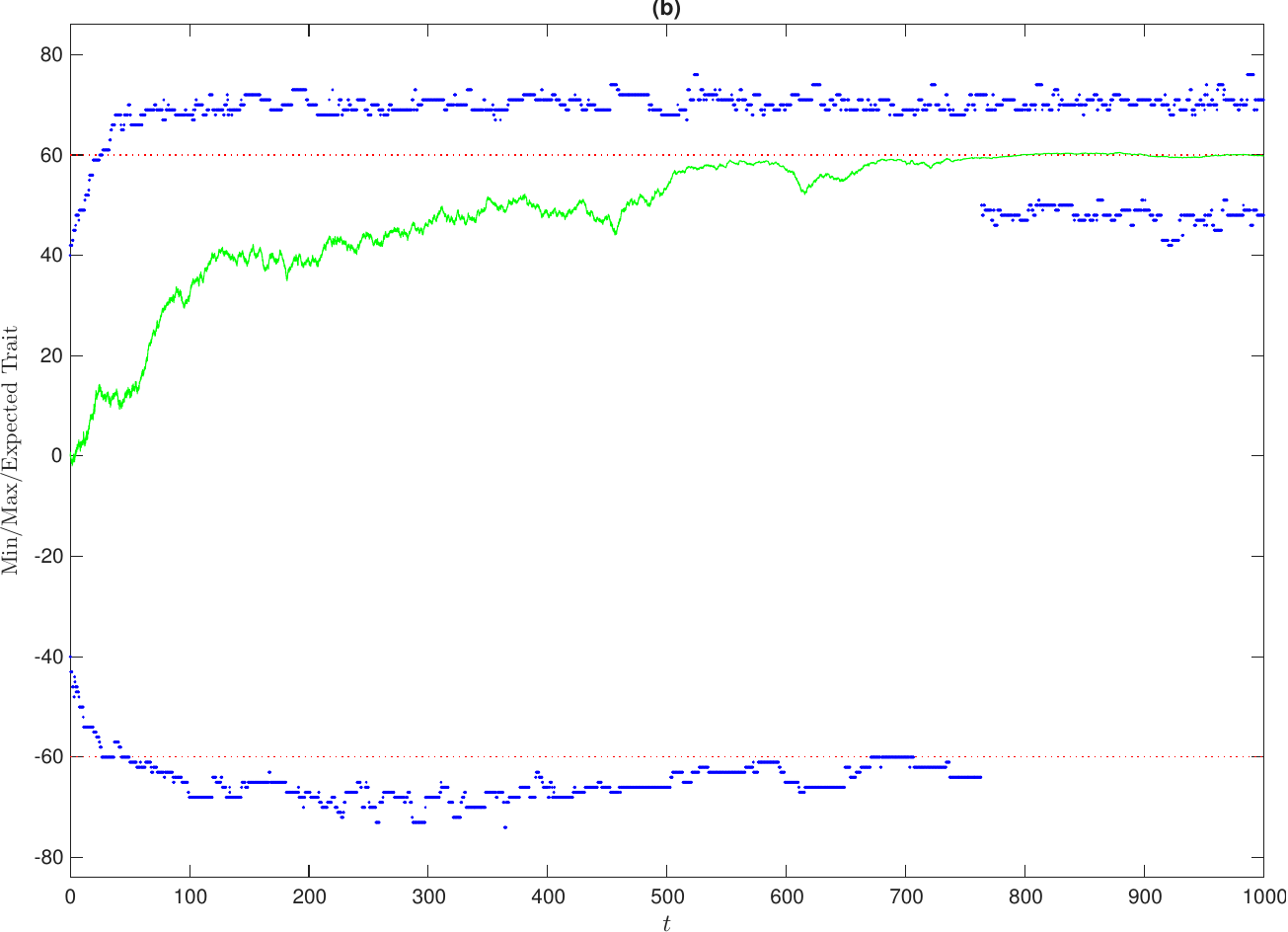}
\includegraphics[width=0.49\textwidth]{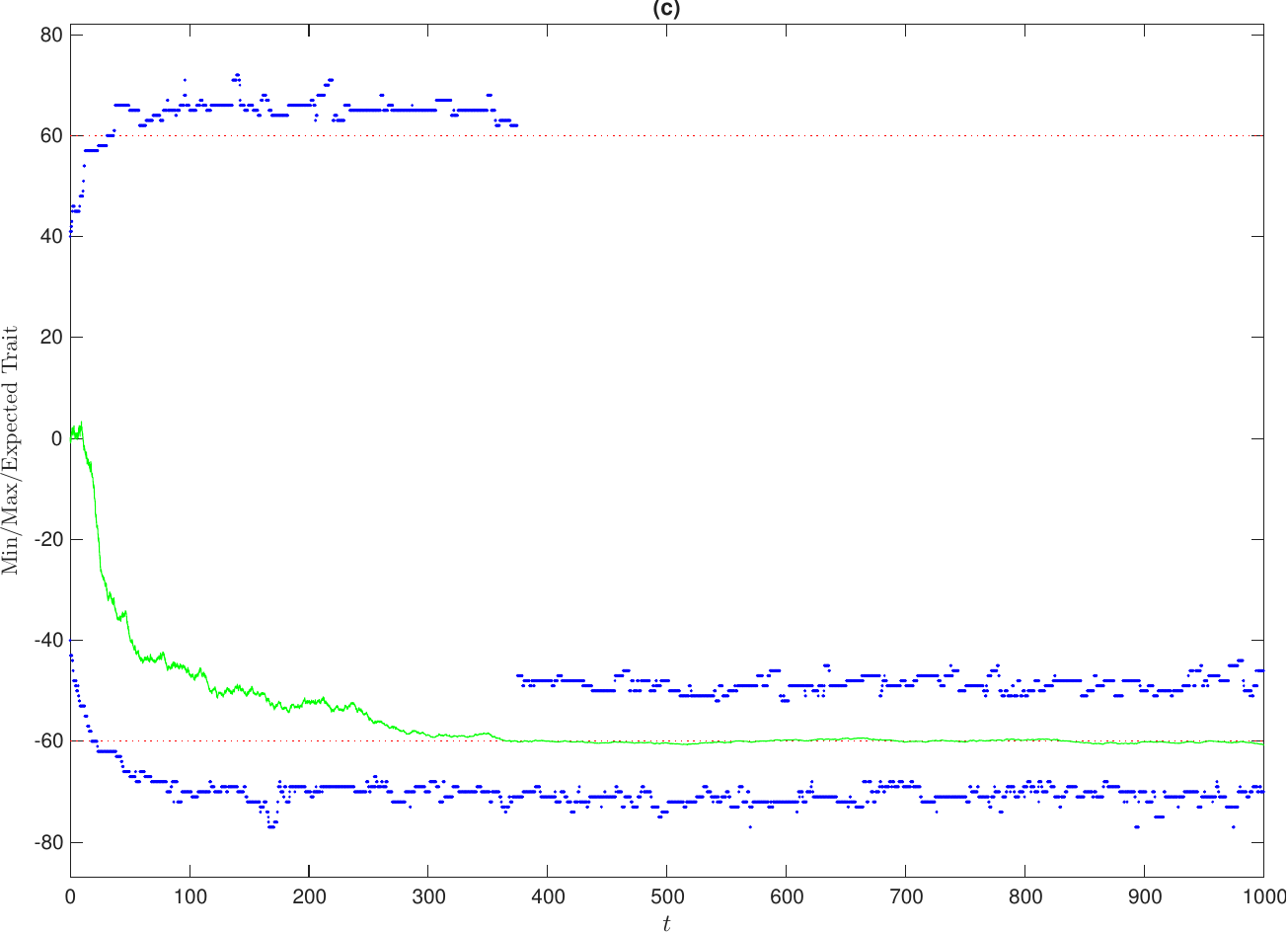}
\includegraphics[width=0.49\textwidth]{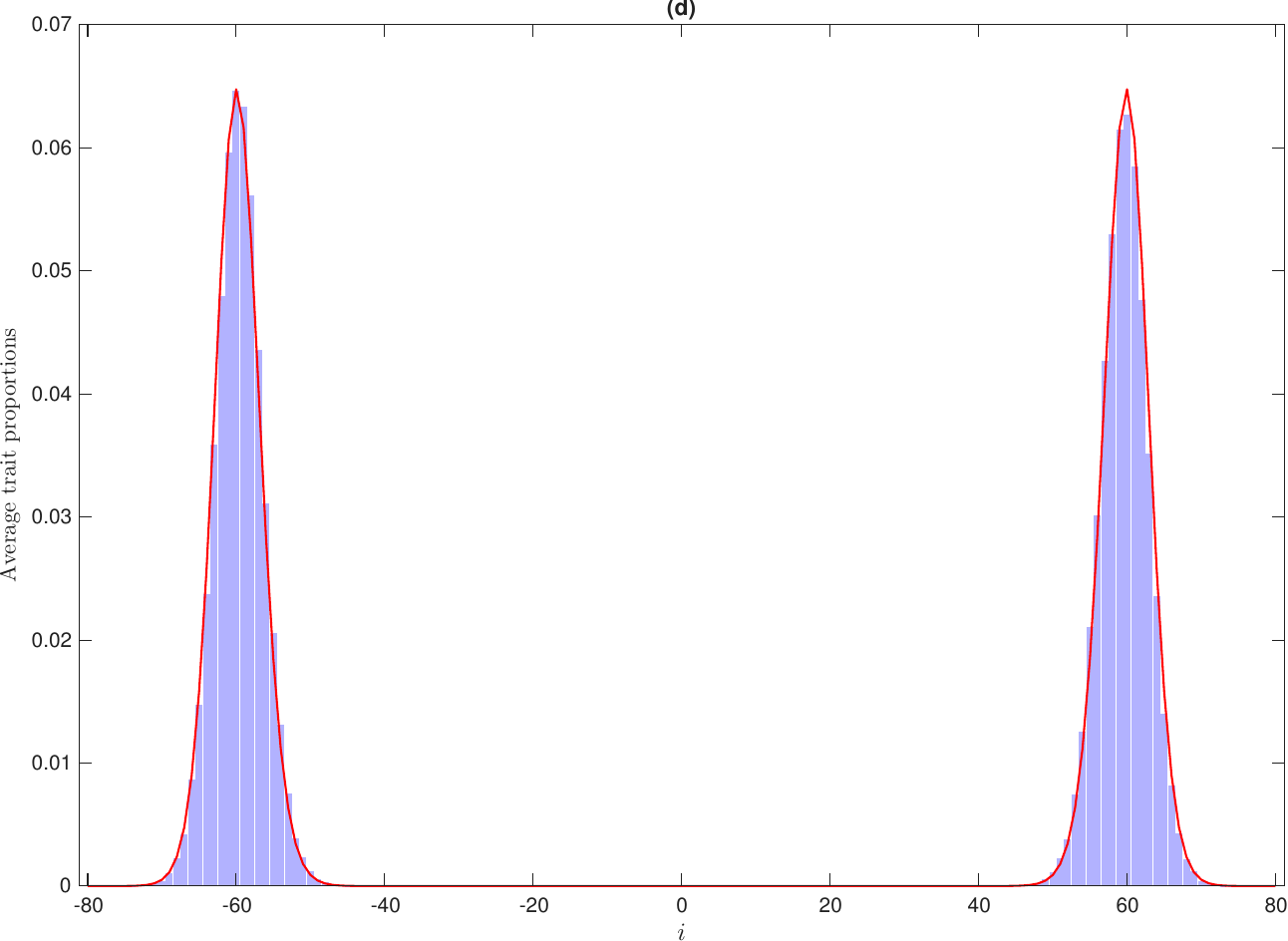}
\caption{
Simulation and deterministic approximation for disruptive selection:
$I=200$, $K=20$, $N=1000$, $\lambda=10$, $r=0.4$, $\mu_{\min}=0.1$, $f=10$, and $i^*=60$. (a)~The
mortality profile. (b)~and (c)~Representative stochastic realizations
showing the expected trait value (green) together with the minimum and
maximum traits present (blue). The red dotted lines indicate the optimal
traits $\pm i^*$. Abrupt changes in the occupied trait range occur when
the last individual in an extreme lineage disappears. Despite these
fluctuations, the population remains concentrated near the two
favourable trait values $\pm i^*$. (d)~Average trait proportions from
$200$ stochastic simulations at time $t=200$ (bars), compared with the
equilibrium profile predicted by the finite-dimensional approximation of
the principal eigenvector of~$L$ (red curve).
}
\label{fig:FiguresD}
\end{figure}

We conclude with a disruptive-selection example. The mortality profile
shown in Figure~\ref{fig:FiguresD}(a) has two equally favourable traits
at $\pm i^*$ and may be viewed as a simple model of a population
exploiting two distinct resources, such as seed-eating birds adapted to
either small or large seeds~\cite{Smith1993}. In such a setting,
intermediate phenotypes are disadvantaged relative to specialists,
producing two fitness peaks. 
Theorem~\ref{critical} predicts the existence of a unique non-zero
equilibrium whenever $\lambda r(L)>1$, and this equilibrium is strictly
positive. For the disruptive mortality profile considered here, the
computed equilibrium trait distribution is bimodal, reflecting the two
mortality minima.
Representative stochastic realizations are shown in
Figures~\ref{fig:FiguresD}(b) and~\ref{fig:FiguresD}(c). Despite the
symmetric mortality profile, sample paths are generally asymmetric. The
minimum and maximum occupied traits exhibit abrupt jumps when the final
lineage in a sparsely populated tail becomes extinct, producing a
discontinuous shift in the observed trait range. Such events reflect
demographic stochasticity, rather than changes in the selective regime.
Even so, the population remains concentrated near the favourable trait
values $\pm i^*$, consistent with confinement by disruptive selection to
the two fitness peaks.

To examine the resulting trait distribution,
Figure~\ref{fig:FiguresD}(d) compares the average trait proportions
obtained from $200$ stochastic simulations with the equilibrium profile
predicted by the principal eigenvector of the finite-dimensional
approximation of $L$. None of the simulated populations was extinct at
time $t=200$. For each realization, the trait proportions were
calculated as $n_i(t)/\sum_j n_j(t)$, and these proportions were then
averaged across the $200$ realizations. The agreement is again seen to
be good. Thus, even in the presence of disruptive selection and a
bimodal equilibrium distribution, the agreement provides further
numerical support for the prediction that the principal eigenvector of
$L$ describes the mutation--selection balance.

These experiments illustrate how confining mortality can produce a
localized mutation--selection balance, in contrast to the continuing
trait-space drift observed under bounded mortality.

\section{Discussion}
\label{sec:discussion}

The present work has examined a trait-structured population model
combining mutation, selection and density-dependent regulation on a
countably infinite trait space. The analysis reveals a qualitative
distinction that is absent from finite-dimensional mutation--selection
systems and appears to be one of the central features of the model.
Whether a localized mutation--selection equilibrium exists depends
crucially on the behaviour of the mortality profile in its tails.
When the mortality rates satisfy $\mu_i\to\infty$ as $|i|\to\infty$,
selection increasingly penalizes extreme traits. This confining effect
has a natural mathematical manifestation: the operator $L=D^{-1}P$
becomes compact on $\BS$. The resulting compactness permits the
application of Kre\u{\i}n--Rutman theory and leads to a remarkably
simple description of the equilibrium structure. The spectral radius
$r(L)$ determines a threshold $\lambda r(L)=1$ separating regimes in
which non-zero equilibria are absent from those in which a unique
non-zero equilibrium exists. In the latter regime, the equilibrium is
strictly positive, and its trait distribution is determined by the
principal eigenvector of $L$. Thus, mutation, selection and regulation
combine to produce a localized mutation--selection balance whenever
$\lambda r(L)>1$.

The situation is fundamentally different when the mortality profile is
bounded. Here mutation is able to transport mass indefinitely through
the trait space without encountering increasingly severe selective
pressure. The special case $\mu_i\equiv\mu$ illustrates this phenomenon
in a particularly transparent manner. Following a suitable time change,
the trait proportions evolve exactly as the law of a continuous-time
random walk on $\Z$. Since the corresponding random walk is not
positive recurrent, no invariant probability distribution exists, and
every fixed trait proportion converges to $0$. Thus, a population may
persist at positive density while its trait distribution continues to
drift through the trait space. From a biological perspective, population
persistence and evolutionary localization are therefore distinct
phenomena.

The bounded-mortality regime remains the most significant unresolved
aspect of the model. Although the analysis establishes the absence of
equilibrium trait distributions in certain cases, a complete
description of the asymptotic behaviour is still lacking. Numerical
experiments suggest that, under some directional-selection profiles, the
trait distribution exhibits persistent evolutionary drift rather than
convergence to a stationary state. In particular, for the arctan
mortality profile, the mean trait value appears to satisfy
$\sum_i i\,\pi_i(t)\sim\gamma\log t$ for some positive
constant~$\gamma$; refer to Figure~\ref{fig:wave}.

\begin{figure}[H]
\centering
\includegraphics[width=0.8\textwidth]{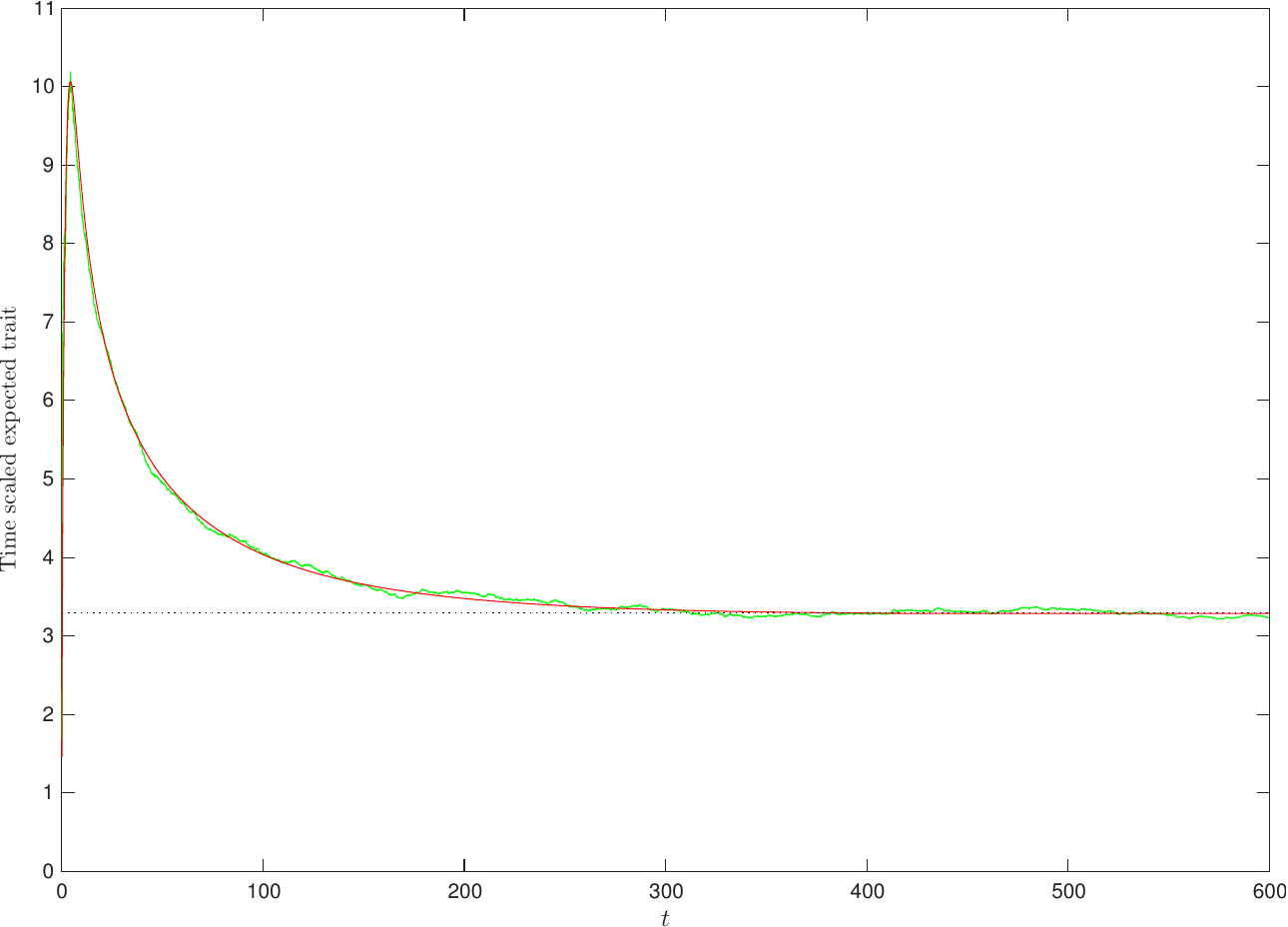}
\caption{Simulation and deterministic approximation for the arctan
mortality profile: $I=200$, $K=20$, $N=1000$, $\lambda=10$,
$r=0.4$, $\mu_{\infty}=0.1$, $c=0.9$, and $e=10$.
The expected trait value
$\sum_i i n_i(t)/\sum_k n_k(t)$ (green) and its deterministic
approximation $\sum_i i\pi_i(t)$ (red) are divided by $\log t$ for
$t>1$. The numerical results suggest that
${\sum_i i\pi_i(t)}/{\log t}\to\gamma\approx3.291$.}
\label{fig:wave}
\end{figure}

Although we do not have a rigorous derivation of this asymptotic law,
the observation is qualitatively consistent with the weakening of
selection at large trait values. For the arctan profile, recall that
$\mu_i-\mu_\infty\sim(c-\mu_\infty)e/(\pi i)$ as $i\to\infty$, so the
selective advantage associated with increasing trait values declines
only slowly as the population moves through the trait space. Selection
therefore continues to favour larger traits, but with progressively
diminishing strength. The numerical evidence suggests that adaptation
does not cease completely, but that its rate slows sufficiently to
produce approximately logarithmic growth of the mean trait.
It is noteworthy that this behaviour appears to depend on the tail of
the mortality profile. The logarithmic growth observed for the arctan
profile is not evident for the logistic profile, whose approach to its
limiting mortality is exponentially fast. This suggests that bounded
mortality profiles may themselves fall into distinct asymptotic
regimes, determined by the rate at which
$\mu_i-\mu_\infty$ decays. Whether logarithmic growth is genuine, and
more generally how the long-term behaviour depends on the tail of the
mortality profile, remain interesting open problems.

A second open question concerns stability of the strictly positive
equilibrium in the confining regime. Theorem~\ref{critical} establishes
existence and uniqueness whenever $\lambda r(L)>1$, but does not
address asymptotic stability. Numerical evidence nevertheless suggests
that the equilibrium is locally exponentially stable.
Because $D$ is unbounded, the linearization should be interpreted as an
operator with domain $\D(D)$ rather than as a bounded Fr\'echet
derivative on all of $\BS$. At the equilibrium $\bx^*$, the formal
linearized operator is
$$
A_*\bh
=
\lambda(1-m^*)P\bh-D\bh
-\lambda m(\bh)P\bx^*,
\qquad
\bh\in\D(D),
$$
where $m(\bh)=\sum_kh_k$. Since
$\lambda(1-m^*)=1/r(L)$, this may be written as the rank-one
perturbation
$$
A_*
=
A-\lambda(P\bx^*)\boldsymbol{1}^{\top},
\qquad
A=\frac{1}{r(L)}P-D,
$$
where $\boldsymbol{1}^{\top}\bh=\sum_kh_k$. Moreover,
$A\bphi=\bzero$, where $\bphi$ is the principal eigenvector of $L$.
Thus, $A$ has a neutral eigenvalue at the origin, and the rank-one term
acts non-trivially on the corresponding eigenspace. The numerical
evidence described below suggests that the resulting linearized
operator has a negative spectral bound.

For the stabilizing-selection example considered above, we computed the
eigenvalue with largest real part for finite-dimensional truncations of
the linearized operator. For truncation sizes $I=400$, $500$, $1000$,
and $2000$, this eigenvalue was approximately $-0.0642656807$, with
variation occurring only in the twelfth decimal place. Thus, the
computed spectral bounds are negative and remain separated from the
imaginary axis over the truncation sizes considered. Moreover, there is
no indication that the rightmost eigenvalue moves towards the imaginary
axis as the truncation size increases. The numerical evidence therefore
suggests that this behaviour is not merely a finite-dimensional
artifact, but reflects a genuine property of the underlying
infinite-dimensional system.

Taken together, these calculations provide consistent numerical evidence
that the linearized operator at the strictly positive equilibrium has a
negative spectral bound that is separated from the imaginary axis. Under
an appropriate linearized stability principle, this would imply local
exponential stability of the equilibrium. A rigorous proof would require
a more detailed analysis of the spectrum and of the semigroup generated
by the linearized operator, and is left for future work.

More broadly, the results indicate that density regulation alone cannot
be relied upon to maintain a localized trait distribution on an
unbounded trait space. Population regulation may ensure persistence at
the demographic level, but additional selective confinement is needed
to prevent continual evolutionary dispersal. In the confining regime
studied here, compactness of $D^{-1}P$ provides the mathematical
mechanism that yields a principal eigenvector and, whenever
$\lambda r(L)>1$, a localized equilibrium distribution. When this
compactness mechanism is absent, as in the bounded-mortality examples
considered above, mutation may instead drive persistent movement
through the trait space.

\section{AI usage}

The author used Microsoft Copilot (GPT-based AI) as a tool to assist
with phrasing, exposition, and checking intermediate calculations. The
author takes full responsibility for all content, including all
mathematical results.

\section{Appendix}
\label{appendix}

Here we collect the longer proofs of results in the order that they
were presented above.

\begin{proof}[Proof of Theorem~\ref{Thm:mildsolution}]
Define the operator $D$ on $\BS$ by $(D\bx)_i=\mu_i x_i$,
$i\in\Z$, with domain
$$
\D(D) :=\left\{\bx\in\BS: {\ts\sum_{i\in\Z}}\,\mu_i|x_i|<\infty\right\},
$$
and let $G(\bx)=\lambda(1-m(\bx))P\bx$. Then we may write~(\ref{PKP1}) as
$\dot{\bx}+D\bx=G(\bx)$, which is a perturbed linear system
(Pazy~\cite{Paz83}, Chapter 6). We will apply Theorem~1.4 of that
chapter (with $A=D$ and $f(t,\bx)=G(\bx)$, in Pazy's notation). First,
$-D$ generates a strongly continuous semigroup on $\BS$, with
domain specified above, namely
$(T(t),\, t\geq 0)$ given by $(T(t)\bx)_i=e^{-\mu_i t} x_i$, $i\in \Z$.
Indeed, $T(t)$ is a contraction, it satisfies the semigroup property,
and
$$
\|T(t)\bx-\bx\|_1 =\sum_{i\in\mathbb Z} |e^{-\mu_i t}-1|\,|x_i| \to 0 
\qquad\text{as }t\downarrow 0,
$$
by dominated convergence. Next, $P$ is bounded: for $\bx\in \BS$,
$$
\norm{P\bx}_1
=\sum_i\Bigl|\sum_k x_k p_{i-k}\Bigr|
\leq \sum_i\sum_k |x_k| p_{i-k}
=\sum_k |x_k|\sum_i p_{i-k}
=\norm{\bx}_1.
$$
We may now prove that $G$ is locally Lipschitz: for all 
$\bx,\by\in \BS$,
$$
G(\bx)-G(\by) = \lambda \left( (1-m(\bx)) P(\bx-\by) +
(m(\by)-m(\bx))P\by \right),
$$ 
and so, since 
$|m(\bx)-m(\by)|\leq \norm{\bx-\by}_1$, we have 
$$
\norm{G(\bx)-G(\by)}_1 
\leq \lambda (1 + \norm{\bx}_1 + \norm{\by}_1)\norm{\bx-\by}_1.
$$
Therefore, $\norm{G(\bx)-G(\by)}_1 \leq \lambda (1 +
2R)\norm{\bx-\by}_1$ on the ball $\{\bx\in \BS: \norm{\bx}_1\leq
R\}$. 
Thus, all the hypotheses of Pazy's~Theorem~1.4 are satisfied.
Consequently, the initial value problem has a unique maximal mild
solution $\bx\in C([0,t_{\max});\BS)$ satisfying the
variation-of-constants formula 
\begin{equation}
\bx(t) = T(t)\bx_0+\int_0^tT(t-s)G(\bx(s))\,ds, \qquad 0\leq t<t_{\max}.
\label{vofc}
\end{equation}
Using the explicit form of $T(t)$, this yields equation~\eqref{mildsolution}
componentwise.
\end{proof}

\begin{proof}[Proof of Theorem~\ref{staysinE}]
Suppose that $\bx_0\in E$. We first introduce an auxiliary equation in
which the effective birth coefficient is truncated below at zero.
Specifically, we let $G_+(\bx):=\lambda(1-m(\bx))^+P\bx$, where
$r^+:=\max\{r,0\}$, 
noting that $G_+$ maps the positive cone 
$$
K:=\ell_1^+(\Z) = \{\bx\in\BS:x_i\geq 0\text{ for all }i\in\Z\}
$$
into itself, and consider $\dot{\bx}+D\bx=G_+(\bx)$, where $D$ is
defined earlier in the proof of Theorem~\ref{Thm:mildsolution}.
We shall show that the mild
solution to this auxiliary equation has total mass at most~$1$. It will
then follow that the truncation is inactive, so that the auxiliary
solution is also our original mild solution~\eqref{mildsolution}.

Since $r\mapsto r^+$ is Lipschitz and $P$ is bounded on
$\BS$, $G_+$ is locally Lipschitz:
for $\bx,\by$ in the ball 
$B_R:=\{\bz\in\BS:\|\bz\|_1\leq  R\}$, we have 
\begin{align*} 
\|G_+(\bx)-G_+(\by)\|_1 &\leq  \lambda (1-m(\bx))^+ \|P(\bx-\by)\|_1 \\ 
&\qquad\qquad +\lambda\left|(1-m(\bx))^+-(1-m(\by))^+\right| \|P\by\|_1
\\ &\leq  \lambda(1+2R)\|\bx-\by\|_1.
\end{align*}
The argument used in the proof of
Theorem~\ref{Thm:mildsolution} therefore gives a unique maximal mild
solution $\by\in C([0,t_{\max}^+);\BS)$ satisfying
\begin{equation}
y_i(t)=e^{-\mu_i t}x_i(0)
+\int_0^t e^{-\mu_i(t-s)}
\lambda(1-m(\by(s)))^+(P\by(s))_i\,ds.
\label{auxiliarymildsolution}
\end{equation}
The semigroup $(T(t))_{t\geq 0}$ generated by $-D$ maps $K$ into itself,
as does $G_+$. 
Hence the usual Picard iterates belong to $K$.
Since $K$ is
closed in $\BS$, their limit also belongs to $K$. Thus
$\by(t)\geq 0$ for $0\leq t<t_{\max}^+$.

For arbitrary $0\leq r\leq t<t_{\max}^+$ 
formula~\eqref{auxiliarymildsolution} gives
$$
y_i(t)=e^{-\mu_i(t-r)}y_i(r)
+\int_r^t e^{-\mu_i(t-s)}
\lambda(1-m(\by(s)))^+(P\by(s))_i\,ds.
$$
All terms are non-negative, so Tonelli's theorem permits summation over
$i$. Since $e^{-\mu_i(t-r)}\leq 1$ and $P\by(s)\geq 0$, we have
$\sum_i e^{-\mu_i(t-r)}y_i(r)\leq m(\by(r))$ and
$$
\sum_i e^{-\mu_i(t-s)}(P\by(s))_i
\leq\sum_i(P\by(s))_i=m(\by(s)).
$$
Consequently,
\begin{equation}
m(\by(t))
\leq m(\by(r))
+\int_r^t\lambda(1-m(\by(s)))^+m(\by(s))\,ds,
\qquad 0\leq r\leq t<t_{\max}^+.
\label{shiftedmassinequality}
\end{equation}
We will show that $m(\by(t))\leq1$ throughout the interval of existence.
Suppose the contrary, and choose $t<t_{\max}^+$ such that
$m(\by(t))>1$. Since $m(\by(0))=m(\bx_0)\leq 1$ and
$t\mapsto m(\by(t))$ is continuous, the set
$\{s\in[0,t]:m(\by(s))=1\}$ is non-empty. Let
$$
\sigma:=\sup\{s\in[0,t]:m(\by(s))=1\}.
$$
Then $m(\by(\sigma))=1$ and $m(\by(s))>1$ for $\sigma<s\leq t$; if
$m(y(s))<1$ for some $s\in(\sigma,t)$, then continuity, together
with $m(y(t))>1$, would imply the existence of some $u\in(s,t)$ such
that $m(y(u))=1$, contradicting the definition of $\sigma$.

Hence $(1-m(\by(s)))^+=0$ on $(\sigma,t]$.
Applying \eqref{shiftedmassinequality} with $r=\sigma$ gives
$$
m(\by(t))
\leq m(\by(\sigma))
+\int_\sigma^t
\lambda(1-m(\by(s)))^+m(\by(s))\,ds
=1,
$$
contrary to the choice of $t$. Therefore
$m(\by(t))\leq 1$ for $0\leq t<t_{\max}^+$.
Since $\by(t)\geq 0$, it follows that $\|\by(t)\|_1=m(\by(t))\leq 1$,
$0\leq t<t_{\max}^+$. Because $G_+$ is locally Lipschitz on
$\BS$, the standard continuation criterion for semilinear
evolution equations (Theorem~1.4 of Pazy~\cite{Paz83}, Chapter 6) implies 
that, if $t_{\max}^+<\infty$, then
necessarily $\|\by(t)\|_1 \to \infty$ as $t\uparrow t_{\max}^+$.
Since the trajectory remains in the unit ball of
$\BS$, this is impossible. Hence $t_{\max}^+=\infty$.
Moreover, $m(\by(t))\leq 1$ implies
$(1-m(\by(t)))^+=1-m(\by(t))$ for all $t\geq 0$. Thus
\eqref{auxiliarymildsolution} is precisely the mild equation for the
original system. Hence $\by$ is a global mild solution of
\eqref{PKP1} with initial value $\bx_0$.
By uniqueness in Theorem~\ref{Thm:mildsolution}, $\by$ agrees with
$\bx$ on their common interval of existence. Since $\by$ is global,
maximality of $\bx$ implies that $t_{\max}=\infty$ and
$\bx(t)=\by(t)$ for all $t\geq 0$. Consequently,
$\bx(t)\geq 0$ and $m(\bx(t))\leq 1$ for all $t\geq 0$, and hence
$\bx(t)\in E$ for all $t\geq 0$.
\end{proof}

\noindent
\emph{Remark}.
The auxiliary equation is needed to avoid a circular argument. For
non-negative $\bx$, the original nonlinearity $G(\bx)$ is non-negative
only when $m(\bx)\leq 1$, so invariance of the positive cone cannot be
used to prove the mass bound before that bound is already known. By
contrast, $G_+$ maps the positive cone into itself without any prior
restriction on $m(\bx)$. Positivity of the auxiliary solution can
therefore be established first; the bound $m(\by(t))\leq 1$ then follows
and shows that the truncation is never active.

\medskip
\begin{proof}[Proof of Theorem~\ref{BL}]
Barbour and Luczak index their countable type space by $\mathbb Z_+$,
whereas ours is indexed by $\mathbb Z$. We may enumerate $\mathbb Z$
so that only finitely many types with bounded $|i|$ occur before any
given level. Their framework is therefore applicable after relabelling,
and we retain our indexing by $\mathbb Z$ throughout.
In their notation, the possible jump vectors are
$\J=\{\be_i,-\be_i:i\in\mathbb Z\}$.
For $\bxi\in\ell_1^+(\mathbb Z)$, define
\[
 \alpha_{\be_i}(\bxi)
 =
 \lambda(1-m(\bxi))^+(P\bxi)_i,
 \qquad
 \alpha_{-\be_i}(\bxi)=\mu_i\xi_i,
 \qquad i\in\mathbb Z,
\]
where $r^+=\max\{r,0\}$. If $\bn\in S$, then
$$
 N\alpha_{\be_i}(N^{-1}\bn)
 = \lambda\left(1-\frac{\|\bn\|}{N}\right) \sum_k n_kp_{i-k},
\qquad
 N\alpha_{-\be_i}(N^{-1}\bn) = \mu_i n_i.
$$
These are precisely the transition rates in~\eqref{TheModel}. The
positive part in the definition of $\alpha_{\be_i}$ is inactive on the
state space of the process, but ensures that the rate functions are
non-negative on the whole positive cone.

We first check conditions~(1.2) and~(1.3) of~\cite{BL12b}. Every jump
vector has $\ell_1$ norm equal to $1$, so condition~(1.2) holds with $J_*=1$.
Every coordinate of a jump vector is at least $-1$, and the rate of the
jump $-\be_i$ vanishes when $\xi_i=0$. If $\bxi\geq0$ has finite
support, then
\[
 \sum_{J\in\J}\alpha_J(\bxi)
 =
 \lambda(1-m(\bxi))^+m(\bxi)
 +
 \sum_i\mu_i\xi_i
 <\infty.
\]
Thus condition~(1.3) also holds. Notice that boundedness of the
mortality profile is not required here, since the final sum is finite
whenever $\bxi$ has finite support.

We next verify Assumption~2.1 of~\cite{BL12b}. Define
\[
 \nu(i):=(1+|i|)^q,\qquad
 S_r(\bxi):=\sum_i\nu(i)^r\xi_i,\qquad
 M_r:=\sum_\ell\nu(\ell)^rp_\ell.
\]
The function $\nu$ is at least $1$, and only finitely many traits satisfy
$\nu(i)\leq K$ for any fixed $K<\infty$. The assumed mutation moment
implies that $M_r<\infty$ for $0\leq r\leq11$.
Since $1+|k+\ell| \leq (1+|k|)(1+|\ell|)$,
we have $\nu(k+\ell)\leq\nu(k)\nu(\ell)$. Consequently,
\begin{equation}
 \sum_i\nu(i)^r(P\bxi)_i
 \leq
 M_rS_r(\bxi),
 \qquad 0\leq r\leq11.
 \label{app:weightedPbound}
\end{equation}

We take $r_{\max}^{(1)}=11$ and $r_{\max}^{(2)}=5$. For
$J=\be_i$ and $J=-\be_i$, respectively,
\[
 J^\top\nu_r=\nu(i)^r
 \qquad\hbox{and}\qquad
 J^\top\nu_r=-\nu(i)^r.
\]
It follows from~\eqref{app:weightedPbound} that the birth contribution
to the sum in condition~(2.5) of~\cite{BL12b} is finite for
$0\leq r\leq11$. The death contribution is finite because the state
has finite support. Thus condition~(2.5) holds.

The drift quantities defined in~(2.4) of~\cite{BL12b} satisfy
$$
 U_r(\bxi)
 =
 \lambda(1-m(\bxi))^+
 \sum_i\nu(i)^r(P\bxi)_i
 -
 \sum_i\mu_i\nu(i)^r\xi_i
 \leq
 \lambda M_rS_r(\bxi),
 \qquad 0\leq r\leq11,
$$
because the mortality contribution is non-positive. These estimates
give the inequalities required in~(2.6) of~\cite{BL12b}.
The corresponding quadratic-variation quantities satisfy
\[
 V_r(\bxi)
 =
 \lambda(1-m(\bxi))^+
 \sum_i\nu(i)^{2r}(P\bxi)_i
 +
 \sum_i\mu_i\nu(i)^{2r}\xi_i.
\]
Since $\mu_i\leq C_{\!\bmu}\nu(i)$, we obtain
$V_0(\bxi)\leq(\lambda+C_{\!\bmu})S_1(\bxi)$
and, for $1\leq r\leq5$,
\[
 V_r(\bxi)
 \leq
 \{\lambda M_{2r}+C_{\!\bmu}\}S_{2r+1}(\bxi).
\]
Thus~(2.7) of~\cite{BL12b} holds with $p(r)=2r+1$. Since
$p(r)\leq11$ for $1\leq r\leq5$, Assumption~2.1 is satisfied.

The infinitesimal drift associated with the rate functions has the
decomposition
\[
 \F_0(\bxi)=A\bxi+\F(\bxi),
 \qquad
 A=-D,
 \qquad
 \F(\bxi)=\lambda(1-m(\bxi))^+P\bxi.
\]
The matrix $A$ is diagonal, with $A_{ij}=-\mu_i\delta_{ij}$. Its
off-diagonal entries are therefore non-negative, and each of its
columns contains only one non-zero entry. Hence condition~(3.1)
of~\cite{BL12b} holds.

To avoid confusion with our mortality profile, denote the weight in
condition~(3.2) of~\cite{BL12b} by~$\omega$. We choose
$\omega(i)=1$ for every $i$. Then $A^\top\omega\leq0$, so
condition~(3.2) holds with $w=0$. The corresponding weighted norm is
the ordinary $\ell_1$ norm. The semigroup constructed from $A$ is
$(T(t)\bxi)_i=e^{-\mu_i t}\xi_i$,
which is the strongly continuous contraction semigroup already
introduced in the proof of Theorem~\ref{Thm:mildsolution}.

We next verify condition~(4.1) of~\cite{BL12b}. If
$\norm{\bxi}_1,\norm{\boldsymbol{\eta}}_1\leq R$, then the estimate used
in the proof of Theorem~\ref{staysinE} gives
\[
 \norm{\F(\bxi)-\F(\boldsymbol{\eta})}_1
 \leq
 \lambda(1+2R)\norm{\bxi-\boldsymbol{\eta}}_1.
\]
there.
Since $\norm{P\bxi}_1\leq\norm{\bxi}_1$ and $|m(\bxi)|\leq\norm{\bxi}_1$,
\[
 \norm{\mathcal F(\bxi)}_1
 \leq
 \lambda(1+\norm{\bxi}_1)\norm{\bxi}_1.
\]
Thus $\F$ maps $\ell_1(\mathbb Z)$ into itself and is locally Lipschitz

With these choices, equation~(4.2) of~\cite{BL12b} is
\[
 \bxi(t)
 =
 T(t)\bxi(0)
 +
 \int_0^t
 T(t-s)\lambda(1-m(\bxi(s)))^+P\bxi(s)\,ds.
\]
This is the auxiliary mild equation used in the proof of
Theorem~\ref{staysinE}. When $\bxi(0)\in E$, that theorem gives
$\bxi(t)\in E$ for every $t\geq0$, so the positive part is inactive.
The solution of equation~(4.2) of~\cite{BL12b} is therefore the mild
solution of~\eqref{PKP1}.

It remains to check Assumption~4.2 of~\cite{BL12b}. 
Since $\omega(i)=1$ and $\nu(i)\geq1$, the first part of
Assumption~4.2 holds with $r_\omega=1$, because
\[
 \sup_i\frac{\omega(i)}{\nu(i)}\leq 1.
\]
Moreover, $\omega(i)/\nu(i)\to0$ as $|i|\to\infty$, so we may take
$\widetilde r_\omega=1$.
For the second part, define
$\zeta(i):=\nu(i)^4=(1+|i|)^{4q}$. Then,
\[
 \sum_i\zeta(i)^{-1/2}
 =
 \sum_i(1+|i|)^{-2q}
 <\infty.
\]
Moreover, using the notation of~(2.23) of~\cite{BL12b},
$d(\be_i,\zeta) = d(-\be_i,\zeta) = \nu(i)^4$.
It follows that
$$
 \sum_{J\in\J}\alpha_J(\bxi)d(J,\zeta)
 \leq
 \lambda M_4S_4(\bxi)+C_{\!\bmu} S_5(\bxi)
 \leq
 (\lambda M_4+C_{\!\bmu})S_5(\bxi).
$$
Thus condition~(2.25) of~\cite{BL12b} holds with
$r(\zeta)=5$ and $b(\zeta)=1$. Finally,
\[
 \sup_i
 \frac{\omega(i)(|A_{ii}|+1)}{\sqrt{\zeta(i)}}
 =
 \sup_i\frac{\mu_i+1}{\nu(i)^2}
 \leq
 C_{\!\bmu}+1.
\]
This is condition~(4.12) of~\cite{BL12b}, and Assumption~4.2 is
verified.
The exponent defined in~(4.13) of~\cite{BL12b} is therefore
\[
 \rho(\zeta,\omega)
 =
 \max\{r(\zeta),p(r(\zeta)),\widetilde r_\omega\}
 =
 \max\{5,11,1\}
 =
 11.
\]
The deterministic initial condition in Theorem~\ref{BL} is therefore
$S_{11}(\bx_0)<\infty$. Moreover, since
$\bX(0)=N^{-1}\bn(0)$, the uniform bound
$S_{11}(\bX(0))\leq C_*$ is equivalent to
\[
 \sum_i\nu(i)^{11}n_i(0)\leq NC_*,
\]
which is the corresponding stochastic initial moment condition in
Theorem~4.7 of~\cite{BL12b}. Finally, the norm associated with
$\omega(i)=1$ is the ordinary $\ell_1$ norm. The conclusion of that
theorem therefore gives constants $K_1$, $K_2$, and $K_3$ with the
properties stated in Theorem~\ref{BL}.
\end{proof}

\begin{proof}[Proof of Theorem~\ref{thm:dmdt}]
Let $M=\sup_i\mu_i<\infty$. Then $D$ is a bounded linear operator on
$\BS$, since $\|D\bx\|_1\leq M\|\bx\|_1$. We have already shown that $P$
is bounded on $\BS$ and that $G(\bx)=\lambda(1-m(\bx))P\bx$ is locally
Lipschitz. Hence $F=G-D$ is locally Lipschitz from $\BS$ into itself.
Since $D$ is bounded, the semigroup generated by $-D$ is uniformly
continuous and differentiable in the operator norm.
Differentiation of the variation-of-constants
formula~\eqref{vofc} shows that the mild solution specified in
Theorem~\ref{Thm:mildsolution} is a classical solution: $\bx\in
C^1([0,\infty);\BS)$ and $\dot \bx(t)=F(\bx(t))$.

The map $m:\BS\to\R$, given by $m(\bx)=\sum_i x_i$, is a bounded linear
functional, since $|m(\bx)|\leq\sum_i|x_i|=\|\bx\|_1$.
Consequently,
$\dot m(t)=m(\dot \bx(t))=m(F(\bx(t)))$. Since $\bx(t)\in E$ by
Theorem~\ref{staysinE}, Tonelli's theorem gives
$$
\sum_i(P\bx(t))_i=\sum_i\sum_k x_k(t)p_{i-k}
=\sum_k x_k(t)\sum_i p_{i-k}=m(t).
$$
Therefore, \eqref{dmdt} follows on applying the linear functional
$m$ to both sides of
$\dot \bx(t)=\lambda(1-m(t))P\bx(t)-D\bx(t)$.
Finally, $m(0)=\sum_i x_i(0)=m_0$.
\end{proof}

\begin{proof}[Proof of Proposition~\ref{minmortality}]
Fix $j\in J^*$. If $x_i(0)=0$, then the invariance of the coordinate
hyperplanes implies that $x_i(t)=0$ for all $t\geq0$. If $x_i(0)>0$
and $i\notin J^*$, then $\mu_i>\mu^*$, and the ratio
formula~\eqref{ratio} gives
$$
x_i(t) = x_j(t)\frac{x_i(0)}{x_j(0)} e^{-(\mu_i-\mu^*)t}.
$$
Since $x_j(t)\leq m(t)\leq1$, it follows that
$$
x_i(t) \leq \frac{x_i(0)}{x_j(0)} e^{-(\mu_i-\mu^*)t},
$$
and hence $x_i(t)\to 0$. Thus $x_i(t)\to 0$ for every $i\notin J^*$.
Moreover, $0\leq x_i(t)\leq x_i(0)/x_j(0)$ for every
$i\notin J^*$, where the upper bound is $0$ if $x_i(0)=0$, and
$$
\sum_{i\notin J^*}x_i(0)/x_j(0)\leq m(0)/x_j(0)<\infty.
$$
Dominated convergence therefore gives $\sum_{i\notin J^*}x_i(t)\to 0$.

If $i,j\in J^*$, then $\mu_i=\mu_j=\mu^*$, so the ratio formula
shows that $x_i(t)/x_j(t)=x_i(0)/x_j(0)$. Consequently, the relative
proportions within $J^*$ are preserved:
$$
\frac{x_i(t)}{\sum_{k\in J^*}x_k(t)}
=
\frac{x_i(0)}{\sum_{k\in J^*}x_k(0)},
\qquad i\in J^*.
$$

It remains to determine the limiting total mass in $J^*$. Define
$m_{J^*}(t):=\sum_{j\in J^*}x_j(t)$ and
$r(t):=\sum_{i\notin J^*}x_i(t)$. Then
$m(t)=m_{J^*}(t)+r(t)$ and, by the preceding argument, $r(t)\to 0$.
For $i\in J^*$, equation~\eqref{eq:nomutation} gives $\dot
x_i(t)=\lambda x_i(t)(\rho^*-m(t))$. Summing over $i\in J^*$ (which is
justified by absolute summability) therefore gives $\dot
m_{J^*}(t)=\lambda m_{J^*}(t) (\rho^*-m_{J^*}(t)-r(t))$.
Let $\varepsilon\in(0,\rho^*)$. Since $r(t)\to 0$, there exists
$t_0\geq0$ such that $0\leq r(t)<\varepsilon$ for all $t\geq t_0$.
Therefore,
$$
\lambda m_{J^*}(t)\bigl(\rho^*-\varepsilon-m_{J^*}(t)\bigr)
\leq\dot m_{J^*}(t)
\leq\lambda m_{J^*}(t)\bigl(\rho^*-m_{J^*}(t)\bigr),
\qquad t\geq t_0.
$$
Comparison with the corresponding Verhulst equations, with the same
initial value at $t_0$, yields
$\liminf_{t\to\infty}m_{J^*}(t)\geq\rho^*-\varepsilon$ and
$\limsup_{t\to\infty}m_{J^*}(t)\leq\rho^*$. Since $\varepsilon>0$
may be chosen arbitrarily small, it follows that
$m_{J^*}(t)\to\rho^*$.
Hence, for each $i\in J^*$,
$$
x_i(t)
=m_{J^*}(t)\frac{x_i(t)}{\sum_{j\in J^*}x_j(t)}
\longrightarrow
\rho^*\frac{x_i(0)}{\sum_{j\in J^*}x_j(0)}.
$$
Together with $x_i(t)\to 0$ for $i\notin J^*$, this proves that
$x_i(t)\to x_i^*$ for every $i\in\Z$, with $\bx^*$ as stated.
\end{proof}

\begin{proof}[Proof of Proposition~\ref{notattained}]
For each $i$ such that $x_i(0)>0$, equation~\eqref{eq:nomutation},
viewed as a scalar linear equation, implies that $x_i(t)>0$ for all
$t\geq0$. Dividing by $x_i(t)$ and integrating from $0$ to $t$ gives
$$
x_i(t)=x_i(0)e^{-\mu_i t}B(t),
\qquad
B(t):=\exp\left(\int_0^t\lambda(1-m(s))\,ds\right).
$$
The same formula holds when $x_i(0)=0$, since in that case
$x_i(t)=0$ for all $t\geq0$. Let
$Z(t):=\sum_i x_i(0)e^{-\mu_i t}$, so that $m(t)=B(t)Z(t)$.
Since $B(t)>0$, set $Y(t):=B(t)^{-1}$. From
$B^{\myprime}(t)=\lambda(1-m(t))B(t)$ and $m(t)Y(t)=Z(t)$, we obtain
$Y^{\myprime}(t)=-\lambda Y(t)+\lambda Z(t)$, with $Y(0)=1$. Hence
$$
Y(t)=e^{-\lambda t}
\left(1+\lambda\int_0^t e^{\lambda s}Z(s)\,ds\right),
$$
and therefore
\begin{equation}
m(t)
=
\frac{e^{\lambda t}Z(t)}
{1+\lambda\int_0^t e^{\lambda s}Z(s)\,ds}.
\label{eq:massrepresentation}
\end{equation}
Let $\delta_i:=\mu_i-\mu^*$ and
$\beta:=\lambda-\mu^*>0$, and write
$Z(t)=e^{-\mu^*t}a(t)$, where
$a(t):=\sum_i x_i(0)e^{-\delta_i t}$. Equation
\eqref{eq:massrepresentation} then becomes
\begin{equation}
m(t)
=
\frac{e^{\beta t}a(t)}
{1+\lambda\int_0^t e^{\beta s}a(s)\,ds}.
\label{eq:massrepresentation2}
\end{equation}
We first show that $-a^{\myprime}(t)/a(t)\to 0$. For every $\varepsilon>0$, let
$c_\varepsilon:=\sum_{\{i:\delta_i\leq\varepsilon/2\}}x_i(0)$.
Since $\mu^*$ is the infimum of the mortality rates on the initial
support, $c_\varepsilon>0$, and hence
$a(t)\geq c_\varepsilon e^{-\varepsilon t/2}$.
We next justify differentiating the series defining $a(t)$. Fix
$t_0>0$. For $t\geq t_0$,
$x_i(0)e^{-\delta_i t}\leq x_i(0)$, while
$$
\delta_i x_i(0)e^{-\delta_i t}
\leq \frac{x_i(0)}{e t_0},
$$
because $\sup_{u\geq0}ue^{-ut_0}=1/(e t_0)$. Since
$\sum_i x_i(0)<\infty$, both the series defining $a$ and the series of
its derivatives converge uniformly on $[t_0,\infty)$. Termwise
differentiation is therefore justified, and
$$
-a^{\myprime}(t)=\sum_i\delta_i x_i(0)e^{-\delta_i t},
\qquad t>0.
$$
Split this sum according to whether $\delta_i\leq\varepsilon$ or
$\delta_i>\varepsilon$. The first part is at most $\varepsilon a(t)$.
For all sufficiently large $t$, the function
$u\mapsto ue^{-ut}$ is decreasing on $[\varepsilon,\infty)$.
Consequently, if $\delta_i>\varepsilon$, then
$\delta_i e^{-\delta_i t}\leq\varepsilon e^{-\varepsilon t}$, and
hence
$-a^{\myprime}(t) \leq \varepsilon a(t)+\varepsilon m(0)e^{-\varepsilon t}$.
Using $a(t)\geq c_\varepsilon e^{-\varepsilon t/2}$, we obtain
$$
0\leq-\frac{a^{\myprime}(t)}{a(t)}
\leq
\varepsilon+
\frac{\varepsilon m(0)}{c_\varepsilon}
e^{-\varepsilon t/2}
$$
for all sufficiently large $t$. Letting first $t\to\infty$ and then
$\varepsilon\downarrow0$ shows that $-a^{\myprime}(t)/a(t)\to 0$.

Now set $f(t):=e^{\beta t}a(t)$. Then
$f^{\myprime}(t)/f(t)=\beta+a^{\myprime}(t)/a(t)\to\beta$. Since $\beta>0$, we have
$f^{\myprime}(t)/f(t)\geq\beta/2$ for all sufficiently large $t$. It follows
that $f(t)\to\infty$. Moreover, $f^{\myprime}(t)>0$ eventually, so
l'H{\^o}pital's rule gives
$$
\frac{\int_0^t f(s)\,ds}{f(t)} \to \frac{1}{\beta}.
$$
Substitution into \eqref{eq:massrepresentation2} yields
$m(t)\to {\beta}/{\lambda} = 1-{\mu^*}/{\lambda} = \rho^*$.
Finally,
$\pi_i(t)=x_i(0)e^{-\delta_i t}/a(t)$. If $x_i(0)=0$, then
$\pi_i(t)=0$ for all $t\geq0$. Otherwise, $\delta_i>0$ because the
infimum is not attained on the initial support. Choose
$\varepsilon\in(0,\delta_i)$ and let
$d_\varepsilon:=\sum_{\{j:\delta_j\leq\varepsilon\}}x_j(0)$.
Again, the definition of $\mu^*$ implies that $d_\varepsilon>0$, and
$a(t)\geq d_\varepsilon e^{-\varepsilon t}$. Therefore,
$$
0\leq\pi_i(t)
\leq
\frac{x_i(0)}{d_\varepsilon}
e^{-(\delta_i-\varepsilon)t}
\to 0.
$$
Thus, $\pi_i(t)\to 0$ for every $i\in\Z$. Since
$x_i(t)=m(t)\pi_i(t)$ and $m(t)\to\rho^*$, it follows that
$x_i(t)\to 0$ for every $i\in\Z$.
\end{proof}

\begin{proof}[Proof of Lemma~\ref{lemma:compact}]
Define $M$ on $\BS$ by $(M\bx)_i=x_i/\mu_i$, $i\in\Z$. Since
$\mu_i\to\infty$ as $|i|\to\infty$, $M$ is the operator-norm limit
of the finite-rank operators $M_n$ defined by
$$
(M_n\bx)_i=
\begin{cases}
x_i/\mu_i, & |i|\leq n,\\
0, & |i|>n.
\end{cases}
$$
Indeed,
$$
\|(M-M_n)\bx\|_1
\leq
\left(\sup_{|i|>n}\frac1{\mu_i}\right)\|\bx\|_1,
$$
so that $\|M-M_n\|\leq\sup_{|i|>n}\mu_i^{-1}$. Conversely, for
each $j$ with $|j|>n$,
$\|(M-M_n)\be_j\|_1=1/\mu_j$. Hence
$$
\|M-M_n\|
=
\sup_{|i|>n}\frac1{\mu_i}
\to 0.
$$
Since each $M_n$ has finite-dimensional range, it is compact, and
therefore $M$ is compact. Finally, $L=MP$, where $P$ is bounded on
$\BS$, so $L$ is compact.
\end{proof}

\begin{proof}[Proof of Lemma~\ref{lemma:ideal}]
Suppose that the coordinate ideal $I_A$ is invariant under $L$, and
that $A\neq\emptyset$. Choose $j\in A$. Then $\be_j\in I_A$, and hence
$L\be_j\in I_A$. Since
$(L\be_j)_i={p_{i-j}}/{\mu_i}$, $i\in\Z$,
and $\mu_i>0$, it follows that $p_{i-j}>0$ implies $i\in A$.
Thus, $A$ contains every trait that can be reached from $j$ in one
mutation step.
Applying the same argument successively shows that $A$ contains every
trait that can be reached from $j$ by a finite sequence of mutation
steps. Since $\bp$ is irreducible, every trait in~$\Z$ can be reached
from~$j$ in this way. Hence $A=\Z$. Therefore, the only coordinate
ideals invariant under $L$ are $I_\emptyset=\{\bzero\}$ and
$I_\Z=\BS$, and so $L$ is ideal-irreducible.
\end{proof}

\begin{proof}[Proof of Theorem~\ref{critical}]
Let $\bx^*$ be a non-zero equilibrium, and set
$m^*=\sum_i x_i^*$. Then
$$
L\bx^*
=
\frac{1}{\lambda(1-m^*)}\bx^*.
$$
Since $\bx^*\in K\setminus\{\bzero\}$,
Theorem~\ref{CWW}(iii) implies that
$$
\frac{1}{\lambda(1-m^*)}=r(L).
$$
Hence $m^*=1-1/\{\lambda r(L)\}$. Since $m^*>0$, it follows that
$\lambda r(L)>1$. This proves~(i).

Now suppose that $\lambda r(L)>1$. Let
$\bphi\in\qit(K)$ be the eigenvector corresponding to $r(L)$ whose
existence is guaranteed by Theorem~\ref{CWW}(ii). As observed above,
$\bphi$ gives rise to the equilibrium
$$
\bx^*
=
\left(1-\frac{1}{\lambda r(L)}\right)
\frac{\bphi}{\sum_i\phi_i}.
$$
Its total mass is $1-1/\{\lambda r(L)\}\in(0,1)$, so
$\bx^*\in E$. Since $\bphi\in\qit(K)$, this equilibrium is strictly
positive.

To prove uniqueness, let $\by^*$ be any non-zero equilibrium. Then
$$
L\by^*
=
\frac{1}{\lambda(1-m(\by^*))}\by^*.
$$
Since $\by^*\in K\setminus\{\bzero\}$,
Theorem~\ref{CWW}(iii) gives
$1/\{\lambda(1-m(\by^*))\}=r(L)$, and hence
$m(\by^*)=1-1/\{\lambda r(L)\}$. Moreover, $\by^*$ is an eigenvector
of $L$ corresponding to $r(L)$. Since $r(L)$ is algebraically simple,
its eigenspace is one-dimensional, so $\by^*$ is a scalar multiple of
$\bphi$. The condition
$m(\by^*)=1-1/\{\lambda r(L)\}$ uniquely determines this scalar.
Therefore $\by^*=\bx^*$.
\end{proof}

\begin{proof}[Proof of Proposition~\ref{truncation}]
By Lemma~\ref{lemma:compact}, $L$ is compact on $\BS$. Since
$$
L-L_I
=
(\Id-Q_I)L
+
Q_IL(\Id-Q_I),
$$
and $\|Q_I\|=1$, it is enough to show that both
$\|(\Id-Q_I)L\|$ and
$\|L(\Id-Q_I)\|$ tend to $0$.

Recall that $L=MP$, where $(M\bx)_i=x_i/\mu_i$. For $\bx\in\BS$,
$$
\|(\Id-Q_I)M\bx\|_1
\leq
\left(\sup_{|i|>I}\frac1{\mu_i}\right)\|\bx\|_1.
$$
Consequently,
$\|(\Id-Q_I)L\|
\leq(\sup_{|i|>I}\mu_i^{-1})\|P\|\to 0$.

We next estimate the second term. Since
$L(\Id-Q_I)$ is positive, the column-sum formula
for positive operators on $\ell_1$ gives
$$
\|L(\Id-Q_I)\|
=
\sup_{|k|>I}\|L\be_k\|_1
=
\sup_{|k|>I}\sum_{j\in\Z}\frac{p_j}{\mu_{j+k}}.
$$
Let $\varepsilon>0$, and
choose $J$ sufficiently large that
$\sum_{|j|>J}p_j<\muinf\varepsilon/2$. Since
$\mu_{j+k}\geq\muinf$, the contribution from $|j|>J$ is, uniformly in
$k$, bounded above by
$$
\sum_{|j|>J}\frac{p_j}{\mu_{j+k}}
\leq
\frac{1}{\muinf}\sum_{|j|>J}p_j
<
\frac{\varepsilon}{2}.
$$
If $|k|>I$ and $|j|\leq J$, then
$|j+k|>I-J$, and hence the contribution from $|j|\leq J$ is
at most
$$
\left(\sup_{|i|>I-J}\frac1{\mu_i}\right)
\sum_{|j|\leq J}p_j
\leq
\sup_{|i|>I-J}\frac1{\mu_i}.
$$
This is less than $\varepsilon/2$ for all sufficiently large
$I$, because $\mu_i\to\infty$ as $|i|\to\infty$.
Therefore
$\|L(\Id-Q_I)\|\to 0$, and combining the two
estimates yields $\|L_I-L\|\to 0$.

By Theorem~\ref{CWW}(i)--(ii), $r(L)>0$ is an algebraically simple
eigenvalue of $L$. Since $L$ is compact, $r(L)$ is an isolated point of
$\sigma(L)$. Choose a positively oriented circle $\Gamma$ contained in
the resolvent set of $L$ and enclosing $r(L)$ but no other point of
$\sigma(L)$. Since $L_I\to L$ in operator norm, the stability theorem
for separated parts of the spectrum implies that, for all sufficiently
large $I$, $\Gamma$ also separates a part of $\sigma(L_I)$ whose total
algebraic multiplicity is equal to that of $r(L)$, namely one. The
continuity of a finite system of eigenvalues then implies that this part
of $\sigma(L_I)$ consists of a single eigenvalue $\eta_I$ satisfying
$\eta_I\to r(L)$. See Theorem~3.16 and Sections~3.4--3.5 of Chapter~IV
of~\cite{Kato95}. Since $|\eta_I|\leq r(L_I)$, it follows that
$$
\liminf_{I\to\infty}r(L_I) \geq \lim_{I\to\infty}|\eta_I| = r(L).
$$
For the reverse inequality, observe that $0\leq L_I\leq L$. Indeed, if
$\bx\geq0$, then $Q_I\bx\leq\bx$, and positivity of $L$ gives
$LQ_I\bx\leq L\bx$. Since $Q_I\by\leq\by$ for every $\by\geq0$, we
therefore have
$$
L_I\bx = Q_ILQ_I\bx \leq LQ_I\bx \leq L\bx.
$$
Since both $L_I$ and $L$ are positive, it follows inductively that
$0\leq L_I^n\leq L^n$, and hence $\|L_I^n\|\leq\|L^n\|$ for every
$n\geq1$. The spectral-radius formula therefore gives $r(L_I)\leq r(L)$.
Combining the two inequalities yields $r(L_I)\to r(L)$.
\end{proof}


\begin{thebibliography}{99}

\bibitem{Anselone71}
P.~M. Anselone.
\newblock {\em Collectively Compact Operator Approximation Theory and
  Applications to Integral Equations}.
\newblock Prentice-Hall, Englewood Cliffs, NJ, 1971.

\bibitem{BL12b}
A.D. Barbour and M.J. Luczak.
\newblock A law of large numbers approximation for {M}arkov population
  processes with countably many types.
\newblock {\em Probab. Theory Related Fields}, 153:727--757, 2012.

\bibitem{Burger2003}
Reinhard B{\"u}rger.
\newblock {\em The Mathematical Theory of Selection, Recombination, and
  Mutation}.
\newblock Wiley, Chichester, 2000.

\bibitem{ChampagnatFerriereMeleard}
Nicolas Champagnat, R\'egis Ferri\`ere, and Sylvie M\'el\'eard.
\newblock Unifying evolutionary dynamics: from individual stochastic processes
  to macroscopic models.
\newblock {\em Theoretical Population Biology}, 69:297--321, 2006.

\bibitem{CWW20}
Kung-Ching Chang, Xuefeng Wang, and Xie Wu.
\newblock On the spectral theory of positive operators and {PDE} applications.
\newblock {\em Discrete Contin. Dyn. Syst.}, 40(6):3171--3200, 2020.

\bibitem{Christiansen}
Freddy~B. Christiansen.
\newblock Density-dependent selection.
\newblock In Rama~S. Singh and Marcy~K. Uyenoyama, editors, {\em The Evolution
  of Population Biology}, pages 139--155. Cambridge University Press,
  Cambridge, 2004.

\bibitem{CrowKimura1970}
J.F. Crow and M.~Kimura.
\newblock {\em An Introduction to Population Genetics Theory}.
\newblock Harper and Row, New York, 1970.

\bibitem{Dei85}
Klaus Deimling.
\newblock {\em Nonlinear Functional Analysis}, volume 113 of {\em Graduate
  Texts in Mathematics}.
\newblock Springer-Verlag, Berlin, Heidelberg, 1985.

\bibitem{Diekmann2021}
O.~Diekmann, J.~A.~P. Heesterbeek, and J.~A.~J. Metz.
\newblock {\em The Legacy of Kermack and McKendrick}.
\newblock SIAM, Philadelphia, 2021.

\bibitem{Ewens2004}
Warren~J. Ewens.
\newblock {\em Mathematical Population Genetics I: Theoretical Introduction},
  volume~27 of {\em Interdisciplinary Applied Mathematics}.
\newblock Springer, 2nd edition, 2004.

\bibitem{Fat99}
H.O. Fattorini.
\newblock {\em Infinite Dimensional Differential Equations in Banach Spaces}.
\newblock North-Holland, 1999.

\bibitem{IK06}
S.~Inusah and T.J. Kozubowski.
\newblock A discrete analogue of the {L}aplace distribution.
\newblock {\em J. Statist. Plann. Inference}, 136:1090--1102, 2006.

\bibitem{Kato95}
Tosio Kato.
\newblock {\em Perturbation Theory for Linear Operators}.
\newblock Classics in Mathematics. Springer-Verlag, Berlin, 2nd edition, 1995.

\bibitem{Kur70}
T.G. Kurtz.
\newblock Solutions of ordinary differential equations as limits of pure jump
  {M}arkov processes.
\newblock {\em J. Appl. Probab.}, 7:49--58, 1970.

\bibitem{Kur71}
T.G. Kurtz.
\newblock Limit theorems for sequences of jump {M}arkov processes approximating
  ordinary differential processes.
\newblock {\em J. Appl. Probab.}, 8:344--356, 1971.

\bibitem{Lande2009}
Russell Lande, Steinar Engen, and Bernt-Erik Saether.
\newblock An evolutionary maximum principle for density-dependent population
  dynamics in a fluctuating environment.
\newblock {\em Philosophical Transactions of the Royal Society B},
  364:1511--1518, 2009.

\bibitem{LawlerLimic2010}
Gregory~F. Lawler and Vlada Limic.
\newblock {\em Random Walk: A Modern Introduction}, volume 123 of {\em
  Cambridge Studies in Advanced Mathematics}.
\newblock Cambridge University Press, Cambridge, 2010.

\bibitem{PN2002}
K.M. Page and M.A. Nowak.
\newblock Unifying evolutionary dynamics.
\newblock {\em J. Theor. Biol.}, 219:93--98, 2002.

\bibitem{Parsons2010}
Todd~L. Parsons, Christopher Quince, and Joshua~B. Plotkin.
\newblock Some consequences of demographic stochasticity in population
  genetics.
\newblock {\em Genetics}, 185:1345--1354, 2010.

\bibitem{Paz83}
Amnon Pazy.
\newblock {\em Semigroups of Linear Operators and Applications to Partial
  Differential Equations}, volume~44 of {\em Applied Mathematical Sciences}.
\newblock Springer-Verlag, New York, 1983.

\bibitem{Pol24}
P.K. Pollett.
\newblock An {SIS} epidemic model with individual variation.
\newblock {\em Math. Biosci. Eng.}, 21:5446--5455, 2024.

\bibitem{Smith1993}
T.B. Smith.
\newblock Disruptive selection and the genetic basis of bill size polymorphism
  in the {A}frican finch {Pyrenestes}.
\newblock {\em Nature}, 363:618--620, 1993.

\bibitem{TravisEtAl2023}
Joseph Travis, Ronald~D. Bassar, Tim Coulson, David Reznick, and Matthew Walsh.
\newblock Density-dependent selection.
\newblock {\em Annual Review of Ecology, Evolution, and Systematics},
  54:85--105, 2023.

\bibitem{Ver1838}
P.F. Verhulst.
\newblock Notice sur la loi que la population suit dans son accroissement.
\newblock {\em Corr. Math. et Phys.}, X:113--121, 1838.

\bibitem{Webb1985}
Glenn~F. Webb.
\newblock {\em Theory of Nonlinear Age-Dependent Population Dynamics}.
\newblock Marcel Dekker, New York, 1985.

\bibitem{WD71}
G.H. Weiss and M.~Dishon.
\newblock On the asymptotic behaviour of the stochastic and deterministic
  models of an epidemic.
\newblock {\em Math. Biosci.}, 11:261--265, 1971.

\end{thebibliography}

\end{document}